\documentclass[a4paper,onecolumn,11pt,unpublished]{quantumarticle}
\pdfoutput=1

\usepackage{aux_tikz}

\begin{document}

\renewcommand{\ref}[1]{{\color{red}[use \texttt{cref} not \texttt{ref}]}}

\title{Cyclic quantum causal modelling with a graph separation theorem}

\author{Carla Ferradini$^1$}
\author{Victor Gitton$^1$}
\author{V.\ Vilasini$^{1,2}$}
\affiliation{$^1$Institute for Theoretical Physics, ETH Zurich, 8093 Zürich, Switzerland\\ $^2$Université Grenoble Alpes, Inria, 38000 Grenoble, France}
\maketitle

\begin{abstract}
\small
  Causal modelling frameworks link observable correlations to causal explanations, which is a crucial aspect of science. These models represent causal relationships through directed graphs, with vertices and edges denoting systems and transformations within a theory. Most studies focus on acyclic causal graphs, where well-defined probability rules and powerful graph-theoretic properties like the $d$-separation theorem apply. However, understanding complex feedback processes and exotic fundamental scenarios with causal loops requires cyclic causal models, where such results do not generally hold. While progress has been made in classical cyclic causal models, challenges remain in uniquely fixing probability distributions and identifying graph-separation properties applicable in general cyclic models. In cyclic quantum scenarios, existing frameworks have focussed on a subset of possible cyclic causal scenarios, with graph-separation properties yet unexplored. This work proposes a framework applicable to all consistent quantum and classical cyclic causal models on finite-dimensional systems. We address these challenges by introducing a robust probability rule and a novel graph-separation property, $p$-separation, which we prove to be sound and complete for all such models. Our approach maps cyclic causal models to acyclic ones with post-selection, leveraging the post-selected quantum teleportation protocol. We characterize these protocols and their success probabilities along the way. We also establish connections between this formalism and other classical and quantum frameworks to inform a more unified perspective on causality. This provides a foundation for more general cyclic causal discovery algorithms and to systematically extend open problems and techniques from acyclic informational networks (e.g., certification of non-classicality) to cyclic causal structures and networks.
\end{abstract}

\newpage
\tableofcontents
\newpage

\section{Introduction}
\label{sec: intro}
Bridging causal explanations with observable correlations lies at the heart of scientific inquiry. Causal models provide a rigorous framework for describing correlations arising from a causal structure and identifying which structures are compatible with observed data. Classical causal modelling~\cite{Pearl_2009, Spirtes1993}, captures causal relationships among random variables and has been widely applied across fields such as machine learning, economics, and clinical trials~\cite{Raita_2021,Kleinberg_2011,Pearl_2009,Spirtes_2005,Petersen_2014,Arti_2020,Liu_2021}. However, as demonstrated by Bell’s theorem~\cite{Bell_1964}, this classical framework cannot account for quantum correlations without invoking fine-tuned mechanisms or modifications to the causal structure naturally associated with a Bell scenario~\cite{Wood_2015}. This limitation has spurred the development of non-classical causal modelling frameworks that encompass quantum and broader operational theories~\cite{Henson_2014, Barrett_2019}, enabling the causal explanation of quantum correlations without invoking fine-tuning or adjustments to the operational causal structure.

Causal models are typically represented as directed graphs, with vertices and edges being associated with systems and transformations within a theory that embody its causal mechanisms. The mechanisms themselves are theory-dependent: functional dependencies in classical models, and quantum channels in quantum models. Broadly, causal modelling is an umbrella term that spans different approaches, in particular (i) approaches which focus on observable correlations generated by causal mechanisms of a theory, which are often referred to as Bayesian/causal networks (e.g., \cite{Geiger1990, Henson_2014}) and (ii) approaches such as structural equation models (e.g., \cite{Forre_2017, Bongers_2021}) or split-node quantum causal models (e.g., \cite{Barrett_2019, Barrett_2021}), which examine properties of the causal mechanisms directly, such as solvability of functions or causal influences between quantum systems in unitary channels. In this work, we use “causal models” to encompass both approaches, distinguishing them where necessary.

The classical and non-classical causal modelling literature have predominantly focused on \emph{acyclic} graphs, where there exists a well-defined probability rule for deriving correlations from causal mechanisms~\cite{Pearl_2009,Henson_2014, Barrett_2019}. Moreover, a foundational and powerful result in acyclic causal models (both classical and non-classical) is the $d$-separation theorem~\cite{Verma1990, Geiger1990, Henson_2014}, proving the soundness and completeness of a central graph-theoretic notion, $d$-separation~\cite{Pearl_2009, Spirtes_2005}. This theorem enables to read off conditional independences in correlations purely from the structure of the graph (causal structure), and is central to how we explain correlations (conditional dependencies) in terms of causal connections. Moreover, the $d$-separation theorem plays an integral role in causal discovery algorithms and inference across data-driven disciplines (see \cite{Spirtes2016}), in causal compatibility problems as well as certification of non-classical correlations in causal structures \cite{Wood_2015, Henson_2014}.

Recent years have witnessed growing interest in causal modelling for \emph{cyclic} graphs, which can represent physical processes with feedback~\cite{Forre_2017, Bongers_2021} as well as provide information-theoretic models for describing exotic and fundamental causal phenomena, such as solutions to general relativity admitting closed timelike curves~\cite{Godel1949, vanStockum_1938,Matzner_1967, Deutsch1991, Lloyd_2011, Lloyd_2011_2}. However, cyclic models introduce challenges, including potential pathologies like the “grandfather paradox”. Two significant open problems arise in this context. First, for general cyclic causal models both in the classical and quantum case\footnote{For example, in the classical literature, the observed distribution is only considered uniquely defined by the causal mechanisms (functional dependencies) in models that admit unique solutions to the functional dependencies \cite{Forre_2017, forre_2018}, but this is not the case for non-uniquely solvable classical causal models, see also \cite{Sister_paper} for details.}, no well-defined and robust method to compute observable probabilities from the causal mechanisms exists.
Secondly, certain cyclic causal models, even in classical settings and where the probabilities are well-defined, violate the soundness of $d$-separation, i.e., there exist classical causal models on cyclic graphs leading to correlations between vertices that are $d$-separated in the graph \cite{Neal_2000}. This complicates causal reasoning, raising concerns whether observations can be causally explained in a systematic manner in scenarios with cyclic causal structures. 

 Existing frameworks for cyclic causal modelling often impose restrictions on causal mechanisms to ensure consistency and define probabilities. For instance, classical cyclic models require global consistency conditions (e.g., unique solvability of functional dependencies~\cite{Forre_2017, Bongers_2021}), while quantum models enforce conditions like factorization of unitary channels (corresponding to requiring valid process operators~\cite{Barrett_2021}). Although these approaches capture meaningful subclasses of cyclic causal models and provide valuable techniques for studying them, the soundness of $d$-separation already fails within such classes of models (in particular, even in uniquely solvable models as shown in \cite{Neal_2000}). Moreover, there exist more general types of consistent cyclic causal models which cannot be captured in these approaches and constructing robust probability rules for arbitrary cyclic causal models has remained unresolved.

A general framework for cyclic causal modelling, applicable to non-classical theories, was proposed in \cite{VilasiniColbeckPRA, VilasiniColbeckPRL}. This top-down approach accommodates any theory, causal mechanism, and probability rule on a given cyclic graph, provided that the resulting probability distribution adheres to the soundness of $d$-separation (i.e., $d$-separations in the graph imply corresponding conditional independence in the distribution). Specifically, \cite{VilasiniColbeckPRA} constructed examples of non-classical cyclic causal models respecting $d$-separation soundness and \cite{VilasiniColbeckPRL} introduced a new method for computing probabilities for a specific classical cyclic model that is not uniquely solvable. However, models violating the soundness of $d$-separation fall outside this framework, and the general applicability of the probability method, beyond the particular examples, was not explored.

\begin{sloppypar}
Finding alternative graph-separation properties to $d$-separation that are sound and complete in the cyclic case is an active research area in classical causality community~\cite{Bongers_2021}. Notably, the concept of $\sigma$-separation was proposed as a sound and complete criterion for a specific subclass of classical models admitting unique solutions~\cite{Forre_2017}. However, to our knowledge, no graph-separation property applicable to all models is known even in the classical case. The problem remains unexplored for quantum and non-classical cyclic models, to the best of our knowledge.
\end{sloppypar}

\paragraph{Contributions.} In this work, we present a general framework for cyclic causal modelling involving finite-dimensional quantum systems. This framework extends beyond existing approaches to include cyclic causal models which are not necessarily associated with valid quantum process operators \cite{Barrett_2019, Barrett_2021} and it embeds classical causal models, including not uniquely solvable ones. We define a robust probability rule for all models in this framework, recovering existing results under the relevant restrictions. A central contribution is the introduction of a novel graph-separation property, $p$-separation, which we prove to be sound and complete for all quantum cyclic causal models within our framework, and which reduces to $d$-separation for directed acyclic graphs. 

Our causal modelling approach aligns more closely with Bayesian networks in terms of its primary focus on questions related to observable correlations and their relationships to underlying graph-separations. However, we also discuss connections to other causality frameworks, emphasizing the broader relevance of our results. Our methods are based on a general mapping of cyclic models to acyclic ones using a post-selected quantum teleportation protocol. This entails new results on quantum teleportation \cite{Bennett1993}, which might be of independent interest to the quantum information community. Thus, our framework equips post-selected closed timelike curves (P-CTCs) \cite{Lloyd_2011, Lloyd_2011_2}, which are a quantum informational model for causal loops motivated by exotic solutions of general relativity, with a causal modelling and graph-separation semantic.

Furthermore, in a companion paper \cite{Sister_paper}, we develop the framework and results for the classical causal modelling community by introducing the concept of classical post-selected teleportation. This includes an alternative formulation of $p$-separation for classical functional causal models which we prove to be equivalent to the quantum version presented herer when restricted to classical scenarios. In~\cite{Sister_paper}, we establish connections between unique solvability properties and post-selection success probabilities in classical models. This enables our results to be directly accessible and applicable within the classical causality community, without requiring familiarity with quantum formalism.

\subsection{Structure of the paper}
\label{sec: structure}
In~\cref{sec: overview}, we provide a high-level overview of the framework, the new graph-separation property of $p$-separation and its soundness and completeness, offering intuition through illustrative examples, without delving into technical details or proofs.  
\begin{itemize}
    \item \textbf{Framework and probability rule:}  
In~\cref{sec: framework}, we introduce a general framework for cyclic quantum causal modelling, starting with the known probability rule for the acyclic case illustrated through examples. In~\cref{sec: cyclic_to_acyc_v3}, we define a probability rule for general models on cyclic graphs, which is achieved by constructing a  mapping from a given (possibly cyclic) causal model to a family of acyclic models with post-selection. A key ingredient, post-selected quantum teleportation protocols, is described in~\cref{sec:post-selected teleportation}. Using this,~\cref{sec:family of acyclic CM} constructs the relevant family of acyclic models, and~\cref{sec:probability rule CM} formalizes the probability rule for cyclic models. By establishing relevant results on post-selected quantum teleportation in \cref{app:ps teleportation}, we prove the robustness of our probability rule, namely, that it is independent of the choice of teleportation protocol and of the particular acyclic model in the family. Finally,~\cref{sec:examples cyclic causal graphs} illustrates these methods with examples of cyclic causal models.  
 \item \textbf{Soundness and completeness of $p$-separation:}  
In~\cref{sec: introducing pseparation}, we discuss graph separation properties, beginning with a review of $d$-separation in~\cref{sec: rev dseparation} and its failure for cyclic graphs in~\cref{sec:failure_dsep}. In~\cref{sec:psep_introduced}, we introduce $p$-separation, the new graph-separation property generalizing $d$-separation, proving its soundness and completeness for all causal models within our framework. Examples of $p$-separation are given in~\cref{sec:examples psep}, including a discussion on how $p$-separation reduces to $d$-separation in the acyclic case. 
\item \textbf{Further results:} In~\cref{sec: quantum_Markov}, we generalize the concept of probabilistic Markovianity to cyclic graphs and derive conditions for this property in terms of post-selection success probabilities. In~\cref{sec: frameworks} (and further in \cref{sec:quantum aspects}), we situate our work within broader causality frameworks (summarised in \cref{fig:frameworks}) and discuss its connections with the causal modelling approaches of  Barrett-Lorenz-Oreshkov (BLO) for cyclic models  and Costa-Shrapnel (CS) for acyclic models, to the causal decomposition problem, as well as links to closed time-like curves and indefinite causal order quantum processes. In particular, in~\cref{sec:quantum aspects} we show that a class of BLO causal models, which we call tensor-restricted BLO models, can be faithfully mapped into our formalism, which implies the same for all CS causal models. In~\cref{sec:classical aspects}, we focus on classical functional models, showing how to faithfully map them to our causal models and proving consistency with the classical formulation of the cyclic probability rule and $p$-separation given in~\cite{Sister_paper}.
\end{itemize}
 
Finally,~\cref{sec: conclusions} summarizes the main contributions and discusses directions for future research.

\subsection{Notation}

We denote with $\hilmap$ a finite-dimensional Hilbert space, i.e., a finite-dimensional complex vector space equipped with an inner product. 
We denote with $\linops(\hilmap)$ the complex vector space of linear operators acting on $\hilmap$. 
The trace is denoted as $\Tr[\cdot]: \linops(\hilmap) \mapsto \mathbb{C}$.
The identity map on $\hilmap$ is denoted $\id \in \linops(\hilmap)$. 

We use the following acronyms: CP map stands for completely positive map, CPTP map stands for completely positive and trace preserving map, and POVM stands for positive operator valued measurement.

We associate to a finite set $\outcomemap$ a finite dimensional Hilbert space $\hilmapsetarg{\outcomemap}$ such that $\dim(\hilmapsetarg{\outcomemap})=|\outcomemap|$. 
The Hilbert space $\hilmapsetarg{\outcomemap}$ is equipped with a preferred basis, labelled $\{\ket{\outcome}\}_{\outcome\in\outcomemap}$ and referred to as the computational basis, so that
\begin{equation}
\label{eq:Hxspace}
    \hilmapsetarg{\outcomemap} = \text{\normalfont{span}}\biglset{ \ket{\outcome} }{ \outcome\in\outcomemap }.
\end{equation}

We denote with $\graphname = \graphexpl$ a directed graph where each edge $\edgename = \edgearg{\vertname_1}{\vertname_2} \in \edgeset$ is an ordered pair of two vertices $\vertname_1,\vertname_2\in\vertset$. We shall always assume, for convenience, that the set of vertices $\vertset$ is equipped with a preferred order, such that we can write $\vertset = \{\vertname_1,\dots,\vertname_\vertcount\}$ where $\vertcount = |\vertset|$.
The incoming and outgoing edges to a vertex $\vertname\in\vertset$ are denoted
\begin{subequations}
\begin{align}
    \inedges{\vertname} &= \biglset
    { e\in\edgeset }{ \exists \vertname'\in\vertset \st e = \edgearg{\vertname'}{\vertname} }, \\
    \outedges{\vertname} &= \biglset
    { e\in\edgeset }{ \exists \vertname'\in\vertset \st e = \edgearg{\vertname}{\vertname'} },
\end{align}
\end{subequations}
while the parents and children of a vertex $\vertname\in\vertset$ are denoted
\begin{subequations}
\begin{align}
    \parnodes{\vertname} &= \biglset{ \vertname' \in \vertset }{ (\vertname',\vertname) \in \edgeset}, \\
    \childnodes{\vertname} &= \biglset{ \vertname' \in \vertset }{ (\vertname,\vertname') \in \edgeset }.
\end{align}
\end{subequations}
Given a $\graphname=\graphexpl$, we say that a vertex $\vertname\in\vertset$ is \hypertarget{exogenous}{exogenous} if the set $\parnodes{\vertname}$ is empty, \hypertarget{endogenous}{endogenous} otherwise and we define the sets
\begin{equation}
    \exnodes = \{\vertname\in\vertset \; : \; \parnodes{\vertname}=\emptyset\} \quad \text{and} \quad  \nexnodes = \vertset \setminus \exnodes,
\end{equation}
of exogenous and endogenous vertices of a graph.

\section{Overview of the framework and main results through examples}
\label{sec: overview}
In this section, we present the results in an informal manner with the scope of providing an intuitive picture of their use. 
Complete definitions and proofs are provided in the following sections and referenced here.

\paragraph{Causal graphs and causal models.}
We present our framework using the following directed graph as example
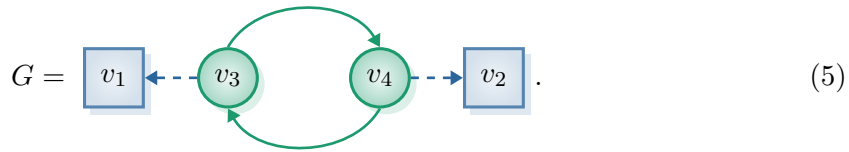
\begin{equation}
\label{eq:example_cycle_overview}
    \graphname=\centertikz{
        \node[unode,minimum size =22pt] (v1) at (-1, 0) {$v_3$};
        \node[unode, minimum size =22pt] (v2) at (1, 0) {$v_4$};
        \node[onode, minimum size =22pt] (v4) at (2.5, 0) {$v_2$};
        \node[onode, minimum size =22pt] (v3) at (-2.5, 0) {$v_1$};
        \draw[qleg] (v1.north) to[in = 120, out = 60] (v2.north);
        \draw[qleg] (v2.south) to[in = 300, out = 240] (v1.south);
        \draw[cleg] (v2) -- (v4);
        \draw[cleg] (v1) -- (v3);
    }.
\end{equation}
The notation allows to distinguish between observed and unobserved processes and classical or quantum systems. Directed graphs that are decorated as follows are called \textit{causal graphs} (\cref{def: causal graph}). Specifically,
\begin{myitem}
    \item \textit{quantum} edges are represented as $\centertikz{
        \draw[qleg] (0,0) -- (0.6,0);
        }$ and carry a finite-dimensional Hilbert space $\hilmap$;
    \item \textit{classical} edges are represented as $\centertikz{
        \draw[cleg] (0,0) -- (0.6,0);
        }$ and carry a finite set $\outcomemap$;
    \item \textit{unobserved} vertices are represented as circles $\centertikz{
            \node[unode] (q) at (0,0) {$\vertname$};
        }$ and carry a CPTP map, $\chanmap$;
    \item \textit{observed} vertices are represented as a rectangles $\centertikz{
            \node[onode] (q) at (0,0) {$\vertname$};
            }$ and carry a finite-cardinality random variable taking values from a finite set $\outcomemap$ and a POVM, $\povm = \{\povmel{\outcome}\}_{\outcome\in\outcomemap}$.
\end{myitem}
Observed vertices correspond to the operation of performing a measurement, described by the POVM, $\povm$, on the incoming edges. The outcome of the measurement, $\outcome\in\outcomemap$, is carried by the vertex itself and broadcast to its children vertices.

Given the causal graph of~\cref{eq:example_cycle_overview}, one can define a causal model (\cref{def:causal model}) by associating finite-dimensional Hilbert spaces, $\hilmap_{(3,4)}$ and $\hilmap_{(4,3)}$ respectively to the edges $(v_3,v_4)$ and  $(v_4,v_3)$, and finite sets to the edges $(v_3,v_1)$ and $(v_4,v_2)$, respectively $\outcomemaparg{(3,1)}$ and $\outcomemaparg{(4,2)}$. The finite sets are associated with finite dimensional Hilbert spaces as in~\cref{eq:Hxspace}, namely 
\begin{equation}
    \hilmapsetarg{\outcomemaparg{(3,1)}} = \textup{span}\{\ket{x}\}_{x\in\outcomemaparg{(3,1)}} \textup{ and } \hilmapsetarg{\outcomemaparg{(4,2)}} = \textup{span}\{\ket{x}\}_{x\in\outcomemaparg{(4,2)}}
\end{equation}
Then, we associate CPTP maps, 
\begin{equation*}
\chanmap_{3}:\linops\left(\hilmap_{(4,3)}\right)\mapsto \linops\left(\hilmap_{(3,4)}\otimes \hilmapsetarg{\outcomemaparg{(3,1)}}\right) \textup{ and } \chanmap_{4}:\linops\left(\hilmap_{(3,4)}\right)\mapsto \linops\left(\hilmap_{(4,3)}\otimes \hilmapsetarg{\outcomemaparg{(4,2)}}\right)
\end{equation*}
respectively to the vertices $v_3$ and $v_4$. To the observed vertices, $\vertname_1$ and $\vertname_2$ we associate the finite sets $\outcomemaparg{1}$ and $\outcomemaparg{2}$ and POVMs 
\begin{equation}
    \povmarg{(1)}=\left\{\povmelarg{\outcome_1}{(1)}\in\linops\left(\hilmapsetarg{\outcomemaparg{(3,1)}}\right)\right\}_{\outcome_1\in\outcomemaparg{1}} \textup{ and } \povmarg{(2)}=\left\{\povmelarg{\outcome_2}{(2)}\in\linops\left(\hilmapsetarg{\outcomemaparg{(4,2)}}\right)\right\}_{\outcome_2\in\outcomemaparg{2}}.
\end{equation} 

In the special case where the causal graph is \textit{acyclic}, the probability distribution over observed vertices is defined as is standard in the literature (\cref{def: acyclic probability})~\cite{Henson_2014}. Concretely, one first takes the tensor product of the maps associated to each exogenous vertex with input the only trace-one and positive linear operator in the trivial input space $\mathbb{C}$, i.e., $1$. Then, one applies the composition of the CPTP or CP maps associated to each endogenous vertex with the composition order given by the directed edges of the graph $\graphname$. Eventually, because the graph is finite and acyclic (where the childless vertices are associated with maps having a trivial output space), and because all maps are completely positive\footnote{Complete positivity implies that the outcome that resulting from composing all these maps is still a positive operator. Since the outcome is a value in $\mathbb{C}$, it implies such value is real and positive. In addition, one can easily see that the maps we associate are also trace non-increasing, thus the final output is smaller than $1$.}, the output of this composition operation will be a real number between $0$ and $1$. 

\paragraph{Mapping cyclic causal models to acyclic with post-selection and the probability rule.}
Firstly, given a well-defined causal model on a causal graph, we aim to define a rule for computing a single, well-defined probability distribution over the observed vertices' random variables. We achieve this by constructing a mapping from cyclic causal models to acyclic ones with post-selection (\cref{sec: cyclic_to_acyc_v3}) using a post-selected teleportation protocol (from the general class of such protocols as per \cref{def:ps teleportation,def:bell tele}) to replace directed edges. Specifically, since directed edges act as identity channels which connect the maps associated to the vertices of the causal model and the post-selected teleportation protocol allows us to simulate identity channels, the latter can be used to replace a directed edge. Let us clarify this statement with an example.

In a typical teleportation protocol~\cite{Bennett1993}, Alice and Bob share a bipartite entangled state, whose preparation is represented by $\prevertname$ and the two subsystems by the edges $(\prevertname,Y)$ and $(\prevertname,\postvertname)$. Alice performs a Bell-state measurement, represented by $\postvertname$, on her half of the entangled state, i.e., the edge $(\prevertname,\postvertname)$, and the state to be teleported, i.e., the edge $(X,\postvertname)$. If Alice's measurement results in a specific outcome, $p=\bellstate$, Bob's system is automatically in the correct quantum state, and Bob does not need to perform any further operations.
Thus, by post-selecting on the successful instances where Alice measures $p=\bellstate$, such a protocol simulates an identity channel and it can be used to simulate a direct edge between $X$ and $Y$, i.e.,
\begin{align}
    \centertikz{
            \begin{scope}[xscale=2.6,yscale=1.3]
                \node (in) at (0,0) {$X$};
                \node (out) at (1.5,1) {$Y$};
                \node[prenode] (pre) at (1,0) {$\prevertname$};
                \node[psnode] (post) at (0.5,1) {$\postvertname$};
                \draw[qleg] (pre) -- node[pos=0.7,right] {\small$\edgearg{\prevertname}{\postvertname}$} (post);
                \draw[qleg] (in) -- node[pos=0.4,right=0pt] {\small$\edgearg{X}{\postvertname}$} (post);
                \draw[qleg] (pre) -- node[pos=0.4,right] {\small$\edgearg{\prevertname}{Y}$} (out);
            \end{scope}
            }
    =
    \centertikz{
    \begin{scope}[xscale=1.2]
        \node (in) at (0,0) {$X$};
        \node (out) at (1.2,1) {$Y$};
        \draw[qleg] (in) -- (out);
    \end{scope}
    }.
\end{align}
With this in mind, a causal model on a cyclic causal graph $\graphname$ is mapped to an acyclic causal model with post-selection through the following steps:
\begin{myitem}
    \item Consider an acyclic subgraph, $\graphname'$, of $\graphname$ obtained only by removing edges, e.g.,
    \begin{equation}
        \graphname'=\centertikz{
        \node[unode,minimum size =22pt] (v1) at (-1, 0) {$v_3$};
        \node[unode, minimum size =22pt] (v2) at (1, 0) {$v_4$};
        \node[onode, minimum size =22pt] (v4) at (2.5, 0) {$v_2$};
        \node[onode, minimum size =22pt] (v3) at (-2.5, 0) {$v_1$};
        \draw[qleg] (v1) -- (v2);
        \draw[cleg] (v2) -- (v4);
        \draw[cleg] (v1) -- (v3);
    }.
    \end{equation}
    \item Construct an acyclic graph, $\graphtele$, by replacing each edge in $\graphname$ that is missing in $\graphname'$ with the edges and vertices of the post-selected teleportation protocol (\cref{def:ps teleportation,def: graph_family_v3}), e.g., since $(\vertname_4,\vertname_3)$ is missing in $\graphname'$ we get 
    \begin{equation}
    \graphtele=
    \centertikz{
        \node[prenode, minimum size =18pt] (pre) at (0,-2.5) {$\prevertname$};
        \node[onode, minimum size =22pt] (v3) at (-1.5, 0.5) {$v_1$};
        \node[unode, minimum size =22pt] (v1) at (-1, -1) {$v_3$};
        \node[unode, minimum size =22pt] (v2) at (1, 1) {$v_4$};
        \node[onode, minimum size =22pt] (v4) at (1.5, 2.5) {$v_2$};
        \node[psnode, minimum size =18pt] (ps) at (0.3,2.5) {$\postvertname$};
        \draw[cleg] (v1.120) -- (v3.270);
        \draw[qleg] (pre.140) -- (v1.300);
        \draw[qleg] (pre.60) -- (ps.240);
        \draw[qleg] (v2.90) -- (ps.300);
        \draw[qleg] (v1.north) to[in=270, out = 60] (v2.240);
        \draw[cleg] (v2.60) -- (v4.270);
    }.
    \end{equation}
    \item Define a causal model on the acyclic graph $\graphtele$ by keeping the same associations of the original causal model to all vertices and edges that are preserved from $\graphname$ in $\graphtele$ and associating a post-selected teleportation protocol on the added vertices and edges  (\cref{def:causal model of graphfamily_v3}), e.g, to $R$ associate the state $\rho_R = \ketbra{\bellstate}$ and to $T$ the binary variable $t\in\{\ok, \fail\}$ given by the measurement $\{\ketbra{\bellstate}, \id-\ketbra{\bellstate}\}$.
    \item Evaluate the probability of the casual model on the acyclic graph $\graphtele$ using the known acyclic probability rule (e.g., \cite{Henson_2014}) and consider the conditioned probability where $t=\ok$, obtaining
    \begin{equation*}
        \probacyc(x_1,x_2, t=\ok)_{\graphtele} \quad \text{and} \quad \probacyc(x_1,x_2| t=\ok)_{\graphtele} = \frac{\probacyc(x_1,x_2, t=\ok)_{\graphtele}}{\probacyc(t=\ok)_{\graphtele}}.
    \end{equation*}
    The conditional probability defines a valid probability distribution over the variables $x_1$ and $x_2$ associated with the observed vertices $v_1$ and $v_2$, unless the probability of successful post-selection, $\probacyc(t=\ok)_{\graphtele}$, vanishes.
    \item If $\probacyc(t=\ok)_{\graphtele}>0$, i.e., the post-selection succeeds with non-zero probability, define the probability distribution over observed vertices in the causal model on the cyclic graph $\graphname$ as the conditional probability distribution $\probacyc(x_1,x_2| t=\ok)_{\graphtele}$, i.e., (\cref{def: probability distribution v3})
    \begin{equation*}
        \prob(x_1,x_2)_{\graphname} := \probacyc(x_1,x_2| t=\ok)_{\graphtele}
    \end{equation*}
    If $\probacyc(t=\ok)_{\graphtele}=0$, we say that the model is inconsistent and the probabilities are undefined. 
\end{myitem}
The construction defines a family of acyclic graphs (\cref{def: graph_family_v3}), $\graphfamily{\graphname}$, that can be obtained from $\graphname$ by performing steps $1$-$2$ starting from different choices of acyclic subgraphs $\graphname'$ of $\graphname$. The probability rule obtained in the last step is independent of which graph in this family is used to define it (\cref{lemma: acyclic_prob_same_v3}).
The choice of causal mechanisms on the pre- and post-selection vertices is also not constrained to the Bell teleportation protocol (\cref{corollary:probs indep of tele implementation v3}). Indeed, any pair of state and measurement that allows to simulate an identity channel can be used to define a causal model in step $3$ and leads to the same probability rule (see also~\cref{app:ps teleportation} for more details on post-selected teleportation protocols).

\paragraph{Cyclic graph separation property: $p$-separation.}
A powerful tool in causal modelling is given by theorems relating graph properties, independent of the model's mechanisms, to conditional independencies of the observed probability that would arise for any choice of mechanisms in a class of theories. Specifically, $d$-separation is a graph-theoretic notion which defines whether two vertices are $d$-separated or $d$-connected conditioned on a third vertex~\cite{Pearl_2009, Spirtes_2005} (\cref{def: d-sep}). For example, in the collider graph \[\centertikz{\node[onode] (v1) at (-1, 0) {$A$};
        \node[onode] (v2) at (0.5, 0) {$C$};
        \node[onode] (v4) at (2, 0) {$B$};
        \draw[cleg] (v1) -- (v2);
        \draw[cleg] (v4) -- (v2);
        },\] 
$A$ and $B$ are $d$-separated and become $d$-connected conditioned on $C$, i.e., once we post-select on $C$.
In the special case of a causal model on an acyclic graph $\graphname$, the $d$-separation theorem states that if two vertices, $A$ and $B$, are $d$-separated conditioned on a third, $C$, then probability distribution $\probacyc_{\graphname}$ over the outcomes of these vertices also presents the same conditional independence (\cref{theorem: dsep theorem}). Specifically, if we denote $d$-separation with $\perp^d$ and conditional independence\footnote{Here we write $A\indep B|C$ in $\undefprob_\graphname$ to denote conditional independence (see \cref{def:conditional independence}) i.e., $\undefprob(a,b|c)_\graphname=\undefprob(a|c)_\graphname\undefprob(b|c)_\graphname$ for all values $a,b,c$ of the variables $A,B,C$, where $\undefprob$ denotes an arbitrary probability distribution over the random variables $A,B$ and $C$.} with $\indep$, it holds:
\begin{itemize}
    \item[]\textbf{Soundness:} if $A\perp^dB|C$ in $\graphname$, then \underline{for all} causal models $A\indep B|C$ in $\probacyc_\graphname$,
    \item[]\textbf{Completeness:} If $A\not\perp^dB|C$ in $\graphname$, then \underline{there exists} a causal model on $\graphname$ where $A\not\indep B|C$ in $\probacyc_\graphname$,
\end{itemize}

This $d$-separation theorem was shown in the classical case in \cite{Verma1990, Geiger1990} and in the non-classical case in \cite{Henson_2014}, both for acyclic graphs. One could hope that an analogous theorem holds for cyclic graphs with some definition for probability distributions compatible with that graph.
However, this typically fails \cite{Pearl_2009,Neal_2000} (\cref{sec:failure_dsep}). For example, consider the following graph:
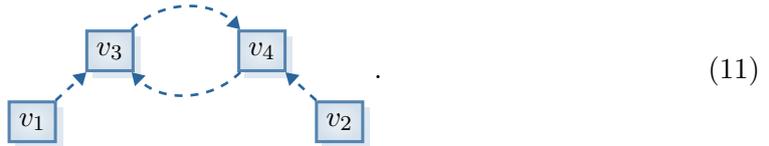
\begin{equation}
\label{eq:graph_failure_dsep}
    \centertikz{
        \node[onode] (v1) at (0,0) {$\vertname_3$};
            \node[onode] (v2) at (2,0) {$\vertname_4$};
            \node[onode] (v3) [below left  =\chanvspace of v1] {$\vertname_1$};
            \node[onode] (v4) [below right  =\chanvspace of v2] {$\vertname_2$};
            \draw[cleg] (v1) to [out=45,in=135]  (v2);
            \draw[cleg] (v2) to [out=225,in=-45] (v1);
            \draw[cleg] (v3) -- (v1);  
            \draw[cleg] (v4) -- (v2);
     }.
\end{equation}
Here, $v_1$ and $v_2$ are $d$-separated, but there exist cyclic causal models on this graph where the two variables must be correlated (\cref{sec:failure_dsep}). Hence, $d$-separation in the graph no longer implies (conditional) independence in the probabilities.
Therefore, there has been interest in finding new graph separation properties and proving their soundness and completeness at least for a subclass of cyclic models.
In the classical causal modelling literature, the notion of $\sigma$-separation is defined as a graph separation property~\cite{Forre_2017}. This was shown to be sound and complete for a subclass of uniquely solvable classical cyclic causal models called modular structural equation models~\cite{Forre_2017, forre_2018}. However, to our knowledge, no general sound and complete graph-separation property is known to hold even for all finite-cardinality classical functional models (including non-uniquely solvable ones).

We propose a new graph-separation property, $p$-separation (\cref{def: p-separation}), based on the correspondence between cyclic causal models and acyclic ones with post-selection. We say that two vertices $A$ and $B$ are $p$-separated conditioned on a third $C$, if there exists an acyclic causal model in the family $\graphfamily{\graphname}$ such that $A$ and $B$ are $d$-separated conditioned on $C$ and on all the post-selection vertices. For the graph in~\cref{eq:graph_failure_dsep}, one can easily see that the following graph, $\graphtele$, is in $\graphfamily{\graphname}$:
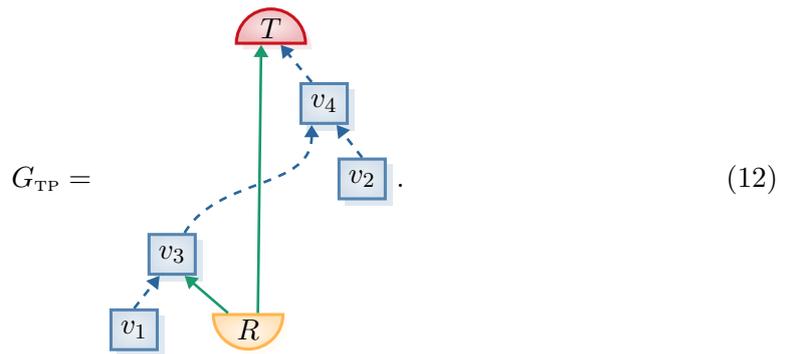
\begin{equation}
    \graphtele=
    \centertikz{
        \node[prenode] (pre) at (0,-2) {$\prevertname$};
        \node[onode] (v3) at (-1.5, -2) {$v_1$};
        \node[onode] (v1) at (-1, -1) {$v_3$};
        \node[onode] (v2) at (1, 1) {$v_4$};
        \node[onode] (v4) at (1.5, 0) {$v_2$};
        \node[psnode] (ps) at (0.3,2) {$\postvertname$};
        \draw[cleg] (v3.90) -- (v1.240);
        \draw[qleg] (pre.140) -- (v1.300);
        \draw[qleg] (pre.60) -- (ps.240);
        \draw[cleg] (v2.120) -- (ps.300);
        \draw[cleg] (v1.60) to[in=270, out = 60] (v2.240);
        \draw[cleg] (v4.90) -- (v2.300);
    }.
\end{equation}
 In $\graphtele$, $\vertname_1$ and $\vertname_2$ are $d$-connected after post-selecting on $\postvertname$. One can show that in all graphs in $\graphfamily{\graphname}$, $\vertname_1$ and $\vertname_2$ become $d$-connected through conditioning on post-selection vertices, thus $\vertname_1$ and $\vertname_2$ are $p$-connected in $\graphname$. 

Similarly to the $d$-separation theorem, we prove the $p$-separation theorem connecting $p$-separation, which is purely a graph property, to conditional independencies of any probability distribution arising from causal models on the graph. Specifically, given a graph $\graphname$, possibly cyclic, and considering the probability rule $\prob_{\graphname}$, defined above (\cref{sec:probability rule CM}), the following hold (\cref{theorem: psep_theorem})
\begin{itemize}
    \item[]\textbf{Soundness:} If $A\perp^pB|C$ in $\graphname$, then \underline{for all} causal models $A\indep B|C$ in $\prob_{\graphname}$,
    \item[]\textbf{Completeness:} If $A\not\perp^pB|C$ in $\graphname$, then \underline{there exists} a causal model on $\graphname$ where $A\not\indep B|C$ in $\prob_{\graphname}$,
\end{itemize}
where we denoted $p$-separation relations with $\perp^p$.
The soundness part of the theorem states that if two vertices are $p$-separated conditioned on a third, then the probability distribution over the outcomes of these vertices also presents the same conditional independence in any (possibly cyclic) quantum or classical causal model involving finite dimensional systems. Completeness ensures that there cannot be a stronger graph-separation property, in the sense of one which entails strictly more graph-separations than $p$-separation, which is also sound.\footnote{If such a property existed, then there would be at least one $p$-connection $A\not\perp^pB|C$ where we would have separation relative to the new property, and if that property were sound, $A\indep B|C$ in all causal models. This is disallowed by completeness.} In particular, this explains by means of a sound and complete graph-separation criterion why the variables associated with the vertices $\vertname_1$ and $\vertname_2$ can become correlated in certain causal models on the cyclic graph in our above example. 

\section{Causal modelling framework}
\label{sec: framework}

In this section, we introduce our causal modelling framework which is based on a graphical representation of cause and effect relations though so-called \textit{causal graphs}. 
We will then describe how to equip causal graphs with a causal model, and define a probability distribution over the observed variables in the special case of \textit{acyclic} causal graphs (deferring the discussion of probabilities in the cyclic case to the next section). 
Finally, we provide examples of causal models.

\subsection{Causal models and acyclic probabilities}
\label{sec: causal models}
A causal graph is a decorated directed graph which is allowed to be cyclic.
The graph specifies causation relations between systems or variables, and the decorations specify whether the variables are observed or unobserved, and whether systems are classical or quantum.

\begin{definition}[\hypertarget{causalgraph}{Causal graph}]
\label{def: causal graph}
    A causal graph is defined to be a graph $G = (V,E)$ such that:
    \begin{myitem}
        \item The set of vertices $V$ can be partitioned into $\vertset = \overtset \cup \uvertset$, where the vertices $\vertname \in \overtset$ are called the \textup{observed} vertices, denoted as $\centertikz{\node[onode] {$\vertname$};}$, while the vertices $\vertname \in \uvertset$ are called the \textup{unobserved} vertices and are denoted as $\centertikz{\node[unode] {$\vertname$};}$.
        \item The set of edges can be partitioned into $\edgeset = \cledgeset \cup \qedgeset$, where the edges $\edgename \in \cledgeset$ are called the \textup{classical} edges, denoted as $\centertikz{\draw[cleg] (0,0) -- (14pt,0);}$, while the edges $\edgename \in \qedgeset$ are called the \textup{quantum} edges and are denoted as $\centertikz{\draw[qleg] (0,0) -- (14pt,0);}$.
        \item The outgoing edges of an observed vertex $\vertname \in \overtset$ are classical, i.e., $\outedges{\vertname} \subseteq \cledgeset$.
    \end{myitem}
\end{definition}

Given a causal graph $\graphname$, the incoming and outgoing classical or quantum edges to a vertex $\vertname\in\vertset$ are denoted
\begin{subequations}
\begin{align}
    \incledges{\vertname} &= \inedges{\vertname}\cap \cledgeset, & \inqedges{\vertname}&= \inedges{\vertname}\cap \qedgeset,\\
    \outcledges{\vertname} &= \outedges{\vertname}\cap \cledgeset, & \outqedges{\vertname}&= \outedges{\vertname}\cap \qedgeset.
\end{align}
\end{subequations}
A causal model assigns to a causal graph specific causal mechanisms for each vertex.
These causal mechanisms have to match the type of the vertex, e.g., if the vertex is observed, the causal mechanism has to specify how the observed outcome is obtained.

\begin{definition}[Causal model on causal graph]
\label{def:causal model}
    A causal model on a causal graph $\graphname=\graphexpl$, $\cm_{\graphname}$, is specified by the following items:
   \begin{myitem}   
        \item A finite-dimensional Hilbert space $\hilmaparg{\edgename}$ is associated to each edge $\edgename\in\edgeset$. We will use the notation where if $\edgeset'\subseteq\edgeset$ is a non-empty subset of edges,
        \begin{align}
            \label{eq:def h set of edges}
            \hilmaparg{\edgeset'} = \bigotimes_{\edgename \in \edgeset'} \hilmaparg{\edgename},
        \end{align}
        and if $\edgeset' = \emptyset$, then $\hilmaparg{\edgeset'} = \hilmaparg{\emptyset} = \mathbb{C}$.

        \item A random variable $X_\vertname$ taking values $x_\vertname$ 
        from a non-empty finite set $\outcomemaparg{\vertname}$ is associated to each observed vertex $\vertname\in\overtset$.
        
        \item A finite set $\outcomemaparg{\edgename}$ is associated to each classical edge $\edgename\in\cledgeset$.
        The Hilbert space associated to a classical edge then has to take the form $\hilmaparg{\edgename} = \hilmapsetarg{\outcomemaparg\edgename}$ with associated computational basis $\{\ket{\outcome}\}_{\outcome\in\outcomemaparg\edgename}$ \textup{(see \cref{eq:Hxspace})}.
        For consistency, we require that for all observed vertices $\vertname\in\overtset$, for all outgoing edges $\edgename \in \outedges{\vertname}$ thereof, the set of outcomes match: $\outcomemaparg{\edgename} = \outcomemaparg{\vertname}$.
        \item A CPTP map is associated to each unobserved vertex $\vertname\in\uvertset$:
        \begin{equation}
            \centertikz{
                \node[unode] (q) at (0,0) {$\vertname$};
            }\mapsto \chanmaparg{\vertname}: 
            \linops\left(\hilmaparg{\inedges{v}}\right)
            \mapsto
            \linops\left(\hilmaparg{\outedges{v}}\right),
        \end{equation}
        where $\hilmaparg{\inedges{\vertname}}$ and $\hilmaparg{\outedges{\vertname}}$ refer to \cref{eq:def h set of edges}.
        If $\outcledges{\vertname}$ is non-empty, the CPTP map $\chanmaparg\vertname$ satisfies the decoherence condition
        \begin{equation}
            \chanmaparg{\vertname} = \mathcal{D}^{\vertname} \circ \chanmaparg{\vertname},
        \end{equation}
        where the channel $\mathcal{D}^\vertname$ is defined as $\mathcal{D}^\vertname = \bigotimes_{\edgename\in\outcledges{\vertname}} \mathcal{D}_\edgename$,
        and  $\mathcal{D}_\edgename : \linops(\hilmaparg{\edgename}) \mapsto \linops(\hilmaparg{\edgename})$ is a decohering channel acting as $\mathcal{D}_\edgename(\rho) = \sum_{\outcome \in \outcomemaparg\edgename} \ketbra{\outcome} \rho \ketbra{\outcome}$.
    \item A POVM is associated to each observed vertex $\vertname\in\overtset$:
    \begin{equation}
        \centertikz{
            \node[onode] (q) at (0,0) {$\vertname$};
        } \mapsto
        \povmarg{\vertname} = \{\povmelarg{\outcome}{\vertname} \in \linops(\hilmaparg{\inedges{\vertname}}) \}_{\outcome \in \outcomemaparg\vertname}.
    \end{equation}
    We furthermore associate to $\vertname$ a set of CP maps
    \begin{align}
        \left\{ 
        \measmaparg{\outcome}{\vertname} :  
        \linops\left(\hilmaparg{\inedges{v}}\right) \mapsto \linops\left(\hilmaparg{\outedges{v}}\right)
        \right\}_{\outcome\in\outcomemaparg{\vertname}}
    \end{align}
    defined as follows: for all $\rho \in \linops(\hilmaparg{\inedges{\vertname}})$,
    \begin{equation}
    \label{eq:def_onode}
        \measmaparg{\outcome}{\vertname}(\rho) = \Tr[\povmelarg{\outcome}{\vertname} \rho] \bigotimes_{\edgename\in\outedges{\vertname}} \ketbra{\outcome}_{\hilmaparg{\edgename}}\,.
    \end{equation}
    By definition of POVM, $\povmarg{\vertname} = \{\povmelarg{\outcome}{\vertname} \}_{\outcome \in \outcomemaparg\vertname}$ satisfies $\sum_{\outcome \in \outcomemaparg\vertname}\povmelarg{\outcome}{\vertname} = \id_{\hilmaparg{\inedges{\vertname}}}$ and the map $\sum_{\outcome \in \outcomemaparg\vertname}\measmaparg{\outcome}{\vertname}$ is CPTP.
    \end{myitem}
\end{definition}

Few comments on the definition are in order:
\begin{itemize}
    \item The decoherence condition required for CPTP maps associated to unobserved vertices ensures that the output of $\chanmaparg{\vertname}$ on systems defined over classical edges is classical, i.e., diagonal in the computational basis.  The analogous condition for classical input edges $\edgename = (\vertname',\vertname)\in\cledgeset$ is already satisfied due to the decoherence condition of the map associated to $\vertname'$. 
    \item If a vertex $\vertname$ is exogenous, the set $\inedges{\vertname}$ is empty. In this case, $\linops(\hilmaparg{\inedges{\vertname}}) = \linops\left(\hilmaparg{\emptyset}\right) = \linops\left(\mathbb{C}\right) \cong \mathbb{C}$ whose only positive and trace-one element is $1 \in \mathbb{C} \simeq \linops(\mathbb{C})$. Therefore, for exogenous unobserved vertices we will write
        $
            \statemaparg\vertname = \chanmaparg\vertname(1) \in \linops\left(\hilmaparg{\outedges\vertname}\right)
        $, which can be shown to be a valid density operator.
        For exogenous observed vertices, we will also denote $\probmaparg\outcome\vertname = \Tr[\povmelarg\outcome\vertname]$, which is a valid probability distribution over $\outcome \in \outcomemaparg{\vertname}$, and  $\substatemaparg{\outcome}{\vertname} = \measmaparg{\outcome}{\vertname}(1)$. It can be checked that $\{\substatemaparg\outcome\vertname\}_{\outcome\in\outcomemaparg\vertname}$ is a collection of sub-normalized density matrices that sum up to a normalized density matrix.
\end{itemize}

Our goal for the next section will be to define a probability distribution over the values of the observed vertices given a causal model on any causal graph.
An important special case of causal graphs happens when the graph is \emph{acyclic}. 
In that case, we can define a probability distribution over the observed vertices in the usual way (in particular, as given by the acyclic causal modelling formalism of Henson, Lal and Pusey for the quantum case ~\cite{Henson_2014}). This is equivalent to considering the quantum protocol specified by the acyclic causal model and applying the Born rule.
\begin{definition}[\hypertarget{probacyc}{Probabilities of acyclic causal graphs}]
\label{def: acyclic probability}
    Consider a causal model on an \emph{acylic} causal graph $\graphname = (\vertset,\edgeset)$ and a global observed event $\outcome:=\{\outcome_\vertname \in \outcomemaparg{\vertname}\}_{\vertname\in\overtset}$.
    The probability $\probacyc(\outcome)_{\graphname} \in [0,1]$ is obtained by composing all the channels according to the graph, and we denote this as follows:
    \begin{align}
    \label{eq:def acyclic composition}
        \probacyc(\outcome)_{\graphname}
        =
        \bigcomp_{\vertname \in \overtset}
        \measmaparg{\outcome_{\vertname}}{\vertname} 
        \bigcomp_{u \in \uvertset} \chanmaparg{u}.
    \end{align}

    Concretely, one first takes the tensor product of the states $\statemaparg u = \chanmaparg{u}(1)$ associated to each exogenous unobserved vertex $u$ and of the subnormalized states $\measmaparg{\outcome_\vertname}{\vertname}(1)$ associated to each exogenous observed vertex $\vertname$.
    Then, one applies the composition of the CPTP maps $\chanmaparg{u}$ associated to each endogenous unobserved vertex $u$ and the CP map $\measmaparg{\outcome_{\vertname}}{\vertname}$ associated to each endogenous observed vertex, with the composition rule\footnote{Meaning: the order in which to compose, and how to take the tensor product of each map with identity channels so that each map acts on the correct subsystems.} being obtained from the connectivity of the graph $\graphname$. 
    Eventually, because the graph is finite and acyclic, the output of this composition operation will be in $\linops(\mathbb{C}\otimes \cdots\otimes\mathbb{C})$, which we canonically identify with $\mathbb{C}$. Thus the composition operation $\bigcomp$ entails both parallel and sequential composition (in an order unambiguously specified by the acyclic causal graph).
\end{definition}

Notice that the maps associated to exogenous vertices in~\cref{def:causal model} are defined for all linear operators on $\linops(\mathbb{C})\cong \mathbb{C}$ but we choose $1$ as input state in the probability rule of~\cref{def: acyclic probability}. By linearity, the choice of any other $\lambda\in\mathbb{C}$ in~\cref{def: acyclic probability} would yield 
\begin{align}      \textup{Pr}_{\textup{acyc}}'(\outcome)_{\graphname}=\lambda\probacyc(\outcome)_{\graphname}.      
\end{align}
It is easy to verify that this defines a valid probability distribution only for $\lambda=1$.
 
\subsection{Examples of acyclic causal graphs}
\label{sec:example acyclic graphs}
In this section, we present few examples of causal models on acyclic graphs. These should clarify how to define causal models on causal graphs and evaluate the acyclic probability distribution over observed vertices given in~\cref{def: acyclic probability}.

\paragraph{Prepare-and-measure scenarios.}
Consider the following causal graphs:
\newcommand{\qgraph}{G_{\textup{q}}}
\newcommand{\cgraph}{G_{\textup{c}}}
\begin{align}
\label{eq:causal graph pm}
    \qgraph = 
    \centertikz{
        \node[onode] (x) {$A$};
        \node[unode] (m) [above=\chanvspace of x] {$L$};
        \node[onode] (y) [above=\chanvspace of m] {$B$};
        \draw[cleg] (x) -- (m);
        \draw[qleg] (m) -- (y);
    }
    \quad \text{and} \quad
    \cgraph =
    \centertikz{
        \node[onode] (x) {$A$};
        \node[unode] (m) [above=\chanvspace of x] {$L$};
        \node[onode] (y) [above=\chanvspace of m] {$B$};
        \draw[cleg] (x) -- (m);
        \draw[cleg] (m) -- (y);
    }.
\end{align}
Let us analyze a general causal model for $\qgraph$.
The vertex $A$ is observed, thus, it has associated a POVM $\{\povmelarg{a}{A} \in \linops(\hilmaparg{\inedges{A}})\}_{a\in \outcomemaparg{A}}$. 
Since the vertex $A$ has no incoming edges, we have by definition that $\hilmaparg{\inedges{A}} = \hilmaparg{\emptyset} = \mathbb{C}$, so that the POVM $\{\povmelarg{a}{A}\}_{a\in\outcomemaparg{A}}$ amounts to a probability distribution $\{\probmaparg{a}{A}\}_{a\in\outcomemaparg{A}}$.
$A$ only has one outgoing edge: $\outedges{A} = \{e = (A,L)\}$.
The corresponding set of CP maps, $\{\measmaparg{a}{A} : \mathbb{C} \mapsto \linops(\hilmaparg{e})\}_{a\in\outcomemaparg{A}}$, can be substituted for a set of sub-normalized density matrices $\substatemaparg a A = \measmaparg{a}{A}(1)$.
We have
\begin{align}
    \substatemaparg{a}{A} = \probmaparg a A \ketbra{a},
\end{align}
where $\{\ket a \}_{a\in\outcomemaparg A}$ is the computational basis of the Hilbert space $\hilmaparg e = \hilmapsetarg{\outcomemaparg A}$ associated to the classical edge $e = (A,L)$.
Then, the unobserved vertex $L$ (for ``latent'') has an associated CPTP map $\chanmaparg{L} : \linops(\hilmaparg e) \mapsto \linops(\hilmaparg{e'})$, where $e' = (L,B)$ is the edge going from $L$ to $B$.
Finally, the vertex $B$ has an associated POVM $\{\povmelarg{b}{B} \in \linops(\hilmaparg{e'})\}_{b\in\outcomemaparg{B}}$, and an associated set of CP maps $\{\measmaparg b B : \linops(\hilmaparg{e'}) \mapsto \mathbb{C}\}_{b \in \outcomemaparg B}$ that simply act as $\measmaparg b B(\rho) = \Tr[\povmelarg b B \rho]$.
Since the causal graph is acyclic, we can use \cref{def: acyclic probability} to obtain the following probability distribution over the observed vertices:
\begin{subequations}
\label{eq:q prob pm}
\begin{align}
    \probacyc(a,b)_{\graphname_q} 
    &=  \measmaparg b B \circ \chanmaparg L \circ \measmaparg a A (1) \\
    &=  \Tr[ \povmelarg b B \chanmaparg L (\substatemaparg a A) ] \\
    &= \probmaparg a A \Tr[ \povmelarg b B \chanmaparg L (\ketbra{a}) ].
\end{align}
\end{subequations}
This probability can be interpreted as follows: the value $a$ is first drawn with probability $\probmaparg a A$, the state $\chanmaparg L(\ketbra{a})$ is then prepared (this can be any state: the value $a$ merely encodes the classical label of a collection of states), and this state $\chanmaparg L(\ketbra{a})$ is then measured with the POVM $\{\povmelarg b B\}_{b\in\outcomemaparg{B}}$.

Let us now analyze the causal graph $\cgraph$ of \cref{eq:causal graph pm}.
Compared to the previous case of $\qgraph$, the difference is that the edge $e' = (L,B)$ is now classical.
In any causal model, it is thus equipped with a set $\outcomemaparg{e'}$.
The channel $\chanmaparg L$ associated to the vertex $L$ has a classical outgoing edge $e'$: it has to satisfy the decoherence condition $\mathcal D \circ \chanmaparg L = \chanmaparg L$.
In particular, for the input state $\ketbra a \in \linops(\hilmaparg e) = \linops(\hilmapsetarg A)$, this implies that
\begin{align}
    \chanmaparg L(\ketbra{a})
    &= \mathcal D \circ \chanmaparg L(\ketbra a) \\
    &= \sum_{l \in \outcomemaparg{e'}} \bra{l} \chanmaparg{L}(\ketbra{a}) \ket{l} \ketbra{l} \\
    &=: \sum_{l\in\outcomemaparg{e'}} \probmaparg{l|a}{L} \ketbra{l},
\end{align}
where it can be checked that $\probmaparg{l|a}{L}$ is a probability distribution on $l \in \outcomemaparg{e'}$.
The probability of \cref{eq:q prob pm} then reads:
\begin{align}
    \probacyc(a,b)_{\graphname_c}
    = \probmaparg a A \sum_{l\in\outcomemaparg{e'}} \probmaparg{l|a} L \bra{l} \povmelarg b B \ket{l}
    =: \probmaparg a A \sum_{l\in\outcomemaparg{e'}} \probmaparg{l|a} L \probmaparg{b|l} B,
\end{align}
where it can be checked again that $\probmaparg{b|l}B$ is a probability distribution on $b$.
This probability can be interpreted as follows: a value $a$ is generated with probability $\probmaparg a A$, followed by a value $l$ generated with probability $\probmaparg{l|a} A$, and finally the result $b$ is obtained with probability $\probmaparg{b|l} B$.

\paragraph{Bell scenario.}
\label{example: Bell scenario}
We now consider the following two causal graphs:
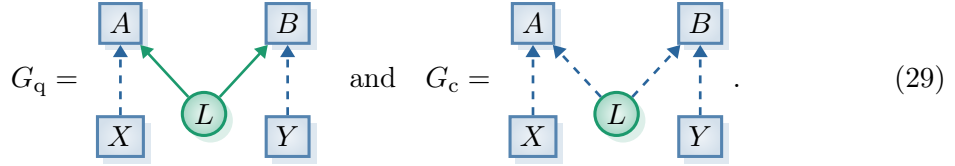
\begin{align}
\label{eq:causal graph bell}
    \qgraph = 
    \centertikz{
    \begin{scope}[xscale=2.2,yscale=1.5]
        \node[onode] (x) at (0,0) {$X$};
        \node[onode] (y) at (1,0) {$Y$};
        \node[onode] (a) at (0,1) {$A$};
        \node[onode] (b) at (1,1) {$B$};
        \node[unode] (l) at (0.5,0.2) {$L$};
        \draw[cleg] (x) -- (a);
        \draw[cleg] (y) -- (b);
        \draw[qleg] (l) -- (a);
        \draw[qleg] (l) -- (b);
    \end{scope}
    }
    \quad \text{and} \quad
    \cgraph = 
    \centertikz{
    \begin{scope}[xscale=2.2,yscale=1.5]
        \node[onode] (x) at (0,0) {$X$};
        \node[onode] (y) at (1,0) {$Y$};
        \node[onode] (a) at (0,1) {$A$};
        \node[onode] (b) at (1,1) {$B$};
        \node[unode] (l) at (0.5,0.2) {$L$};
        \draw[cleg] (x) -- (a);
        \draw[cleg] (y) -- (b);
        \draw[cleg] (l) -- (a);
        \draw[cleg] (l) -- (b);
    \end{scope}
    }.
\end{align}
In both cases, we will label the Hilbert spaces as follows: $\hilmaparg{(X,A)} = \hilmaparg{X}$, $\hilmaparg{(L,A)} = \hilmaparg{L_1}$, $\hilmaparg{(L,B)} = \hilmaparg{L_2}$, and $\hilmaparg{(Y,B)} = \hilmaparg{Y}$.
We start with the quantum Bell scenario, described by the causal graph $\qgraph$.
The unobserved vertex $L$ is associated a CPTP map $\chanmaparg L : \mathbb{C} \mapsto \linops(\hilmaparg{L_1} \otimes \hilmaparg{L_2})$, from which we obtain the state $\statemaparg L  = \chanmaparg L(1)$.
The treatment of the vertices $X,Y$ is analogous to the vertices $A$ and $B$ in the previous prepare-and-measure scenario. Thus, the probability distribution $p_\outcome^X$ is associated to $X$ and $p_\outcomealt^Y$ to $Y$. These vertices have one child, thus the associated CP maps are
\begin{equation}
    \measmaparg x X(1) = p_\outcome^X \ketbra{x}_{X} \textup{ and } \measmaparg y Y(1) = p_\outcomealt^Y \ketbra{y}_{Y}.
\end{equation}
The vertices $A$ and $B$ are observed with two incoming edges each, thus they are associated respectively with POVMs $\{\povmelarg{a}{A}\in\linops(\hilmaparg{X}\otimes\hilmaparg{L_1})\}_{a\in\outcomemaparg{A}}$
 and $\{\povmelarg{b}{B}\in\linops(\hilmaparg{L_2}\otimes\hilmaparg{Y})\}_{b\in\outcomemaparg{B}}$. Since these vertices have no children, the associated CP maps act on states of the form $\rho_{L_1}\otimes\ketbra{\outcome}_X$ for all $\rho_{L_1}\in\linops(\hilmaparg{L_1})$ and  $\rho_{L_2}\otimes\ketbra{\outcome}_Y$ for all $\rho_{L_2}\in\linops(\hilmaparg{L_2})$ as
\begin{equation}
    \measmaparg a A(\rho_{L_1}\otimes\ketbra{\outcome}_X) = \Tr_{L_1X}\left[\povmelarg{a}{A}\rho_{L_1}\otimes\ketbra{\outcome}_X\right] =: \Tr_{L_1}\left[\povmelarg{a|x}{A}\rho_{L_1}\right]
\end{equation}
and
\begin{equation}
   \measmaparg b B(\rho_{L_2}\otimes\ketbra{\outcomealt}_Y) = \Tr_{L_2Y}\left[\povmelarg{b}{B}\rho_{L_2}\otimes\ketbra{\outcomealt}_Y\right] =: \Tr_{L_2}\left[\povmelarg{b|y}{B}\rho_{L_2}\right]  
\end{equation}
where we defined $\povmelarg{a|x}{A}=\Tr_X[\povmelarg{a}{A}\ketbra{x}_X]\in\linops(\hilmaparg{L_1})$, and $\povmelarg{b|y}{B}=\Tr_Y[\povmelarg{b}{B}\ketbra{y}_Y]\in\linops(\hilmaparg{L_2})$
The acyclic probability rule yields in this case:
\begin{align}
    \probacyc(a,b,x,y)_{\graphname_q}
    &= (\measmaparg a A)_{X L_1} \otimes (\measmaparg b B)_{L_2 Y} \Big(\measmaparg x X(1) \otimes \chanmaparg L(1) \otimes \measmaparg y Y(1)\Big) \\
    &= (\measmaparg a A)_{X L_1} \otimes (\measmaparg b B)_{L_2 Y} \Big(\probmaparg x X \ketbra{x}_X \otimes \statemaparg L_{L_1 L_2} \otimes \probmaparg y Y \ketbra{y}_{Y}\Big) \\
    &= \probmaparg x X\probmaparg y Y \Tr_{XL_1L_2Y}\Big[ (\povmelarg a A)_{XL_1} (\povmelarg b B)_{L_2Y} \ketbra{x}_X \statemaparg L_{L_1 L_2}   \ketbra{y}_{Y}\Big] \\
    &= \probmaparg x X\probmaparg y Y \Tr_{L_1L_2}\Big[ (\povmelarg {a|x} A)_{L_1} (\povmelarg {b|y} B)_{L_2} \statemaparg L_{L_1 L_2} \Big]. \label{eq:q bell prob}
\end{align}
The expression in \cref{eq:q bell prob} is what we expect in the quantum Bell scenario: Alice and Bob sample a setting $x$ ($y$) with probability $\probmaparg x X$ ($\probmaparg y Y$), and then measure their shared state $\statemaparg L_{L_1L_2}$ with their local POVMs $\{\povmelarg{a|x}A\}_a$ ($\{\povmelarg{b|y}B\}_b$).

Let us now analyze the case of the classical Bell scenario, described by the causal graph $\cgraph$ in \cref{eq:causal graph bell}.
The edges $L_1 = (L,A)$ and $L_2 = (L,B)$ are now assigned a finite set $\outcomemaparg{L_1}$ and $\outcomemaparg{L_2}$, respectively.
The channel $\chanmaparg L$ satisfies the decoherence condition $\chanmaparg L = \mathcal D \circ \chanmaparg L$, which implies that
\begin{align}
    \statemaparg L
    &= \chanmaparg L(1)
    = \mathcal D \circ \chanmaparg L(1)
    = \mathcal D(\statemaparg L) \\
    &= \sum_{\substack{l_1 \in \outcomemaparg{L_1} \\ l_2 \in \outcomemaparg{L_2}}} 
    \left(\bra{l_1}_{L_1} \otimes \bra{l_2}_{L_2} (\statemaparg L)_{L_1L_2} \ket{l_1}_{L_1} \otimes \ket{l_2}_{L_2}\right) \ketbra{l_1}_{L_1} \otimes \ketbra{l_2}_{L_2} \\
    &=: \sum_{\substack{l_1 \in \outcomemaparg{L_1} \\ l_2 \in \outcomemaparg{L_2}}} \probmaparg{l_1,l_2}L \ketbra{l_1}_{L_1} \otimes \ketbra{l_2}_{L_2},
\end{align}
We then define the response functions of the parties as follows: 
\begin{align}
    \probmaparg{a|x,l_1} A &= \Tr\Big[(\povmelarg{a}{A})_{XL_1} \ketbra{x}_X \otimes \ketbra{l_1}_{L_1})\Big],\\
    \probmaparg{b|y,l_2} B &= \Tr\Big[(\povmelarg{b}{B})_{YL_2} \ketbra{y}_Y \otimes \ketbra{l_2}_{L_2})\Big].
\end{align}
It can be checked that the acyclic probability rule simplifies in this case to
\begin{align}
    \probacyc(a,b,x,y)_{\graphname_c} = \probmaparg x X \probmaparg y Y \sum_{\substack{l_1 \in \outcomemaparg{L_1}\\l_2\in\outcomemaparg{L_2}}}
    \probmaparg{l_1,l_2} L \probmaparg{a|x,l_1} A \probmaparg{b|y,l_2} B,
\end{align}
which satisfies Bell's local causality condition, and can be obtained in a local hidden variable model where Alice and Bob use classical response functions, and where the shared source distributes two systems described classically, i.e., by a probability distribution $\probmaparg{l_1,l_2}L$.

\section{Mapping cyclic to acyclic causal models with post-selection}
\label{sec: cyclic_to_acyc_v3}
Our aim is to define a probability distribution over the values of the observed vertices given a causal model on a causal graph. 
To this goal, we map a given causal model on a (possibly cyclic) causal graph to a causal model with post-selection on an acyclic causal graph. 
This allows us to define the probability rule of the original causal model using the acyclic probability rule in~\cref{def: acyclic probability}.
The mapping is constructed from post-selected teleportation protocols, introduced in the next section, which are known to simulate closed time-like curves~\cite{Lloyd_2011_2,Lloyd_2011}.

\subsection{Post-selected teleportation}
\label{sec:post-selected teleportation}

Teleportation is a protocol that proceeds on three systems $A$, $B$ and $C$. The systems $A$ and $C$ are identified as the same type, so that it makes sense to say that $A$ and $C$ are in the same state.
The teleportation protocol goes as follows: $A$ is prepared in some initial state, and $B$ and $C$ are prepared in some correlated state.
Then, a measurement is performed on the systems $A$ and $B$.
After this, a correction that depends on the measurement outcome is applied on the system $C$, so that the final state of $C$ is the same as the initial state on $A$.
The final state on $A$ is arbitrary.
We may represent this teleportation protocol as follows:
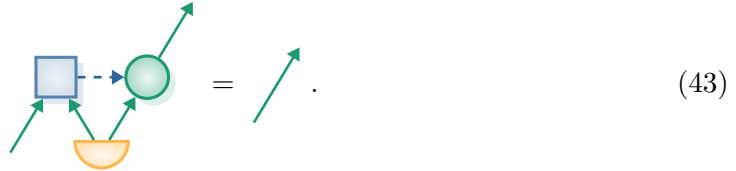
\begin{align}
    \centertikz{
    \begin{scope}[xscale=1.2]
        \node[prenode] (pre) at (1,0) {};
        \node[onode] (post) at (0.5,1) {};
        \draw[qleg] (pre) -- (post);
        \draw[qleg] (0,0) -- (post);
        \node[unode] (op) at (1.5,1) {};
        \draw[cleg] (post) -- (op);
        \draw[qleg] (pre) -- (op);
        \draw[qleg] (op) -- ++(0.5,1);
    \end{scope}
    }
    =
    \centertikz{
    \begin{scope}[xscale=1.2]
        \draw[qleg] (0,0) -- (0.5,1);
    \end{scope}
    }.
\end{align}
A teleportation protocol may have the following feature: it can be that for a specific outcome of the $AB$ measurement, there is no correction to apply on the system $C$.
If that is the case, this allows for a post-selected teleportation protocol: upon conditioning on this specific outcome of the $AB$ measurement (i.e., discarding all rounds where a different outcome occurred), the teleportation succeeds without corrections.
We will represent this post-selected teleportation protocol as follows:
\begin{align}
    \centertikz{
    \begin{scope}[xscale=1.2]
        \node (A) at (0,0) {$A$};
        \node (C) at (2.25,1.25) {$C$};
        \node[prenode] (pre) at (1.5,0) {};
        \node[psnode] (post) at (0.75,1.25) {};
        \draw[qleg] (pre) -- node[midway,right] {$B$} (post);
        \draw[qleg] (A) -- (post);
        \draw[qleg] (pre) -- (C);
    \end{scope}
    }
    =
    \centertikz{
    \begin{scope}[xscale=1.2]
        \node (A) at (0,0) {$A$};
        \node (C) at (0.75,1.5) {$C$};
        \draw[qleg] (A) -- (C);
    \end{scope}
    }.
\end{align}
We now restrict the discussion to quantum theory.
\begin{definition}[Quantum post-selected teleportation protocol]
\label{def:ps teleportation}
    Let $\hilmaparg A = \hilmaparg C$ be finite-dimensional Hilbert spaces.
    A quantum post-selected teleportation protocol consists of a finite-dimensional Hilbert space $\hilmaparg B$ and a pair $(\telepovm_{AB},\telestate_{BC})$ where $\telepovm_{AB} \in \linops(\hilmaparg A \otimes \hilmaparg B)$ is a POVM element\footnote{In the teleportation protocol, there is a POVM describing the measurement. $\telepovm_{AB}$ here is the POVM element that corresponds to the outcome for which there is no correction to apply.} and $\telestate_{BC} \in \linops(\hilmaparg B \otimes \hilmaparg C)$ is a density matrix, such that for all density matrices $\rho_A \in \linops(\hilmaparg A)$,
    \begin{align}
    \label{eq:ps teleportation condition}
        \Tr_{AB}[\telepovm_{AB} \rho_A \telestate_{BC}] = \teleprob \rho_C,
    \end{align}
    where $\teleprob \in (0,1]$ is the success probability of the post-selected teleportation protocol.
\end{definition}
As we show in \cref{app:ps teleportation}, \cref{eq:ps teleportation condition} together with the linearity of quantum theory implies that $\teleprob$ is independent of the state $\rho_A$ being teleported. 
However, it may depend on the pair $(\telepovm_{AB},\telestate_{BC})$.
We also show that for any post-selected teleportation protocol $\teleprob\leq 1/\dim(\hilmaparg{A})^2$. 
Out of all possible choices of $(\telepovm_{AB},\telestate_{BC})$ that satisfy~\cref{def:ps teleportation}, we consider the following as canonical choice~\cite{Bennett1993}, which defines a valid post-selected protocol with optimal success probability of $\teleprob = 1/\dim(\hilmaparg{A})^2$.

\begin{definition}[Bell post-selected teleportation protocol]
\label{def:bell tele}
    The canonical choice of quantum post-selected teleportation protocol consists of choosing $\hilmaparg B = \hilmaparg A$ and 
    \begin{subequations}
    \begin{align}
        \telepovm_{AB} = \ketbra{\bellstate}_{AB}, \quad \telestate_{BC} = \ketbra{\bellstate}_{BC}, \\
        \text{where } \ket{\bellstate}_{AB} = \frac{1}{\sqrt{\dim(\hilmaparg{A})}} \sum_{i=1}^{\dim(\hilmaparg{A})} \ket{i}_A \ket{i}_B,
    \end{align}
    \end{subequations}
    and where $\{\ket i \}_{i=1}^{\dim(\hilmaparg A)}$ is an orthonormal basis. We refer to this choice as Bell post-selected teleportation protocol.
\end{definition}

\subsection{Family of acyclic causal models from a cyclic model}
\label{sec:family of acyclic CM}

In this section, we construct a family of acyclic causal models from a given, possibly cyclic model. 
The construction involves replacing a subset of edges of the graph with post-selected teleportation protocols. 
We first construct a family of acyclic causal graphs from a given causal graph.

\begin{restatable}[Family of acyclic causal graphs $\graphfamily\graphname$]{definition}{graphfamq}
\label{def: graph_family_v3}
 Given a causal graph $\graphname = (\vertset,\edgeset)$, we define an associated family $\graphfamily\graphname$ of directed acyclic causal graphs, where each element $\graphtele \in \graphfamily\graphname$ is obtained from the causal graph $\graphname$ as follows.
 \begin{myitem}
     \item Choose any subgraph $\graphname':=(\vertset',\edgeset')$ of $\graphname=(\vertset,\edgeset)$ with $\vertset'=\vertset$ and $\edgeset'\subseteq \edgeset$, such that $\graphname'$ is acyclic. 
 
     \item Include in $\graphtele$ all the vertices and edges of the subgraph $\graphname'$ associated with the same vertex types (observed or unobserved) and edge types (classical or quantum) as the original causal graph $\graphname$. 
     \item Denoting the set of so-called split edges $\splitedges{\graphtele}:=E\backslash E'$, for each edge $\edgearg{\vertname_i}{\vertname_i'} \in \splitedges{\graphtele}$, include in $\graphtele$, two vertices $\postvertname_i$ and $\prevertname_i$ and three edges  $\edgearg{\vertname_i}{\postvertname_i}$, $\edgearg{\prevertname_i}{\postvertname_i}$ and $\edgearg{\prevertname_i}{\vertname_i'}$. 
     \item The vertex $\postvertname_i$ is observed and $\prevertname_i$ unobserved. Outgoing edges from $\prevertname_i$, $(\prevertname_i,\postvertname_i)$ and $(\prevertname_i,\vertname_i')$, are quantum edges. The edge $(\vertname_i,\postvertname_i)$ is of the same type of the edge $(\vertname_i,\vertname_i')$ in the original causal graph $\graphname$.       
 \end{myitem}
 $\graphtele$ constructed in this manner is thus a causal graph. It will be useful to refer to $\postvertname_i$ and $\prevertname_i$ as pre and post-selection vertices respectively and depict them with distinct vertex styles  $\centertikz{\node[psnode] {$\postvertname_i$};}$ and $\centertikz{\node[prenode] {$\prevertname_i$};}$, as these will play a special role in our framework. This makes $\graphtele$ identical to $\graphname$ up to replacing each split edge $\edgearg{\vertname_i}{\vertname_i'} \in \splitedges{\graphtele}\subseteq \edgeset$ with the following structure:
     
        \begin{equation}
        \label{eq: acyc_graph_split_edge_v3}
            \centertikz
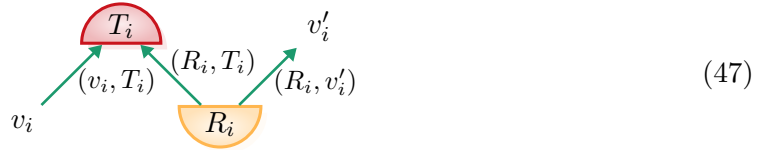
{
            \begin{scope}[xscale=2.6,yscale=1.3]
                \node (in) at (0,0) {$\vertname_i$};
                \node (out) at (1.5,1) {$\vertname_i'$};
                \node[prenode] (pre) at (1,0) {$\prevertname_i$};
                \node[psnode] (post) at (0.5,1) {$\postvertname_i$};
                \draw[qleg] (pre) -- node[pos=0.7,right] {\small$\edgearg{\prevertname_i}{\postvertname_i}$} (post);
                \draw[qleg] (in) -- node[pos=0.4,right=0pt] {\small$\edgearg{\vertname_i}{\postvertname_i}$} (post);
                \draw[qleg] (pre) -- node[pos=0.4,right] {\small$\edgearg{\prevertname_i}{\vertname_i'}$} (out);
            \end{scope}
            }
        \end{equation}
We will refer to every $\graphtele \in \graphfamily\graphname$ as a \textup{teleportation graph}, as we will later associate teleportation protocols to such graphs. It will be useful to denote the set of all post-selection vertices and the set of all pre-selection vertices in $\graphtele$ as $\psvertset:=\{\postvertname_i\}_{\edgearg{\vertname_i}{\vertname_i'} \in \splitedges{\graphtele}}$ and $\prevertset:=\{\prevertname_i\}_{\edgearg{\vertname_i}{\vertname_i'} \in \splitedges{\graphtele}}$.
\end{restatable}

As an example, consider the cyclic causal graph
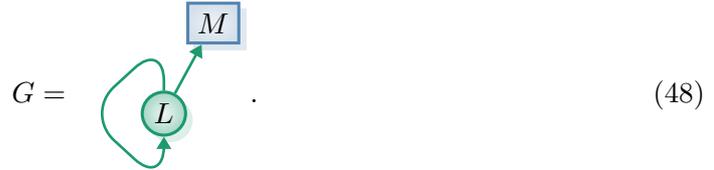
\begin{align}
\label{eq:selfcyclegraph1}
    \graphname = \centertikz{
        \begin{scope}[xscale=1.3,yscale=1.2]
            \node[unode] (rho) at (0,0) {$L$};
            \node[onode] (m) at (0.5,1) {$M$};
            \draw[qleg] (rho) -- (m);
            \draw[qleg,rounded corners=15pt] (rho) -- ++(0,0.8) -- ++(-0.8,-0.8) -- ++(0.8,-0.8) -- (rho);
        \end{scope}
    }.
\end{align}
Let us follow the steps of \cref{def: graph_family_v3}.
\begin{myitem}
    \item Choosing $\edgeset'= \{\edgearg{L}{M}\}$, we obtain the following acyclic subgraph $\graphname'$:
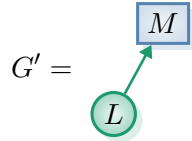
    $$
  \graphname' = \centertikz{
        \begin{scope}[xscale=1.3,yscale=1.2]
            \node[unode] (rho) at (0,0) {$L$};
            \node[onode] (m) at (0.5,1) {$M$};
            \draw[qleg] (rho) -- (m);
        \end{scope}
    }
    $$
    \item The only edge $\edgename$ which is absent in $\graphname'$ compared to $\graphname$ is $\edgename = \edgearg{L}{L}$, therefore we have $\splitedges{\graphtele} = \{\edgearg{L}{L}\}$.
    We have to add the post-selection vertex $\postvertname$ and pre-selection vertex $\prevertname$ to the graph, as well as the appropriate quantum edges.
    We obtain the following \ccdag:
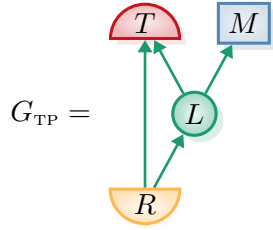
    $$
    \graphtele = \centertikz{
        \begin{scope}[xscale=1.3,yscale=1.2]
            \node[unode] (rho) at (0,0) {$L$};
            \node[onode] (m) at (0.5,1) {$M$};
            \draw[qleg] (rho) -- (m);
            \node[psnode] (o) at (-0.5,1) {$\postvertname$};
            \node[prenode] (r) at (-0.5,-1) {$\prevertname$};
            \draw[qleg] (rho) -- (o);
            \draw[qleg] (r) -- (o);
            \draw[qleg] (r) -- (rho);
        \end{scope}
    }
    $$
\end{myitem}
Alternatively, we could remove all edges of the graph, corresponding to $\edgeset'=\emptyset$ and $\splitedges{\graphtele} = \{\edgearg{L}{L},\edgearg{L}{M}\}$, resulting in the following \ccdag:
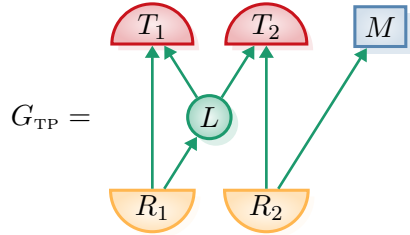
$$
    \graphtele = \centertikz{
        \begin{scope}[xscale=1.5,yscale=1.2]
            \node[unode] (rho) at (0,0) {$L$};
            \node[onode] (m) at (1.5,1) {$M$};
            \node[psnode] (o) at (-0.5,1) {$\postvertname_1$};
            \node[prenode] (r) at (-0.5,-1) {$\prevertname_1$};
            \draw[qleg] (rho) -- (o);
            \draw[qleg] (r) -- (o);
            \draw[qleg] (r) -- (rho);
            \node[prenode] (r2) at (0.5,-1) {$\prevertname_2$};
            \node[psnode] (o2) at (0.5,1) {$\postvertname_2$};
            \draw[qleg] (r2) -- (o2);
            \draw[qleg] (rho) -- (o2);
            \draw[qleg] (r2) -- (m);
        \end{scope}
    }
$$
The following lemma, which is proven in~\cref{app:proofs_map}, confirms that the result of the above definition is indeed a family of \emph{acyclic} graphs.
\begin{restatable}[Acyclicity of teleportation graphs]{lemma}{acyclicitytele}
    \label{lemma: teleportation_graph_acyclic}
    Each directed graph $\graphtele$ obtained from a directed graph $\graphname$ as described in \cref{def: graph_family_v3} is acyclic. 
\end{restatable}

We constructed a family $\graphfamily{\graphname}$ of acyclic causal graphs for any given causal graph $\graphname$. 
Given a causal model on $\graphname$, we can construct a corresponding family of causal models on $\graphfamily{\graphname}$, as in the following definition.

\begin{definition}[Teleportation causal models on the graph family $\graphfamily{\graphname}$]
\label{def:causal model of graphfamily_v3}
 Given a causal model, $\cmG$, associated with a causal graph $\graphname$, we can define a corresponding family of causal models, by associating a causal model $\cm_{\graphtele}$ to each teleportation graph $\graphtele\in \graphfamily\graphname$, as follows: 
 \begin{myitem}
     \item For every edge, as well as every vertex present in both $\graphname$ and $\graphtele$, the assigned Hilbert spaces, outcome sets, CPTP maps and POVMs (according to \cref{def:causal model}) are the same in the two causal models $\cmG$ and $\cmtele$.
     \item  Using the same notation as in \cref{def: graph_family_v3}, for the pre-selection vertex $\prevertname_i$, post-selection vertex $\postvertname_i$ and corresponding three edges  $\edgearg{\vertname_i}{\postvertname_i}$, $\edgearg{\prevertname_i}{\postvertname_i}$ and $\edgearg{\prevertname_i}{\vertname_i'}$ introduced in 
     $\graphtele$ for each edge $\edgearg{\vertname_i}{\vertname_i'} \in \splitedges{\graphtele}$ that was removed from $\graphname$, the causal model $\cmtele$ has the following specifications: 
     
     \begin{enumerate}
         \item  The Hilbert spaces associated to the edges are
        \begin{align}
        \label{eq:hilbert space identities_v3}
            \underbrace{
            \hilmaparg{\edgearg{\vertname_i}{\postvertname_i}}
            =
            \hilmaparg{\edgearg{\prevertname_i}{\vertname_i'}}
            }_{\textup{in $\cmtele$}}
            =
            \underbrace{\hilmaparg{\edgearg{\vertname_i}{\vertname_i'}}}_{\textup{in $\cm_{\graphname}$}}.
        \end{align}
          \item The outcome set associated to the post-selection vertex consists of $\outcomemaparg{\postvertname_i} = \{\ok,\fail\}$, the outcome taking values in this set will be denoted as $\postoutcome_i$.
        \item The POVM element $\povmelarg{\ok}{\postvertname_i}$ of the post-selection vertex and the state $\statemaparg{\prevertname_i}$ of the pre-selection vertex form a post-selected teleportation protocol (\cref{def:ps teleportation}).\footnote{As shown in \cref{lemma:self test}, this implies that $\dim(\hilmaparg{\edgearg{\prevertname_i}{\postvertname_i}}) \geq \dim(\hilmaparg{\edgename_i}) = \dim(\hilmaparg{\edgename_i'})$. The assigned Hilbert spaces in \cref{eq:hilbert space identities_v3} satisfy this condition.}
     \end{enumerate}
Furthermore, the labelling of the CPTP and CP maps of the causal models are such that the $\hilmaparg{\edgearg{\vertname_i}{\postvertname_i}}$ output of a map in $\cmtele$ is identified with the $\hilmaparg{\edgearg{\vertname_i}{\vertname_i'}}$ output of the same map in $\cmG$. Similarly, the $\hilmaparg{\edgearg{\prevertname_i}{\vertname_i'}}$ input of a map in $\cmtele$ is identified with the $\hilmaparg{\edgearg{\vertname_i}{\vertname_i'}}$ input of the same map in $\cmG$.

We will refer to each such $\cmtele$ as a \textup{teleportation causal model}.
 \end{myitem}
\end{definition}

As every causal model constructed according to \cref{def:causal model of graphfamily_v3} is an acyclic causal model, we can readily compute probabilities within such models using the acyclic probability rule of \cref{def: acyclic probability}. Specifically, note that all observed vertices $\overtset\subseteq \vertset$ of $\graphname$ are preserved in every $\graphtele\in \graphfamily\graphname$. For a given teleportation graph $\graphtele$, denoting $\psvertset$ as the set of all post-selection vertices of the graph (which are by construction also observed), $\overtset\cup \psvertset$ corresponds to the set of all observed vertices of $\graphtele$. 
Therefore, by applying the acyclic probability rule (\cref{def: acyclic probability}) to any causal model $\cmtele$ defined on $\graphtele$ (\cref{def:causal model of graphfamily_v3}), we can compute the probability distribution $   \probacyc\left(\outcome, \{\postoutcome_i = \ok\}_{\postvertname_i\in\psvertset}\right)_{\graphtele}$ on the observed vertices of $\graphtele$, where $\outcome:=\{\outcome_\vertname\}_{\vertname\in\overtset}$. We are interested in the event that the post-selections involved in all the teleportation protocols of the model succeed. 

\begin{definition}[Success probability in a teleportation causal model]
\label{def:success_prob_causal_model}
Let $\cmG$ be a causal model of a causal graph $\graphname$ and $\cmtele$ be a corresponding teleportation causal model defined on $\graphtele\in \graphfamily\graphname$. The post-selection success probability of $\cmtele$ is defined as follows
\begin{align}
      \successprob := \probacyc\left(\{\postoutcome_{i} = \ok\}_{\postvertname_i\in\psvertset}\right)_{\graphtele}
            = \sum_{\outcome} \probacyc\left(
           \outcome,
            \{\postoutcome_{i} = \ok\}_{\postvertname_i\in\psvertset}
            \right)_{\graphtele},
\end{align}
where the summation is over $\outcome:=\{\outcome_\vertname\}_{\vertname\in\overtset}$ and $\overtset$ is the set of all observed vertices of $\graphname$.
\end{definition}

This allows to obtain the probability associated with observed events in $\overtset$, conditioned on the success of the post-selections in the teleportation causal model, which will be the probability of interest in our framework,
\begin{align}
\label{eq: teleportationCM_prob}
\probacyc\left(\outcome \middle| \{\postoutcome_i = \ok\}_{\postvertname_i\in\psvertset}\right)_{\graphtele} 
            = \frac{\probacyc\left(
         \outcome,
            \{\postoutcome_i = \ok\}_{\postvertname_i\in\psvertset}
            \right)_{\graphtele}}{\successprob}.
\end{align}

In~\cref{app:proofs_map}, we prove the following proposition which shows that this conditional probability is independent of the choice of teleportation graph $\graphtele\in \graphfamily\graphname$.

\begin{restatable}[Equivalent probabilities from different teleportation graphs]{prop}{equivprobabilitiestele}
    \label{lemma: acyclic_prob_same_v3}
  Let $\cmG$ be a causal model on a causal graph $\graphname$, and let $\overtset \subseteq \vertset$ be the set of all observed vertices of $\graphname$.
  Consider any $\graphtele_{,1},\graphtele_{,2} \in \graphfamily\graphname$, and for $i\in\{1,2\}$, let $\cm_{\graphtele_{,i}}$ be a causal model on the teleportation graph $\graphtele_{,i}$ that is associated to $\cmG$ (\cref{def:causal model of graphfamily_v3}). Then, we have
  \begin{align}
      \probacyc\left(\outcome \middle| \{\postoutcome_i = \ok\}_{\postvertname_i\in\psvertset^1}\right)_{\graphtele_{,1}} =      \probacyc\left(\outcome \middle| \{\postoutcome_i = \ok\}_{\postvertname_i\in\psvertset^2}\right)_{\graphtele_{,2}},
  \end{align} 
  where we denoted the set of all post-selection vertices of $\graphtele_{,1},\graphtele_{,2}$ as $\psvertset^1$ and $\psvertset^2$ respectively, and the joint observed event of $\overtset$ in short as $\outcome := \{\outcome_\vertname \in \outcomemaparg{\vertname}\}_{\vertname\in\overtset}$. 
\end{restatable}

Thus, as long as we are only interested in the conditional probabilities of~\cref{eq: teleportationCM_prob}, all teleportation graphs $\graphtele\in\graphfamily{\graphname}$ are equivalent. We define the maximal teleportation graph as canonical choice.

\begin{definition}[Maximal teleportation graph]
\label{def:max tele graph}
    Given a causal graph $\graphname=\graphexpl$, consider the family of acyclic causal graphs $\graphfamily{\graphname}$. The canonical choice of teleportation graph $\graphname_{\textsc{tp}}^{\textup{max}}\in\graphfamily{\graphname}$ consists in choosing $\splitedges{\graphname_{\textsc{tp}}^{\textup{max}}} = \edgeset$. We refer to this choice as maximal teleportation graph.
\end{definition}

\subsection{General probability rule for cyclic quantum causal models}
\label{sec:probability rule CM}
In the previous section, we have constructed a family of acyclic causal models associated with a given causal model~(\cref{def:causal model of graphfamily_v3}). 
We have proven that the acyclic distribution conditioned on successful post-selection is independent of which teleportation causal model we choose in this family.
In this section, we use these results to define a probability rule for the cyclic causal model underlying the family.

\begin{definition}[\hypertarget{prob}{Probabilities in a general causal model}]
\label{def: probability distribution v3}
 Consider a causal model $\cm_{\graphname}$ on a causal graph $\graphname$ associated with a set $\overtset$ of observed vertices. Let $\cmtele$ be a teleportation causal model on $\graphtele\in\graphfamily\graphname$ obtained from $\cm_{\graphname}$ (\cref{def:causal model of graphfamily_v3}) and $\successprob$ be the success probability of post-selection in $\cmtele$ (\cref{def:success_prob_causal_model}). If $\successprob>0$, the probability associated with the joint observed event $\outcome := \{\outcome_\vertname \in \outcomemaparg{\vertname}\}_{\vertname\in\overtset}$ in $\cm_{\graphname}$ is defined as
 \begin{align}
 \label{eq: gen_prob_rule}
     \prob\left(\outcome \right )_{\graphname}:=\probacyc\left(\outcome \middle| \{\postoutcome_i = \ok\}_{\postvertname_i\in\psvertset}\right)_{\graphtele}=\frac{\probacyc\left(\outcome, \{\postoutcome_i = \ok\}_{\postvertname_i\in\psvertset}\right)_{\graphtele}}{\successprob}.
 \end{align}
 If $\successprob=0$, we say that the causal model $\cm_{\graphname}$ is inconsistent and the probabilities $ \prob\left(\outcome\right )_{\graphname}$ are undefined.
\end{definition}

In the above definition of $  \prob\left(\outcome \right )_{\graphname}$, there are two choices involved. Firstly the choice $\graphtele\in\graphfamily\graphname$ of the teleportation graph and secondly, once $\graphtele$ is fixed, there is still freedom in the specific implementation of post-selected teleportation protocol one picks for each pair of pre- and post-selection vertices in $\graphtele$. \Cref{lemma: acyclic_prob_same_v3} proves that~\cref{def: probability distribution v3} is independent of the first choice. In the following, we show that it is also independent of the second choice. 
In fact, we show a general formula that makes this fact explicit. This involves the following composition rule, which is a special case of loop composition introduced in~\cite{Portmann_2017}.

\begin{definition}[Self-cycle composition]
\label{def:selfcycle_v3}
    Let $\hilmaparg A \cong \hilmaparg C$ with $d = \dim(\hilmap_A)=\dim(\hilmap_C)$. Further, let $\{\ket{k}_A\}_{k=1}^{d}$ and $\{\ket{l}_A\}_{l=1}^{d}$ be any orthonormal bases of $\hilmaparg A$ and $\{\ket{k}_C\}_{k=1}^{d}$ and $\{\ket{l}_C\}_{l=1}^{d}$ be the corresponding bases of $\hilmaparg{C}$ i.e., $\ket{k}_A\cong \ket{k}_C$ and $\ket{l}_A\cong \ket{l}_C$ for all $k,l=1,\dots,d$.
    Then for any linear map $\mathcal{M}_{A|C} : \linops(\hilmaparg C) \mapsto \linops(\hilmaparg A)$, we define the self-cycle composition $\selfcycle(\mathcal{M}_{A|C}) \in \mathbb{C}$ as follows
    \begin{align}
        \selfcycle(\mathcal{M}_{A|C}) = \sum_{k,l=1}^{d} \bra{k}_A \mathcal M(\ketbraa{k}{l}_C) \ket{l}_A.
    \end{align}
\end{definition}
It is easy to check that $\selfcycle(\mathcal{M}_{A|C})$ is indeed independent of the choice of orthonormal basis one makes in the above definition. 

\begin{restatable}[General probability rule in terms of self-cycle composition]{prop}{selfcycleprob}
    \label{prop:probs as self cycles_v3}
    Consider a causal model $\cm_{\graphname}$ on any causal graph $\graphname$, associated with the sets $\overtset$ and $\uvertset$ of observed and unobserved vertices. Let $\graphtele\in \graphfamily\graphname$ be a teleportation graph and $\cmtele$ a corresponding teleportation causal model on $\graphtele$ constructed from $\cm_{\graphname}$. By construction, pre and post-selection vertices in $\graphtele$ come in pairs $\prevertname_i$ and $\postvertname_i$, and are associated in particular with edges $\edgename_i:=\edgearg{\vertname_i}{\postvertname_i}$ and $\edgename_i':=\edgearg{\prevertname_i}{\vertname_i'}$. Suppose that the cardinality of the split edges set $\splitedges{\graphtele}$ is $k$, then we have the edges $\{\edgename_1,\dots,\edgename_k\}$ and $\{\edgename_1',\dots,\edgename_k'\}$ defined as above.
    We let a joint observed event associated with $\overtset$ be denoted as $\outcome := \{\outcome_\vertname \in \outcomemaparg{\vertname}\}_{\vertname\in\overtset}$, and define the collection of CP maps
    \begin{align}
        \Big\{\etot_\outcome : \textstyle\linops\left(\bigotimes_{i=1}^k \hilmaparg{\edgename_i'}\right) \mapsto \linops\left(\bigotimes_{i=1}^k \hilmaparg{\edgename_i}\right) \Big\}_{\outcome}
    \end{align} 
    through
    \begin{align}
    \label{eq:def etot}
        \etot_\outcome = \bigcomp_{\vertname\in\overtset} \measmaparg{\outcome_\vertname}{\vertname} \bigcomp_{u\in\uvertset} \chanmaparg{u},
    \end{align}
    where $\measmaparg{\outcome_\vertname}{\vertname}$ and $\chanmaparg u$ refer to the maps of the causal model $\cm_{\graphname}$ (see \cref{def:causal model}),
    and where the composition rule is dictated by the acyclic subgraph $\graphname'$ of $\graphname$ which is used in~\cref{def: graph_family_v3} to construct $\graphtele$
    \footnote{This is similar to how the composition works in \cref{def: acyclic probability}, but here specified by the chosen graph $\graphname'$.}.
    It holds that the success probability $\successprob$ and the probabilities $\prob(x)_{\graphname}$ of \cref{def: probability distribution v3}, defined if and only if $\successprob \neq 0$, satisfy
    \begin{align}
    \label{eq:cyclic probs_v2}
        \successprob &= \left(\prod_{i=1}^k \teleprob^{(i)}\right) \sum_{x} \selfcycle(\etot_x), \qquad
        \prob(x)_{\graphname} = \frac{\selfcycle(\etot_x)}{\sum_x \selfcycle(\etot_x)},
    \end{align}
    where the $\selfcycle$ operation was introduced in \cref{def:selfcycle_v3} and $\teleprob^{(i)}$ is the teleportation probability associated to the post-selected teleportation protocol implemented by the vertices $\prevertname_i$ and $\postvertname_i$ (\cref{def:causal model of graphfamily_v3}). 
\end{restatable}

The proof of the proposition above is given in~\cref{app:proofs_map}. Notice that~\cref{prop:probs as self cycles_v3} holds for any teleportation graph $\graphtele\in\graphfamily{\graphname}$. In the special case of the maximal teleportation graph~(\cref{def:max tele graph}), $\graphname_{\textsc{tp}}^{\textup{max}}\in\graphfamily{\graphname}$, the collection of CP maps of~\cref{eq:def etot} reduces to the tensor product of the maps of the causal model $\cm_{\graphname}$, i.e., 
\begin{align}
        \etot_\outcome^{\textup{max}} = \bigotimes_{\vertname\in\overtset} \measmaparg{\outcome_\vertname}{\vertname} \bigotimes_{u\in\uvertset} \chanmaparg{u}.
\end{align}

In addition, we can prove that $\selfcycle(\etot_\outcome)$ is independent of the choice of teleportation graph $\graphtele\in \graphfamily\graphname$ used in its construction. This follows from the proof of the above proposition (specifically \cref{lemma:cyclic indep of tele implementation} used therein).
\begin{corollary}
\label{cor:selfcycle etot indepentend on telegraph}
  For any causal model on an arbitrary directed graph $\graphname$, the quantity $\selfcycle(\etot_\outcome)$, defined with respect to the causal mechanisms $\{\measmaparg{\outcome_\vertname}{\vertname}\}_{\vertname\in\overtset}$ and $\{\chanmaparg{\vertname}\}_{\vertname\in\uvertset}$ of the causal model as in \cref{prop:probs as self cycles_v3}, is the same independently of the choice of teleportation graph $\graphtele\in \graphfamily\graphname$ used in its construction. 
\end{corollary}

Finally, we obtain the following corollary showing that~\cref{def: probability distribution v3} does not depend on the implementation of teleportation protocol~(\cref{def:ps teleportation}).
This immediately follows from~\cref{prop:probs as self cycles_v3,cor:selfcycle etot indepentend on telegraph}, and observing that, by definition, the success probability of a post-selected teleportation protocol is strictly greater than zero.

\begin{corollary}
\label{corollary:probs indep of tele implementation v3}
Consider a causal model $\cm_{\graphname}$ on a causal graph $\graphname$ and any causal model $\cmtele$ in the family of teleportation causal models derived from $\cm_{\graphname}$ according to \cref{def:causal model of graphfamily_v3}. 
  It holds that whether or not the probabilities $\prob(\outcome)_{\graphname}$ of $\cm_{\graphname}$ (\cref{eq: gen_prob_rule}) are defined and the values they take (if they are defined) do not depend on the choice of implementation of post-selected teleportation protocol (\cref{def:ps teleportation}) one makes in the construction of $\cmtele$ from $\cm_{\graphname}$.
\end{corollary}
Thus, we have proven that~\cref{def: probability distribution v3} does not depend on which teleportation graph one picks in the family of~\cref{def: graph_family_v3} nor on the implementation of teleportation protocol~\cref{def:ps teleportation}.

\subsection{Examples of cyclic causal graphs}
\label{sec:examples cyclic causal graphs}
In this section, we present few examples of causal models on cyclic graphs. These should clarify how to construct a family of acyclic teleportation graphs (\cref{def: graph_family_v3}), define a teleportation graph on elements of this family (\cref{def:causal model of graphfamily_v3}) and evaluate the probability distribution over observed vertices~(\cref{def: probability distribution v3}).

\paragraph{The self-cycle graph.}

We start by considering the causal graph $\graphname$ given below (which is the same as the example of \cref{eq:selfcyclegraph1}), and we choose the following teleportation graph $\graphtele\in\graphfamily\graphname$:
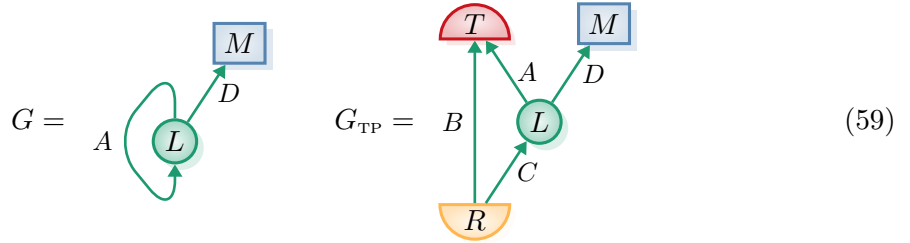
\begin{align}
    \graphname = \centertikz{
        \begin{scope}[xscale=1.7,yscale=1.3]
            \node[unode] (rho) at (0,0) {$L$};
            \node[onode] (m) at (0.5,1) {$M$};
            \draw[qleg] (rho) -- node[right] {\small$D$} (m);
            \draw[qleg,rounded corners=15pt] (rho) -- ++(0,0.8) -- ++(-0.5,-0.8) node[left=-5pt] {\small$A$} -- ++(0.5,-0.8) -- (rho);
        \end{scope}
    }
    \qquad
    \graphtele = \centertikz{
        \begin{scope}[xscale=1.7,yscale=1.3]
            \node[unode] (rho) at (0,0) {$L$};
            \node[onode] (m) at (0.5,1) {$M$};
            \draw[qleg] (rho) -- node[right] {\small$D$} (m);
            \node[psnode] (o) at (-0.5,1) {$\postvertname$};
            \node[prenode] (r) at (-0.5,-1) {$\prevertname$};
            \draw[qleg] (rho) -- node[right] {\small$A$} (o);
            \draw[qleg] (r) -- node[left] {\small$B$} (o);
            \draw[qleg] (r) -- node[right] {\small$C$} (rho);
        \end{scope}
    }
\end{align}
Consider a generic causal model on $\graphname$, namely, a channel $\chanmaparg L_{AD|A}$ and a POVM $\{(\povmelarg m M)_D\}_{m\in\outcomemaparg M}$.
We want to compute the probability $\prob(m)_{\graphname}$ as in \cref{def: probability distribution v3}.
For that, we pick the above $\graphtele$, and we pick a causal model on $\graphtele$ associated to the causal model of $\graphname$, which simply means that $L$ and $M$ get the same causal mechanisms as in $\graphname$, but we also associate a post-selected teleportation protocol $(\telepovm_{AB},\telestate_{BC})$ with teleportation probability $\teleprob$ to the pre- and post-selection vertices $\prevertname$ and $\postvertname$ of $\graphtele$.
Notice also that we set $\hilmaparg A = \hilmaparg C$, as prescribed by \cref{def:causal model of graphfamily_v3}.
The success probability of the causal model in $\graphtele$ is given by
\begin{align}
    \successprob &= \sum_{m\in\outcomemaparg M} \Tr_{ABD}\left[\left(\telepovm_{AB} \otimes (\povmelarg m M)_D\right) \chanmaparg L_{AD|C}(\telestate_{BC}) \right].
\end{align}
For concreteness, we choose the post-selected teleportation protocol to be implemented with Bell states as in \cref{def:bell tele}, thus obtaining
\begin{align}
    \successprob &= \frac{1}{\dim(\hilmaparg A)^2} \sum_{i,j=1}^{\dim(\hilmaparg A)} \sum_{k=1}^{\dim(\hilmaparg D)} \bra{i}_A \otimes \bra{k}_D \chanmaparg L_{AD|C}( \ketbraa{i}{j}_C ) \ket{j}_A \otimes \ket{k}_D \\
    &= \frac{\selfcycle(\Tr_D \chanmaparg L_{AD|C})}{\dim(\hilmaparg A)^2}
\end{align}
as expected from \cref{prop:probs as self cycles_v3}, where we can see that $\etot_m = \Tr_D[\povmelarg m M \chanmaparg L_{AD|C}] : \linops(\hilmaparg C) \mapsto \linops(\hilmaparg A)$.
If this success probability is nonzero, the causal model on $\graphname$ is non-paradoxical, and we then obtain the probabilities
\begin{align}
    \prob(m)_{\graphname} = \frac{\selfcycle(\Tr_D[\povmelarg m M \chanmaparg L_{AD|C}])}{\selfcycle(\Tr_D[\chanmaparg L_{AD|C}])}.
\end{align}

\paragraph{A two-cycle with inputs.}

We now move on to a case that displays some quantum-classical interplay.
We let
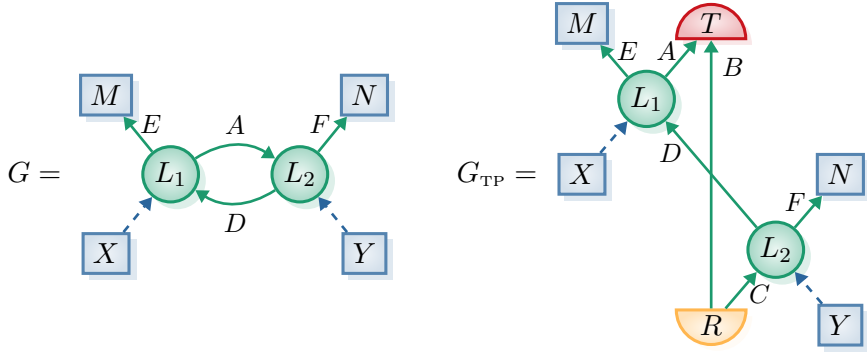
\begin{align*}
    \graphname = \centertikz{
        \begin{scope}[xscale=1.7,yscale=2.1]
            \node[unode] (l1) at (0,0) {$L_1$};
            \node[unode] (l2) at (1,0) {$L_2$};
            \node[onode] (a) at (-0.5,0.5) {$M$};
            \node[onode] (b) at (1.5,0.5) {$N$};
            \node[onode] (x) at (-0.5,-0.5) {$X$};
            \node[onode] (y) at (1.5,-0.5) {$Y$};
            \draw[qleg] (l1) to[out=30,in=150] node[above] {\small$A$} (l2);
            \draw[qleg] (l2) to[out=210,in=330] node[below] {\small$D$} (l1);
            \draw[qleg] (l1) -- node[right,pos=0.8] {\small $E$} (a);
            \draw[qleg] (l2) -- node[left,pos=0.8]  {\small $F$} (b);
            \draw[cleg] (x) -- (l1);
            \draw[cleg] (y) -- (l2);
        \end{scope}
    }
    \qquad
    \graphtele = \centertikz{
        \begin{scope}[xscale=1.7,yscale=2]
            \node[unode] (l1) at (0,0) {$L_1$};
            \node[unode] (l2) at (1,-1) {$L_2$};
            \node[onode] (a) at (-0.5,0.5) {$M$};
            \node[onode] (b) at (1.5,-0.5) {$N$};
            \node[onode] (x) at (-0.5,-0.5) {$X$};
            \node[onode] (y) at (1.5,-1.5) {$Y$};
            \node[prenode] (pre) at (0.5,-1.5) {$\prevertname$};
            \node[psnode] (ps) at (0.5,0.5) {$\postvertname$};
            \draw[qleg] (l2) -- node[left,pos=0.7] {\small $D$} (l1);
            \draw[qleg] (l1) -- node[right,pos=0.8] {\small $E$} (a);
            \draw[qleg] (l2) -- node[left,pos=0.8]  {\small $F$} (b);
            \draw[cleg] (x) -- (l1);
            \draw[cleg] (y) -- (l2);
            \draw[qleg] (l1) -- node[left,pos=0.7] {\small $A$} (ps);
            \draw[qleg] (pre) -- node[right,pos=0.9] {\small $B$} (ps);
            \draw[qleg] (pre) -- node[right,pos=0.4] {\small $C$} (l2);
        \end{scope}
    }
\end{align*}
A causal model on $\graphname$ consists of distributions $\{\probmaparg x X\}_{x\in\outcomemaparg X}$ and $\{\probmaparg y Y\}_{y\in\outcomemaparg Y}$ for $X$ and $Y$, of POVMs $\{(\povmelarg m M)_E\}_{m\in\outcomemaparg M}$ and $\{(\povmelarg n N)_F\}_{n\in\outcomemaparg N}$ for $M$ and $N$, and collections of channels $\{\chanmaparg{L_1,x}_{AE|D}\}_{x\in\outcomemaparg X}$ and $\{\chanmaparg{L_2,y}_{DF|A}\}_{y\in\outcomemaparg Y}$ for $L_1$ and $L_2$.
In \cref{sec:example acyclic graphs}, we applied \cref{def:causal model} to examples of acyclic graphs, a similar procedure applies here noting that $\graphtele$ is also acyclic.
To define the causal model on $\graphtele$ associated to the above causal model on $\graphname$, we take $\hilmaparg A = \hilmaparg C$, and we let $(\telepovm_{AB},\telestate_{BC})$ be the post-selected teleportation protocol of $R$, $T$, with teleportation probability $\teleprob$.
The success probability can then be written as 
\begin{align}
    \successprob &= \sum_{\substack{x\in\outcomemaparg X\\ y\in\outcomemaparg Y}} \sum_{\substack{m\in\outcomemaparg M\\ n\in\outcomemaparg N}}  \probmaparg x X \probmaparg y Y \Tr_{ABEF}[(\povmelarg{m}{M})_{E}\otimes\telepovm_{AB}\otimes(\povmelarg{n}{N})_{F}(\chanmaparg{L_1,x}_{AE|D}\circ\chanmaparg{L_2,y}_{DF|C})(\telestate_{BC})] \\
    &= \sum_{\substack{x\in\outcomemaparg X\\ y\in\outcomemaparg Y}}   \probmaparg x X \probmaparg y Y \Tr_{AB}[\telepovm_{AB}(\chanmaparg{L_1,x}_{A|D}\circ\chanmaparg{L_2,y}_{D|C})(\telestate_{BC})] \\
    &= \sum_{x\in\outcomemaparg X}\sum_{y\in\outcomemaparg Y} \probmaparg x X \probmaparg y Y \selfcycle\left(\chanmaparg{L_1,x}_{A|D}\circ\chanmaparg{L_2,y}_{D|C}\right),
\end{align}
where we defined 
\begin{equation}
\begin{split}
    \chanmaparg{L_1,x}_{A|D} = \sum_{m\in\outcomemaparg M}\Tr_E\left( (\povmelarg{m}{M})_E\chanmaparg{L_1,x}_{AE|D}\right) = \Tr_E\left(\chanmaparg{L_1,x}_{AE|D}\right),\\
    \chanmaparg{L_2,y}_{D|C} = \sum_{n\in\outcomemaparg N}\Tr_F\left( (\povmelarg{n}{N})_{F}\chanmaparg{L_2,y}_{DF|C}\right)= \Tr_F\left(\chanmaparg{L_2,y}_{DF|C}\right),
\end{split} 
\end{equation}
and used that $\sum_{m\in\outcomemaparg M}\povmelarg{m}{M} = \id$ and $\sum_{n\in\outcomemaparg N} \povmelarg{n}{N} = \id$ by the definition of POVMs.

If the success probability is non-zero, the causal model is non-paradoxical, and we obtain the following probabilities:
\begin{align}
    \prob(m,n,x,y)_{\graphname} = \frac{
        \probmaparg x X \probmaparg y Y \selfcycle\left( \Tr_E[(\povmelarg m M)_E \chanmaparg{L_1,x}_{AE|D}] \circ \Tr_F[(\povmelarg n N)_F \chanmaparg{L_2,y}_{DF|C}] \right)
    }{
        \sum_{x\in\outcomemaparg X}\sum_{y\in\outcomemaparg Y} \probmaparg x X \probmaparg y Y \selfcycle\left(\chanmaparg{L_1,x}_{A|D}\circ\chanmaparg{L_2,y}_{D|C}\right)
    }
\end{align}
In particular, observe that the marginal $\prob(x)_{\graphname}$ is in general not equal to the distribution of the causal mechanism $\probmaparg x X$.
Furthermore, we have $\prob(x,y)_{\graphname} \neq \probmaparg x X \probmaparg y Y$, and in particular, $\prob(x,y)_{\graphname} \neq \prob(x)_{\graphname}\prob(y)_{\graphname}$.
This means that the causal model yields a distribution that features correlations between $X$ and $Y$.
As we will see in the next section, these correlations are generated by requiring logical consistency of the model, which effectively post-selects over the variables in the loop which are not paradoxical. This feature is relevant in understanding the failure of the soundness of existing graph-separation criteria like $d$-separation as well as the intuition behind the new property, $p$-separation that we propose in the following section.

\section{A sound and complete graph-separation property for cyclic QCMs}
\label{sec: introducing pseparation}
A powerful feature of the causal modelling approach is that by providing a graph-theoretic framework for describing information-theoretic circuits and networks, it allows us to infer properties of observable data from the topology of the underlying causal graph. Specifically, $d$-separation is an important graph separation property, originally introduced in the classical causal modelling literature for acyclic graphs and subsequently generalised to quantum and post-quantum theories. It is useful as the associated $d$-separation theorem shows that we can use the property to ``read off'' conditional independences in the outcome probabilities of any causal model on an acyclic graph, through the connectivity of the graph \cite{Verma1990, Geiger1990, Pearl_2009, Henson_2014}. However, this theorem has been known to fail even in classical causal models on cyclic graphs, and a general graph separation property applicable to all finite dimensional (or with finite-cardinality variables in the classical case) causal models on cyclic graphs has been lacking. 

In this section, we apply our cyclic causal modelling framework to propose such a graph separation property, which we call $p$-separation. We begin by discussing $d$-separation for acyclic graphs and examples indicating its failure in cyclic scenarios before introducing $p$-separation and discussing its applications.

\subsection{Previous results for acyclic case: $d$-separation}
\label{sec: rev dseparation}
We motivate the concept of $d$-separation and then provide the general definition. Consider the following three types of simple directed graphs. Although $d$-separation is a purely graph-theoretic property defined for any directed graph, for explaining the motivation behind this concept, it is useful to consider classical causal models and probabilities on the associated vertices. We therefore stylize all vertices as observed $\centertikz{\node[onode]{v};}$ and all edges as classical edges $\centertikz{
        \draw[cleg] (0,0) -- (0.6,0);
        }$ for now.

\begin{equation}
\label{eq: chain}
 \text{Chain: } \centertikz{  \node[onode] (a) at (-1.5,0) {$A$};  \node[onode] (c) at (0,0) {$C$};  \node[onode] (b) at (1.5,0) {$B$}; \draw[cleg] (a) --(c); \draw[cleg] (c)--(b);}
\end{equation}
\begin{equation}
\label{eq: fork}
 \text{Fork: }  \centertikz{  \node[onode] (a) at (-1.5,0) {$A$};  \node[onode] (c) at (0,0) {$C$};  \node[onode] (b) at (1.5,0) {$B$}; \draw[cleg] (c) --(a); \draw[cleg] (c)--(b);}
\end{equation}
\begin{equation}
\label{eq: collider}
 \text{Collider: } \centertikz{  \node[onode] (a) at (-1.5,0) {$A$};  \node[onode] (c) at (0,0) {$C$};  \node[onode] (b) at (1.5,0) {$B$}; \draw[cleg] (a) --(c); \draw[cleg] (b)--(c);}
\end{equation}
Suppose that the classical vertices $A$, $B$ and $C$ are each associated with a random variable of the same name.
For causal models on the chain and the fork, we would generally expect $A$ and $B$ to get correlated, through the (indirect) causal influence of $A$ on $B$ in the first case and through the common cause $C$ in the second case. However, we would expect that $A$ and $B$ would become independent conditioned on $C$ as $C$ mediates the correlations between $A$ and $B$ in both cases. In the chain and fork, we would say $A$ is $d$-connected to $B$, denoted $A\not\perp^d B$ while $A$ is $d$-separated from $B$ given $C$, denoted $A\perp^d B|C$. For the collider on the other hand, $A$ and $B$ have no prior causes and we would expect them to be uncorrelated and we would say $A\perp^d B$. However, conditioning on their common child $C$ amounts to post-selection and can correlate $A$ and $B$, i.e., they become $d$-connected conditioned on $C$, $A\not\perp^d B|C$. In the case of the collider, if we further had $C\rightarrow D$ as shown below, then we would also expect $A$ and $B$ to get correlated given $D$ (as it is in the future of both), and expect $A\not\perp^d B|D$ also for descendants $D$ of a collider $C$. 

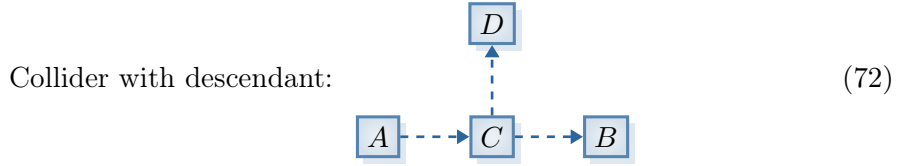
\begin{equation}
\label{eq: collider_desc}
 \text{Collider with descendant: }  \centertikz{  \node[onode] (a) at (-1.5,0) {$A$};  \node[onode] (c) at (0,0) {$C$};  \node[onode] (b) at (1.5,0) {$B$};  \node[onode] (d) at (0,1.5) {$D$};\draw[cleg] (a) --(c); \draw[cleg] (c)--(b); \draw[cleg] (c)--(d);}
\end{equation}

This is the intuition behind the following definitions. Although the intuition comes from thinking of the correlations, the definition is purely graph-theoretic. A subsequent theorem links this definition to conditional independences between the classical outcomes in a quantum causal model, which recovers the above intuition for the simple classical examples. 

\begin{definition}[Blocked paths]
Let $\graphname$ be a directed graph in which $V_1$, $V_2$ and $V_3$ are disjoint sets of vertices with the former two being non-empty.  A path (not necessarily directed) from $V_1$ to $V_2$ is
said to be \emph{blocked} by $V_3$ if it contains, for some vertices $A$ and $B$ in the path, either $A\rightarrow W\rightarrow B$
with $W\in V_3$, $A\leftarrow W\rightarrow B$
with $W\in V_3$ or $A\rightarrow W\leftarrow B$ such that neither the vertex $W$ nor any descendant of $W$ belongs to $V_3$.
\end{definition}

\begin{definition}[$d$-separation]
\label{def: d-sep}
Let $\graphname$ be a directed graph in which $V_1$, $V_2$ and $V_3$ are disjoint
sets of vertices with $V_1$ and $V_2$ non-empty.  $V_1$ and $V_2$ are \emph{$d$-separated} by $V_3$ in
$\graphname$, denoted as $(V_1\perp^d V_2|V_3)_{\graphname}$ if for every pair of vertices in $V_1$ and $V_2$ there is no path between them, or if every path from a vertex in $V_1$ to a vertex in
$V_2$ is \emph{blocked} by $V_3$. Otherwise, $V_1$ is said to be \emph{$d$-connected} with $V_2$ given $V_3$, denoted as $(V_1\not\perp^d V_2|V_3)_{\graphname}$.
\end{definition}

\begin{definition}[Conditional independence]
\label{def:conditional independence}
    Let $V$ be a non-empty finite set, and
    let $\mathcal{P}(x)$ be a joint probability distribution over a set $X =
    \{X_v\}_{v \in V}$ of finite-cardinality random variables, whose values are
    denoted $x = \{x_v\}_{v \in V}$.
    Let $V_1$, $V_2$ and $V_3$ be three disjoint subsets of $V$, with $V_1$ and $V_2$ being non-empty. 
    We denote the corresponding sets of random variables as $X_i = \{ X_v \}_{v\in V_i}$ and the corresponding values as $x_i = \{ x_v \}_{v\in V_i}$ for $i \in \{1,2,3\}$.
    We say that $X_1$ is conditionally independent of $X_2$ given $X_3$ and denote it as $(X_1\indep X_2|X_3)_{\mathcal{P}}$ if, for all $x_1, x_2,x_3$, it holds that $\mathcal{P} (x_1,x_2|x_3)=\mathcal{P}(x_1|x_3)\mathcal{P}(x_2|x_3)$.
\end{definition}

Then the next theorem follows from the theory-independent $d$-separation theorem of~\cite{Henson_2014}, when restricted to the case of quantum theory.

\begin{theorem}[$d$-separation theorem for acyclic graphs]
\label{theorem: dsep theorem}
Consider a directed acyclic graph $\graphname$ and let $V_1$, $V_2$ and $V_3$ be any three disjoint sets of the vertices of $\graphname$ with $V_1$ and $V_2$ being non-empty. 
    Then, the following holds:
    \begin{itemize}
        \item[]\textbf{\textup{(Soundness)}} For any causal model $\cm_{\graphname}$ on $\graphname$ where the sets $V_i$ are observed, we have that $d$-separation between the vertex sets $V_i$ implies conditional independence for the corresponding sets of random variables $X_i:=\{X_{\vertname}\}_{\vertname\in V_i}$ where $i\in\{1,2,3\}$, i.e.,
        \begin{equation}
         (V_1\perp^d V_2|V_3)_{\graphname} \implies (X_1\indep X_2|X_3)_{\probfacyc_{\graphname} }.
        \end{equation} 
        \item[]\textbf{\textup{(Completeness)}} If the $d$-connection $(V_1\not\perp^d V_2|V_3)_{\graphname}$ holds in $\graphname$, then there exists a causal model  $\cm_{\graphname}$ such that the sets $V_i$ are observed and ${(X_1\not \indep X_2|X_3)_{\probacyc_{\graphname}}}$..
    \end{itemize}
    The above conditional (in)dependence statements are relative to the marginal $\probacyc(x_1,x_2,x_3)_\graphname$ on $X_1\cup X_2\cup X_3$, where $x_i = \{x_\vertname \in \outcomemaparg{\vertname}\}_{\vertname\in V_i}$, of the observed distribution $\probacyc(\outcome)_{\graphname}$, where $\outcome=\{\outcome_{\vertname}\}_{\vertname\in \vertset}$, in the causal model $\fcm_{\graphname}$.
\end{theorem}

\paragraph{Example: Bell scenario.}
We illustrate these concepts and the conventions of our framework, through their application to the familiar example of a Bell scenario (see also~\cref{example: Bell scenario}). This is represented through the following graph

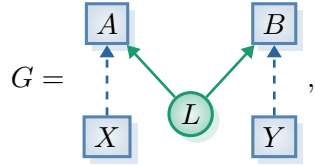
\begin{equation*}
\graphname = 
    \centertikz{
    \begin{scope}[xscale=2.2,yscale=1.5]
        \node[onode] (x) at (0,0) {$X$};
        \node[onode] (y) at (1,0) {$Y$};
        \node[onode] (a) at (0,1) {$A$};
        \node[onode] (b) at (1,1) {$B$};
        \node[unode] (l) at (0.5,0.2) {$L$};
        \draw[cleg] (x) -- (a);
        \draw[cleg] (y) -- (b);
        \draw[qleg] (l) -- (a);
        \draw[qleg] (l) -- (b);
    \end{scope}
    },
\end{equation*}
representing two parties, Alice and Bob, who perform measurements on sub-systems of a shared bipartite system $L$. The exogenous vertices $X$ and $Y$ model the classical settings of Alice and Bob and are associated with channels that prepare a prior distribution on the associated classical variables $x$ and $y$ (see~\cref{sec:example acyclic graphs} for more details). 

We have $\overtset=\{X,Y,A,B\}$ and the outcome distribution is $\probacyc(x,y,a,b)_{\graphname}$. We observe the following $d$-separations between observed vertices in $\graphname$: $(X\perp^d B|Y)_{\graphname}$ and $(Y\perp^d A|X)_{\graphname}$ due to the absence of collider-free unblocked paths between the setting of one party and the outcome of another. These imply corresponding conditional independences $(X\indep B|Y)_{\probacyc_{\graphname}}$ and $(Y\indep A|X)_{\probacyc_{\graphname}}$ in the distribution $\probacyc(x,y,a,b)$, which are precisely the two bipartite non-signalling conditions of the Bell scenario, i.e.,

$$(X\perp^d B|Y)_{\graphname} \quad\Rightarrow \quad \probacyc(a|x,y)_{\graphname}= \probacyc(a|x)_{\graphname},$$
$$(Y\perp^d A|X)_{\graphname} \quad\Rightarrow \quad  \probacyc(b|x,y)_{\graphname}= \probacyc(b|y)_{\graphname}.$$

\subsection{Failure of $d$-separation in cyclic graphs}
\label{sec:failure_dsep}

\Cref{theorem: dsep theorem} can fail already for classical causal models on cyclic graphs. We provide a simple example, based on what is already known in the classical causal modelling literature~\cite{Pearl_2009,Neal_2000}. As this is a classical example, we use a causal graph with observed vertices and classical edges. Consider the following cyclic causal graph:
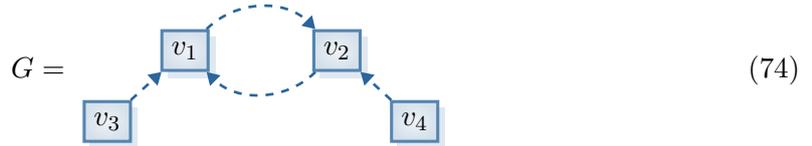
\begin{equation}
\label{eq: dsep_cycle_example}
    \graphname=\centertikz{
    \node[onode] (v1) at (0,0) {$\vertname_1$};
    \node[onode] (v2) at (2,0) {$\vertname_2$};
    \node[onode] (v3) [below left  =\chanvspace of v1] {$\vertname_3$};
    \node[onode] (v4) [below right  =\chanvspace of v2] {$\vertname_4$};
        \draw[cleg] (v1) to [out=45,in=135]  (v2);
    \draw[cleg] (v2) to [out=-135,in=-45] (v1);
 \draw[cleg] (v3) -- (v1);  \draw[cleg] (v4) -- (v2);  
    }
\end{equation}
Clearly we have $(\vertname_3\perp^d \vertname_4)_{\graphname}$. However, it is possible to construct a simple classical causal model (more precisely, a functional model, see~\cref{sec:classical aspects} and section 2 of \cite{Sister_paper}) where the associated outcomes $\outcome_3$ and $\outcome_4$ are correlated in the resulting probability distribution. 

A functional model (see \cref{sec:classical aspects}) involves specifying a prior distribution for every exogenous vertex along with a function for each vertex that determines the value of the variable of the vertex given values of its parental variables.
Here, consider binary outcome variables, $\outcome_i$ associated to $\vertname_i$, for all four vertices and the following functional dependencies: the outcome associated to $\vertname_1$ depends on its parents as $\outcome_1=\outcome_2\oplus \outcome_3$ and $\outcome_2$, associated to $\vertname_2$, depends on its parents as $\outcome_2=\outcome_1\oplus \outcome_4$. Here $\oplus$ denotes modulo 2 addition. It is easy to see that $\outcome_3=\outcome_4$ is the only consistent solution to this causal model, independently of the priors $\probex{\vertname_3}(\outcome_3)$ and $\probex{\vertname_4}(\outcome_4)$ specified on these vertices in the functional model. Indeed, computing the probability through our rule, it can be checked that we would obtain $\prob(\outcome_3,\outcome_4)_{\graphname}=0$ whenever $\outcome_3\neq \outcome_4$ indicating the perfect correlation.

Notice that the probability for the variables associated to exogenous vertices (computed through~\cref{def: probability distribution v3}) does not equal the prior distribution, $$\prob(\outcome_3,\outcome_4)_{\graphname}\neq \probex{\vertname_3}(\outcome_3)\probex{\vertname_4}(\outcome_4).$$ As a consequence of the $d$-separation theorem (\cref{theorem: dsep theorem}) the equality is always true in acyclic causal models. Hence, this non-equivalence demonstrates a violation of the $d$-separation theorem (specifically, the soundness of $d$-separation) in this cyclic causal model as $\outcome_3$ and $\outcome_4$ are correlated despite being $d$-separated.

Looking carefully, one can identify the reason for the failure of the $d$-separation property: The loop between $\vertname_1$ and $\vertname_2$ effectively acts as a collider for the exogenous vertices $\vertname_3$ and $\vertname_4$. Recall that conditioning on a collider can $d$-connect previously $d$-separated vertices. Here, the collider is not explicitly conditioned upon in the original cyclic graph, however the consistency conditions of the model impose an effective post-selection on the values of the loop variables. 
This intuition is already incorporated in how probabilities are computed in our formalism (\cref{def: probability distribution v3}), which maps probability computations in cyclic graphs to corresponding computations in acyclic graphs with post-selection, namely teleportation graphs. In each of the acyclic teleportation graphs, $\vertname_3$ and $\vertname_4$ are in fact $d$-connected when additionally conditioning on the post-selection vertices. 
The correlation between $\outcome_3$ and $\outcome_4$ (without $d$-connection) in the cyclic graph here is due to the fact that in the corresponding acyclic model (on any of the associated teleportation graphs) the post-selection success probability $\successprob$ is zero whenever $\outcome_3\neq \outcome_4$. This forces $\outcome_3=\outcome_4$ in the post-selected distribution, which defines the probability of our original causal model. We illustrate this point more explicitly in the following section where we introduce a new graph separation property, $p$-separation, that generalises $d$-separation.

We note that the example provided here for the violation of the soundness of $d$-separation in cyclic causal models involves a scenario where the model is inconsistent (in the classical case here, the functional dependencies admit no solutions) for certain values of the exogenous variables $v_3$ and $v_4$. Such functional causal models are called non-uniquely solvable. It is important to note that there exists an example \cite{Neal_2000} of a uniquely solvable functional models defined on finite, binary variables where the soundness of $d$-separation fails, in that case, the associated graph is more complex than the simple one considered here. As we will see in the next section, the soundness and completeness of our new graph-separation property $p$-separation, will apply to such examples as well.
\subsection{Introducing $p$-separation: soundness and completeness}
\label{sec:psep_introduced}

The example from the previous sub-section provides an intuition for why $d$-separation fails in cyclic graphs. It also suggests that the properties of correlations in cyclic graphs are more naturally captured if, rather than looking at $d$-separation in the original cyclic graph $\graphname$, we consider $d$-separation in a corresponding acyclic graph $\graphtele\in \graphfamily\graphname$ while also conditioning on the graph's post-selection vertices. This motivates the introduction of a new graph separation property, called $p$-separation, where $p$ stands for post-selection. As with the case of $d$-separation, $p$-separation is also purely a property of a directed graph. 

\begin{restatable}[$p$-separation]{definition}{psep}
\label{def: p-separation}
Let $\graphname$ be a directed graph and $V_1$, $V_2$ and $V_3$ denote any three disjoint subsets of the vertices of $\graphname$ with $V_1$ and $V_2$ being non-empty. Then, we say that $V_1$ is $p$-separated from $V_2$ given $V_3$ in $\graphname$, denoted $(V_1\perp^p V_2|V_3)_{\graphname}$, if and only if there exists $\graphtele\in \graphfamily{\graphname}$ (\cref{def: graph_family_v3}) such that $(V_1\perp^d V_2|V_3\cup\psvertset)_{\graphtele}$, where $\perp^d$ denotes $d$-separation and $\psvertset$ denotes the set of all post-selection vertices in the teleportation graph $\graphtele\in \graphfamily{\graphname}$. 
 Otherwise, we say that $V_1$ is $p$-connected to $V_2$ given $V_3$ in $\graphname$, and we denote it $(V_1\not\perp^p V_2|V_3)_{\graphname}$.
 To summarize,
 \begin{equation}
 \begin{aligned}
     \text{$p$-separation: } (V_1\perp^pV_2|V_3)_\graphname \equiva \exists \graphtele \in \graphfamily{\graphname} \st (V_1\perp^d V_2|V_3 \cup \psvertset)_{\graphtele}, \\
     \text{$p$-connection: } (V_1\not\perp^pV_2|V_3)_\graphname \equiva \forall \graphtele \in \graphfamily{\graphname} \st (V_1\not\perp^d V_2|V_3 \cup \psvertset)_{\graphtele}.
 \end{aligned}
 \end{equation}
\end{restatable}
Notice that $p$-separation is defined as a property of a graph, not necessarily decorated as a causal graph.
The following theorem establishes that $p$-separation is a sound and complete graph separation property for cyclic quantum causal models within our framework.

\begin{restatable}[$p$-separation theorem]{theorem}{pseptheorem}
\label{theorem: psep_theorem}
Consider a directed graph $\graphname$ and let $V_1$, $V_2$ and $V_3$ be any three disjoint sets of the vertices of $\graphname$ with $V_1$ and $V_2$ being non-empty. 
    Then, the following holds:
    \begin{itemize}
        \item[]\textbf{\textup{(Soundness)}} For any causal model $\cm_{\graphname}$ on $\graphname$ where the sets $V_i$ are observed, we have that $p$-separation between the vertex sets $V_i$ implies conditional independence for the corresponding sets of random variables $X_i:=\{X_{\vertname}\}_{\vertname\in V_i}$ where $i\in\{1,2,3\}$, i.e.,
        \begin{equation}
         (V_1\perp^p V_2|V_3)_{\graphname} \implies (X_1\indep X_2|X_3)_{\prob_{\graphname} }.
        \end{equation} 
      
        \item[]\textbf{\textup{(Completeness)}} If the $p$-connection $(V_1\not\perp^p V_2|V_3)_{\graphname}$ holds in $\graphname$, then there exists a causal model  $\cm_{\graphname}$ such that the sets $V_i$ are observed and ${(X_1\not \indep X_2|X_3)_{\prob_{\graphname}}}$.
    \end{itemize}
    The above conditional (in)dependence statements are relative to the marginal $\prob(x_1,x_2,x_3)_\graphname$ on $X_1\cup X_2\cup X_3$, where $x_i = \{x_\vertname \in \outcomemaparg{\vertname}\}_{\vertname\in V_i}$ of the observed distribution $\prob(\outcome)_{\graphname}$, where $\outcome=\{\outcome_{\vertname}\}_{\vertname\in \vertset}$, in the causal model $\cm_{\graphname}$.
\end{restatable}

Soundness and completeness of $p$-separation are proven in~\cref{app:pseparation}.

\subsection{Examples: $p$-separation in acyclic and cyclic graphs} 
\label{sec:examples psep}

\paragraph{Recovering $d$-separation in acyclic graphs}
Notice that our definition of $p$-separation reduces to $d$-separation whenever the causal graph $\graphname$ is acyclic. This is because in this case the graph is a representative of its own acyclic family, $\graphname\in \graphfamily\graphname$ with an empty set of post-selection vertices $\psvertset=\emptyset$. Thus we immediately have that $d$-separation of $V_1$ from $V_2$ given $V_3$ implies $p$-separation of $V_1$ from $V_2$ given $V_3$ in $\graphname$, according to \cref{def: p-separation}. The other direction, namely that $p$-separation implies $d$-separation is also true from the contrapositive $d$-connection implies $p$-connection. Indeed, if $V_1$ and $V_2$ are $d$-connected conditioned on $V_3$ in a given graph, the same $d$-connection will hold in all its teleportation graphs once we additionally condition on the post-selection vertices. Thus, $V_1$ and $V_2$ are $p$-connected conditioned on $V_3$.

To further motivate the definition of $p$-separation (\cref{def: p-separation}), we now consider some examples. These highlight why it is natural to define $p$-separation through the $\exists$ $\graphtele\in \graphfamily\graphname$ rather than $\forall$ $\graphtele\in \graphfamily\graphname$, specifically for recovering $d$-separation in the acyclic case. For this, consider the chain, fork and collider graphs of \cref{eq: chain,eq: fork} and \cref{eq: collider} respectively, and consider the graph $\graphname_{\textsc{tp}}^{\textup{max}}\in \graphfamily\graphname$ obtained by adding pre and post-selection vertices to every edge (which leads to the maximal number of such vertices for the given graphs). Then, we obtain the following graphs:
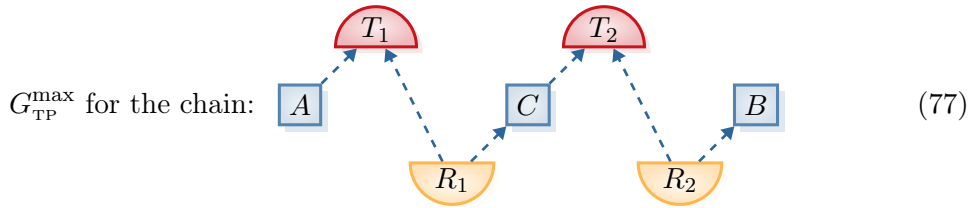
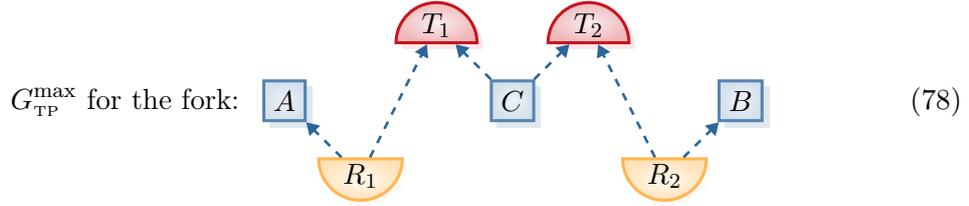
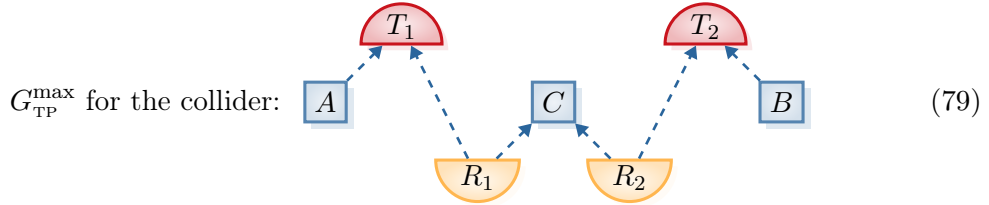
\begin{equation}
\label{eq: chain_PS}
\graphname_{\textsc{tp}}^{\textup{max}} \text{ for the chain: } \centertikz{
        \node[onode] (A) at (0,0) {$A$};
        \node[psnode] (P1) at (1,1) {$\postvertname_1$};
        \node[prenode] (Q1) at (2,-1) {$\prevertname_1$};
        \node[onode] (C) at (3,0) {$C$};
        \node[psnode] (P2) at (4,1) {$\postvertname_2$};
        \node[prenode] (Q2) at (5,-1) {$\prevertname_2$};
         \node[onode] (B) at (6,0) {$B$};
        \draw[cleg] (A) -- (P1); \draw[cleg] (Q1) -- (P1); \draw[cleg] (Q1) -- (C);
        \draw[cleg] (C) -- (P2); \draw[cleg] (Q2) -- (P2); \draw[cleg] (Q2) -- (B);
    }
\end{equation}

\bigskip
\begin{equation}
\label{eq: fork_PS}
\graphname_{\textsc{tp}}^{\textup{max}} \text{ for the fork: } \centertikz{
        \node[onode] (A) at (0,0) {$A$};
        \node[prenode] (Q1) at (1,-1) {$\prevertname_1$};
        \node[psnode] (P1) at (2,1) {$\postvertname_1$};
        \node[onode] (C) at (3,0) {$C$};
        \node[psnode] (P2) at (4,1){$\postvertname_2$};
        \node[prenode] (Q2) at (5,-1) {$\prevertname_2$};
         \node[onode] (B) at (6,0) {$B$};
        \draw[cleg] (C) -- (P1); \draw[cleg] (Q1) -- (P1); \draw[cleg] (Q1) -- (A);
        \draw[cleg] (C) -- (P2); \draw[cleg] (Q2) -- (P2); \draw[cleg] (Q2) -- (B);
    }
\end{equation}

\bigskip
\begin{equation}
\label{eq: collider_PS}
\graphname_{\textsc{tp}}^{\textup{max}} \text{ for the collider: } \centertikz{
       \node[onode] (A) at (0,0) {$A$};
        \node[psnode] (P1) at (1,1) {$\postvertname_1$};
        \node[prenode] (Q1) at (2,-1) {$\prevertname_1$};
        \node[onode] (C) at (3,0) {$C$};
        \node[prenode] (Q2) at (4,-1) {$\prevertname_2$};
        \node[psnode] (P2) at (5,1) {$\postvertname_2$};
         \node[onode] (B) at (6,0) {$B$};
        \draw[cleg] (A) -- (P1); \draw[cleg] (Q1) -- (P1); \draw[cleg] (Q1) -- (C);
        \draw[cleg] (B) -- (P2); \draw[cleg] (Q2) -- (P2); \draw[cleg] (Q2) -- (C);
    }
\end{equation}
Recall that we had $(A\not\perp^d B)_{\graphname}$ and $(A\perp^d B|C)_{\graphname}$ in the original chain and fork. Denoting $\psvertset:=\{\postvertname_1,\postvertname_2\}$, we see that we have $(A\not\perp^d B|\psvertset)_{\graphname_{\textsc{tp}}^{\textup{max}}}$ and $(A\perp^d B|C,\psvertset)_{\graphname_{\textsc{tp}}^{\textup{max}}}$ in the above graphs for the chain and fork. For the collider, recall that we had $(A\perp^d B)_{\graphname}$ and $(A\not\perp^d B|C)_{\graphname}$, and we can see that we also have $(A\perp^d B|\psvertset)_{\graphname_{\textsc{tp}}^{\textup{max}}}$ and $(A\not\perp^d B|C,\psvertset)_{\graphname_{\textsc{tp}}^{\textup{max}}}$. That is, for the graphs of the chain, fork and collider without descendants, the $d$-separation in the original graph $\graphname$ matches $d$-separation with conditioning on the post-selection vertices in the representative graph $\graphname_{\textsc{tp}}^{\textup{max}}$. However, an interesting difference occurs when we consider colliders with descendants as in \cref{eq: collider_desc}. The graph $\graphname_{\textsc{tp}}^{\textup{max}}$ for the collider with descendant is 
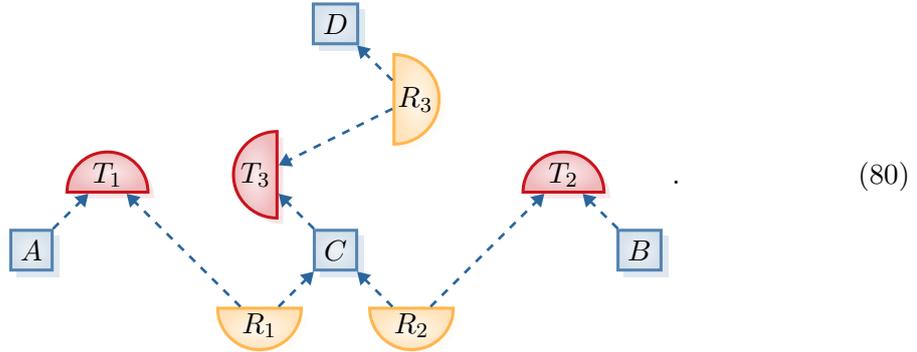
\begin{equation}
\label{eq: collider_desc_PS}
\centertikz{
  \node[onode] (A) at (-1,0) {$A$};
        \node[psnode] (P1) at (0,1) {$\postvertname_1$};
        \node[prenode] (Q1) at (2,-1) {$\prevertname_1$};
        \node[onode] (C) at (3,0) {$C$};
        \node[prenode] (Q2)at (4,-1) {$\prevertname_2$};
        \node[psnode] (P2) at (6,1) {$\postvertname_2$};
         \node[onode] (B) at (7,0) {$B$};
 \node[psnode, rotate=90] (P3) at (2,1) {\phantom{$\postvertname_{33}$}};
 \node at (1.95,1) {$\postvertname_3$};
  \node[prenode, rotate=90] (Q3) at (4,2) {\phantom{$\prevertname_{33}$}};
   \node at (4.05,2) {$\prevertname_3$};
   \node[onode] (D)at (3,3) {$D$};
         
        \draw[cleg] (A) -- (P1); \draw[cleg] (Q1) -- (P1); \draw[cleg] (Q1) -- (C);
        \draw[cleg] (B) -- (P2); \draw[cleg] (Q2) -- (P2); \draw[cleg] (Q2) -- (C);
         \draw[cleg] (C) -- (P3); \draw[cleg] (Q3) -- (P3); \draw[cleg] (Q3) -- (D);
    }.
\end{equation}
The $d$-separation $(A\perp^d B)_{\graphname}$ still holds in the collider with descendant, even though $(A\perp^d B|\psvertset)_{\graphname_{\textsc{tp}}^{\textup{max}}}$ no longer holds here for $\psvertset:=\{\postvertname_1,\postvertname_2,\postvertname_3\}$ since $\postvertname_3$ is now a descendant of the collider which we condition on. However, notice that $\graphname_{\textsc{tp}}^{\textup{max}}$ is just one representative of $\graphfamily\graphname$. If we consider a teleportation graph $\graphtele\in\graphfamily\graphname$ where the edge $C\rightarrow D$ does not receive any post-selection vertices, then $(A\perp^d B|\psvertset)_{\graphname_{\textsc{tp}}^{\textup{max}}}$ will still hold. Therefore, by \cref{def: p-separation}, $(A\perp^p B)_{\graphname}$ holds in this case.

In other words, our definition only implies a $p$-connection in $\graphname$ when that connection is reflected in \emph{all} the graphs of the graph family $\graphfamily\graphname$. Indeed, we have shown (see \cref{lemma: acyclic_prob_same_v3} and \cref{def: probability distribution v3}) that the observed probabilities for the causal model are independent of the representative  of $\graphfamily\graphname$ chosen in computing them. It is easy to check that applying our probability rule, the outcomes $a$ and $b$ of the vertices $A$ and $B$ will be conditionally independent even when conditioned on the colliders $\psvertset:=\{\postvertname_1,\postvertname_2,\postvertname_3\}$ in \cref{eq: collider_desc_PS}, this is because the post-selection on the unblocking collider $\postvertname_3$ is fine-tuned such that it simulates a directed edge from $C$ to $D$ through an identity channel (without post-selecting a particular outcome value on $C$ or $D$). Thus our definition of $p$-separation serves to avoid such fine-tuning and enables us to construct a sound and complete graph separation property for general directed graphs that reduces to $d$-separation in the acyclic case.

\paragraph{Example of \cref{sec:failure_dsep}}
Let us now apply the results of this section to the cyclic graph of \cref{eq: dsep_cycle_example}, to see how $p$-separation captures the failure of $d$-separation discussed in the previous sub-section. We consider two representatives of $\graphtele_{,1}, \graphtele_{,2}\in \graphfamily\graphname$ for $\graphname$ given by \cref{eq: dsep_cycle_example}. These are given by introducing pre- and post-selection vertices in one or both edges in the cycle. 

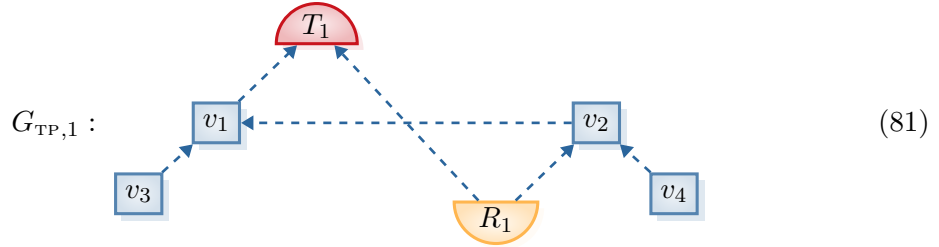
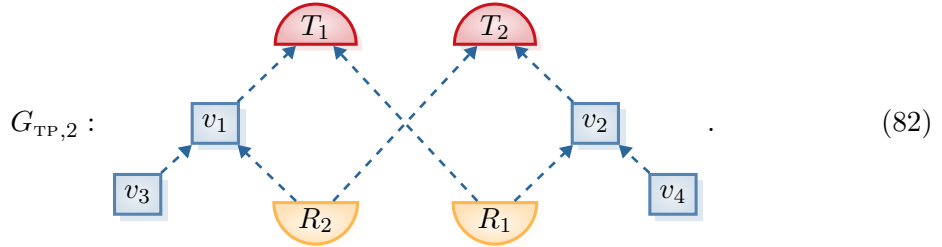
\begin{equation}
 \graphtele_{,1}:   \centertikz{
    \node[onode] (v1) at (0,0) {$\vertname_1$};
 \node[psnode] (p1) [above right =2*\chanvspace of v1] {$\postvertname_1$};
  
    \node[onode] (v2) at (5,0) {$\vertname_2$};

     \node[prenode] (q1) [below left =2*\chanvspace of v2] {$\prevertname_1$};

     \node[onode] (v3) [below left  =\chanvspace of v1] {$\vertname_3$};
   \node[onode] (v4) [below right  =\chanvspace of v2] {$\vertname_4$};
        \draw[cleg] (v1) -- (p1);  \draw[cleg] (q1) -- (p1);  \draw[cleg] (q1) -- (v2);  
               \draw[cleg] (v2) -- (v1);

 \draw[cleg] (v3) -- (v1);  \draw[cleg] (v4) -- (v2);  
    }
\end{equation}

\begin{equation}
 \graphtele_{,2}:   \centertikz{
    \node[onode] (v1) at (0,0) {$\vertname_1$};
 \node[psnode] (p1) [above right =2*\chanvspace of v1] {$\postvertname_1$};
  
    \node[onode] (v2) at (5,0) {$\vertname_2$};

     \node[prenode] (q1) [below left =2*\chanvspace of v2] {$\prevertname_1$};  
 \node[psnode] (p2) [above left =2*\chanvspace of v2] {$\postvertname_2$};
 \node[prenode] (q2) [below right =2*\chanvspace of v1] {$\prevertname_2$};  
    
     \node[onode] (v3) [below left  =\chanvspace of v1] {$\vertname_3$};
   \node[onode] (v4) [below right  =\chanvspace of v2] {$\vertname_4$};
        \draw[cleg] (v1) -- (p1);  \draw[cleg] (q1) -- (p1);  \draw[cleg] (q1) -- (v2);  
               \draw[cleg] (v2) -- (p2);  \draw[cleg] (q2) -- (p2);  \draw[cleg] (q2) -- (v1);

 \draw[cleg] (v3) -- (v1);  \draw[cleg] (v4) -- (v2);  
    }.
\end{equation}
Notice that we have $(\vertname_3\not\perp^d \vertname_4|\postvertname_1)_{\graphtele_{,1}}$ and $(\vertname_3\not\perp^d \vertname_4|\postvertname_1,\postvertname_2)_{\graphtele_{,2}}$. It is easy to see that this $d$-connection would hold in all other teleportation graphs in this case, as all of them would involve at least one of the post-selection vertices $\postvertname_1$, $\postvertname_2$. Thus we have the $p$-connection, $(\vertname_3\perp^p \vertname_4)_\graphname$, which justifies why we would expect the outcomes of $\vertname_3$ and $\vertname_4$ to get correlated when we have a causal model on the original graph \cref{eq: dsep_cycle_example}, as was the case in the example of \cref{sec:failure_dsep}.
Indeed, probabilities of that model are nothing but conditional probabilities of a corresponding model on $\graphtele_{,1}$ with conditioning on $\postvertname_1$ or equivalently of a model on $\graphtele_{,2}$ with conditioning on $\{\postvertname_1,\postvertname_2\}$.

Notice that the same graph of \cref{eq: dsep_cycle_example} also has a non-trivial $p$-separation: $(\vertname_3\perp^p \vertname_4|\vertname_1,\vertname_2)_{\graphname}$.\footnote{This is because conditioning on $\vertname_1,\vertname_2$ blocks all paths from the remaining two vertices to the post-selection vertices within the loop which will ensure that $\vertname_3$ and $\vertname_4$ remain $d$-separated conditioned on $\vertname_1,\vertname_2$ in all teleportation graphs, even when additionally conditioning on the post-selection vertices.} Further, the following is an example of another cyclic graph with a non-trivial $p$-separation:
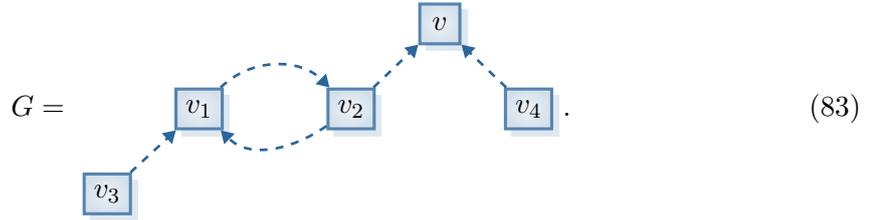
\begin{equation}
    \graphname=\centertikz{
    \node[onode] (v1) at (0,0) {$\vertname_1$};
    \node[onode] (v2) at (2,0) {$\vertname_2$};
    \node[onode] (v3) [below left  =1.5*\chanvspace of v1] {$\vertname_3$};
        \draw[cleg] (v1) to [out=45,in=135]  (v2);
    \draw[cleg] (v2) to [out=215,in=-45] (v1);
 \draw[cleg] (v3) -- (v1); 
 \node[onode] (v) [above right  =1.5*\chanvspace of v2] {$\vertname$};
  \node[onode] (v4) [below right  =1.5*\chanvspace of v] {$\vertname_4$};

 \draw[cleg] (v2) -- (v); \draw[cleg] (v4) -- (v);
    }.
\end{equation}
Here, we have the non-trivial $p$-separation $(\vertname_3\perp^p \vertname_4)_{\graphname}$. Here, the only (possibly undirected) path between $\vertname_3$ and $\vertname_4$ goes through $\vertname$. However,  $\vertname$ acts as an unconditioned collider on this path for any $\graphtele\in \graphfamily\graphname$, even when we condition on the set of post-selection vertices in $\graphtele$, i.e., $(\vertname_3\perp^p \vertname_4|\psvertset)_{\graphtele}$ holds for all $\graphtele\in \graphfamily\graphname$ implying the $p$-separation $(\vertname_3\perp^p \vertname_4)_{\graphname}$ by \cref{def: p-separation}.

\section{Links between Markovianity, linearity and success of post-selection}
\label{sec: quantum_Markov}
Consider a classical causal model (functional model or classical Bayesian network, see also~\cref{sec:classical aspects}) defined on a directed acyclic graph $\graphname$  where all vertices are observed, i.e., $\overtset=\vertset$. The joint distribution on the observed outcomes, which we denote with $\probacyc$ (see~\cref{def: distribution_functional_cm}), respects the Markov property, which corresponds to the following factorisation of this distribution
\begin{equation}
\label{eq: classical_Markov}
       \probfacyc(\{\outcome_\vertname\}_{\vertname\in\vertset})_{\graphname}=\prod_{\vertname\in\vertset} \probfacyc\left(\outcome_\vertname|\parnodes{\outcome_\vertname}\right)_{\graphname},
\end{equation}
where $\parnodes{\outcome_\vertname}$ denotes the outcome set of all parents of the vertex $\vertname$. If some of the vertices are unobserved, one simply marginalises over the outcomes of $\vertset\setminus\overtset$ on both sides of the above equation, to define the Markov property for $\probfacyc(\{\outcome_\vertname\}_{\vertname\in\overtset})_{\graphname}$.

In quantum causal models on directed acyclic graphs, we have the probability rule of \cref{eq:def acyclic composition} which expresses the joint probability of observed outcomes as a composition of completely positive linear maps, one for each vertex. We can simplify this by combining the maps $\measmaparg{\outcome_\vertname}{\vertname}$ and 
$\chanmaparg{\vertname}$ associated with observed and unobserved vertices, 
to a single map $\mathcal{T}^\vertname_{\outcome_\vertname}$ 
 associated with each vertex $\vertname$ with an empty outcome set $\outcome_\vertname=\emptyset$ (and the associated map being CPTP) whenever the vertex is unobserved, $\vertname\in \uvertset$. Observed vertices, $\vertname\in \overtset$, have non-trivial outcomes and the map associated to each outcome, $\mathcal{T}^\vertname_{\outcome_\vertname}$, is CP but not necessarily trace preserving. Then, we have 
\begin{equation}
\label{eq: acyclic prob simple}
    \probacyc(\{\outcome_\vertname\}_{\vertname\in\overtset})_{\graphname}
        =
        \bigcomp_{\vertname \in \vertset}
        \mathcal{T}_{\outcome_{\vertname}}^{\vertname}.
\end{equation}
One can check that the above expression is identical to that given by the Markov property in quantum Bayesian networks \cite{Henson_2014}, where the observed distribution is given by a generalised product of ``tests'', one for each vertex $\vertname$. Each test is a map from incoming quantum systems from unobserved parents to outgoing systems to its children, conditioned on the outcomes of its observed parents. When $\vertname$ is observed, the test is associated with a particular outcome on that vertex of which we compute the probability.\footnote{Here, the composition operation instantiates the generalised product in the quantum case, the tests correspond to the quantum CP(TP) maps $\mathcal{T}^{\vertname}$, and we do not have an explicit conditioning on outcomes of observed parents as these are encoded into a preferred basis of a quantum system and we do not need to differentiate incoming edges from observed vs unobserved parents, and can treat all of them as quantum inputs to our ``tests''. }

The expressions of \cref{eq:def acyclic composition} or equivalently \cref{eq: acyclic prob simple}, which can be applied to define the Markov property for $\probacyc(\outcome)_\graphname$ in the acyclic case,  are not defined for cyclic models. Indeed, in this case the composition operation $\bigcomp$, which captures parallel and sequential composition, is not defined. To understand how we can extend Markovianity to the cyclic case, we relate \cref{eq: acyclic prob simple} to the self-cycle composition of \cref{def:selfcycle_v3}, which is a special case of loop composition introduced in \cite{Portmann_2017}. Recall that in \cref{prop:probs as self cycles_v3}, we showed that the general probability rule for cyclic quantum causal models can be written in terms of self-cycle composition of a channel $\etot_\outcome$ constructed for a given choice of teleportation graph $\graphtele\in \graphfamily\graphname$ and the causal mechanisms of the original causal model. In addition, in~\cref{lemma: acyclic_prob_same_v3} we proved that the probability is independent of the choice of teleportation graph $\graphtele\in \graphfamily\graphname$. It follows from the results of \cite{Portmann_2017, VilasiniRennerPRA, VilasiniRennerPRL}, which express sequential composition in terms of loop composition\footnote{Self-cycle composition of \cref{def:selfcycle_v3} is a special case of loop composition \cite{Portmann_2017} for finite-dimensional systems, where the result is a number (a map with no in or outputs). Generally, loop composition for a map with inputs $A$ and $C$, and outputs $B$ and $D$ allows part of the output (say $C$) to be connected to part of the input (say $B$) to result in a new a map from $A$ to $D$. Then sequential composition of first applying $\mathcal{M}_1:\linops(\hilmaparg{A})\mapsto \linops(\hilmaparg{B})$ and then $\mathcal{M}_2:\linops(\hilmaparg{C})\mapsto \linops(\hilmaparg{D})$ (with $\hilmaparg{B}\cong \hilmaparg{C}$) is equivalent to loop composition of $C$ to $C$ in $\mathcal{M}_1\otimes \mathcal{M}_2: \linops(\hilmaparg{A})\otimes\linops(\hilmaparg{C})\mapsto \linops(\hilmaparg{B})\otimes \linops(\hilmaparg{D})$. Using this, it is easy to see that applying all the compositions of an acyclic graph, we will end up with a cycle composition once all maps are fully composed. }, that in acyclic causal models, the acyclic probability of \cref{eq:def acyclic composition} and \cref{eq: acyclic prob simple} is equivalent to
\begin{equation}
\label{eq: acyclic prob markov}
     \probacyc(\{\outcome_\vertname\}_{\vertname\in\overtset})_{\graphname}
    =\selfcycle(\etot_\outcome),
\end{equation}
where $\outcome$ in the self-cycle expression is short hand for $\{\outcome_\vertname\}_{\vertname\in\overtset}$, as we had in \cref{prop:probs as self cycles_v3}.

While \cref{eq:def acyclic composition} and \cref{eq: acyclic prob simple} are not defined in the cyclic case, the above equation is, since the $\selfcycle$ operation is defined generally. This expression is independent of the choice of teleportation graph used (\cref{cor:selfcycle etot indepentend on telegraph}), thus, without loss of generality one can choose the teleportation graph obtained by splitting all edges, in which case 
\begin{equation}
    \etot_\outcome=\bigotimes_{\vertname\in\vertset}\mathcal{T}_{\outcome_\vertname}^\vertname=\bigotimes_{\vertname\in\overtset} \measmaparg{\outcome_\vertname}{\vertname}\bigotimes_{\vertname\in \uvertset}\chanmaparg{\vertname},
\end{equation}
and 
\begin{equation}
   \label{eq: acyclic prob markov2}
     \prob(\{\outcome_\vertname\}_{\vertname\in\overtset})_{\graphname}
    =\selfcycle\left(\bigotimes_{\vertname\in\vertset}\mathcal{T}_{\outcome_\vertname}^\vertname\right)=\selfcycle\left(\bigotimes_{\vertname\in\overtset} \measmaparg{\outcome_\vertname}{\vertname}\bigotimes_{\vertname\in \uvertset}\chanmaparg{\vertname}\right). 
\end{equation}
The above expression gives the observed probability in terms of a composition operation on all the maps of the causal model where $\selfcycle$ links a pair of isomorphic in and output systems for each edge of the causal graph of the model. 

This suggests the following definition of Markovianity for the general cyclic case.

\begin{definition}[Markov property in cyclic quantum causal models]
\label{def: markov}
For a causal model defined on a directed graph $\graphname$, the associated observed probability distribution $\prob(\{\outcome_\vertname\}_{\vertname\in\overtset})_\graphname$ is said to be Markov with respect to $\graphname$ if it can be expressed as in \cref{eq: acyclic prob markov2}. If this is the case, we will say that the causal model satisfies the Markov property or is probabilistically Markovian. 
\end{definition}
Then the following corollary is implied by \cref{prop:probs as self cycles_v3}.

\begin{corollary}
Given a causal model on a directed graph $\graphname$, $\cm_{\graphname}$, the following statements are equivalent:
\begin{myitem}
    \item $\cm_{\graphname}$ is probabilistically Markovian;
    \item $\sum_\outcome\selfcycle(\etot_\outcome)=1$;
    \item For every choice of $\graphtele\in \graphfamily{\graphname}$ used in the definition of $\mathcal{C}_\outcome$ and $\successprob$, the post-selection success probability is 
\begin{equation}
     \successprob = \left(\prod_{i=1}^k \teleprob^{(i)}\right),
\end{equation}
where $k$ is the cardinality of the split edges set  of $\graphtele$, i.e., $k = |\splitedges{\graphtele}|$, and $\teleprob^{(i)}$ is the teleportation probability associated to the post-selected teleportation protocol implemented instead of the $i$-th split edge in $\splitedges{\graphtele}$ (see~\cref{prop:probs as self cycles_v3}).
\end{myitem} 
\end{corollary}
This corollary has an interesting interpretation. Recall that in general,
we have \begin{equation}
     \successprob = \left(\prod_{i=1}^k \teleprob^{(i)}\right)\sum_\outcome\selfcycle(\etot_\outcome).
\end{equation}
The term $\sum_\outcome\selfcycle(\etot_\outcome)$ depends on the causal mechanisms of the causal model and not on the choice of teleportation protocols used on the split edges, while the term $\left(\prod_{i=1}^k \teleprob^{(i)}\right)$ is a product of the individual success probabilities of the teleportation protocols used on the split edges and does not depend on the causal mechanisms of the original causal model at all. When we use the maximally entangled post-selected teleportation protocol (involving an entangled state of two $d$-dimensional systems) on the $i$-th edge then $\teleprob^{(i)}=\frac{1}{d^2}$. This corollary tells us that Markovianity holds whenever the overall post-selection success probability of the causal model is entirely independent of its causal mechanisms. 

Moreover, notice that $\successprob$ and hence the $\sum_\outcome\selfcycle(\etot_\outcome)$ term appears in the denominator of our general probability expression (\cref{prop:probs as self cycles_v3}), this makes the probabilities generally non-linear in the causal mechanisms $\{\measmaparg{\outcome_\vertname}{\vertname}\}_{\vertname\in\overtset}$ and $\{\chanmaparg{\vertname}\}_{\vertname\in\uvertset}$ of the causal model. Causal models leading to Markovian distributions are therefore linear in the causal mechanisms of the model. Note, however, that for linearity, it suffices that $\sum_\outcome\selfcycle(\etot_\outcome)$ equals a constant value $c$ not necessarily equal to 1, while Markovianity requires $c=1$.

\paragraph{Link to linear P-CTCs.} The intuition behind the results of this section are similar in spirit to certain characterisation results known 
for process matrices or equivalently quantum supermaps \cite{Oreshkov_2012, Chiribella_2013}, which are formalisms for describing so-called indefinite causal order (ICO) quantum processes. These are known to be a linear subset \cite{Araujo_2017} of the general (possibly non-linear) post-selected closed timelike curves (P-CTCs) \cite{Lloyd_2011, Lloyd_2011_2}. In particular, P-CTCs are defined with respect to maximally entangled pre and post-selections and process matrices correspond to P-CTCs where the overall success probability equals the product of $\frac{1}{d^2}$ for each maximally entangled state of two $d$-dimensional systems used in the construction. This is precisely what we would get for the success probability $\successprob$ here by choosing the maximally entangled pre and post-selections on each split edge. Thus, process matrices, when mapped to our framework as discussed in \cref{sec:process_matrices}, would correspond to probabilistically Markov causal models. More generally, our framework can give a causal modelling semantic to general non-linear P-CTCs as well, which would not necessarily correspond to probabilistically Markov models.

\section{Overview of relationships to other causality frameworks}
\label{sec: frameworks}
\begin{figure}
    \centering
  \adjustbox{max width=\textwidth}{   \begin{tikzcd}[arrows=rightarrow, row sep=1.7cm, column sep=1.5cm]
  \begin{tikzpicture}\node[green!70!black] at (0,0.6) {$\sigma$-sep};
\node[fill=green] at (0,0) {\fbox{mfCM \cite{forre_2018}}};\end{tikzpicture}\arrow[shift left=1.5ex]{r}{}\arrow[red,degil, shift left=1.5ex]{d}{\text{continuous var.}}&\fbox{fCM \cite{Pearl_2009}}\arrow[red,degil, shift right=1.5ex, swap]{r}{\text{continuous var.}}\arrow[red,degil, shift left=1.5ex]{l}{\text{non-unique solv.}}\arrow[red,degil, shift left=1.5ex]{d}{\text{continuous var.}}&\fbox{Non-classical cyclic models \cite{VilasiniColbeckPRA, VilasiniColbeckPRL}}\arrow[red,degil, shift right=1.5ex, swap]{rd}{\text{cyclic}}\arrow[red,degil, shift right=1.5ex, swap]{l}{\text{non-classical}}\arrow[red,degil, shift left=1.5ex]{d}{\text{non-quantum}}& \\
\fbox{mfCM$_{\textup{finite}}$}\arrow[shift right=1.5ex]{r}{}\arrow[shift left=1.5ex]{u}{}&\fbox{fCM$_{\textup{finite}}$}\arrow[red, degil,shift right=1.5ex, swap]{r}{\text{no $d$-sep}}\arrow[red,degil, shift right=1.5ex, swap]{l}{\text{non-unique solv.}}\arrow[shift left=1.5ex]{u}{}\arrow[shift right=1.5ex, swap]{d}{\Cref{sec:fCM_to_CM}}&\fbox{This framework$_{d-\text{sep}}$}\arrow[red, degil, shift left=1.5ex]{r}{\text{cyclic}}\arrow[shift left=1.5ex]{ld}{}\arrow[red,degil,shift right=1.5ex, swap]{l}{\text{quantum}}\arrow[shift left=1.5ex]{u}{}&\begin{tikzpicture}\node[green!70!black] at (0,0.6) {$d$-sep};\node[fill=green] at (0,0) {\fbox{HLP \cite{Henson_2014}}};\end{tikzpicture}\arrow[shift left=-1.5ex]{lu}{}\arrow[red, degil, shift left=1.5ex]{d}{\text{non-quantum}}\arrow[red, degil, shift left=1.5ex]{l}{\text{non-quantum}}\\
\fbox{\parbox{6cm}{Cyclic q. networks \cite{VilasiniRennerPRA, VilasiniRennerPRL}\\ \centering String diagrams (quantum)\cite{Coecke_Kissinger_2017}}}\arrow[shift left=1.5ex]{d}{}\arrow[shift left=1.5ex]{r}{\text{\shortstack{up to modelling\\  open wires}}}&\begin{tikzpicture}\node[green!70!black] at (0,0.6) {$p$-sep};\node[fill=green] at (0,0) {\fbox{This framework}};\end{tikzpicture}\arrow[red,degil,shift right=1.5ex, swap]{rr}{\text{cyclic}}\arrow[red,degil,shift left=1.5ex]{ru}{\text{no $d$-sep}}\arrow[red,degil, shift left=-1.5ex, swap]{u}{\text{quantum}}\arrow[shift left=1.5ex]{l}{}\arrow[red,degil, shift left=1.5ex]{rd}{non-process}\arrow[red,degil, shift left=1.5ex]{dd}{non-process}&&\begin{tikzpicture}\node[green!70!black] at (0,0.6) {$d$-sep};\node[fill=green] at (0,0) {\fbox{HLP$_Q$}};\end{tikzpicture}\arrow[shift right=1.5ex]{ll}{}\arrow[shift left=1.5ex]{u}{}\arrow[shift left=1.5ex]{d}{}\\
\fbox{Post-selected CTCs \cite{Lloyd_2011}}\arrow[shift left=1.5ex]{u}{}\arrow[red,degil,shift left=1.5ex]{d}{\text{non-linear}}&&\fbox{BLO-QCM$_{\otimes}$}\arrow[red,degil,shift right=1.5ex, swap]{r}{\text{cyclic}}\arrow[shift left=1.5ex]{lu}{\Cref{sec: BLOtoOurs}}\arrow[shift left=-1.5ex, swap]{ld}{\Cref{def: BLOTensor}}&\begin{tikzpicture}\node[green!70!black] at (0,0.6) {$d$-sep};\node[fill=green] at (0,0) {\fbox{Costa-Shrapnel \cite{Costa_2016}}};\end{tikzpicture}\arrow[shift right=1.5ex]{l}{}\arrow[shift left=1.5ex]{u}{}\arrow[shift right=1.5ex]{d}{}\\
\fbox{Process matrices \cite{Oreshkov_2012}}\arrow[red,degil,shift left=1.5ex]{r}{\text{non-unitary}}\arrow[shift left=1.5ex]{u}{}&\fbox{BLO-QCM \cite{Barrett_2021}}\arrow[shift left=1.5ex]{l}{\text{coarse}}\arrow[red,degil,shift right=1.5ex, swap]{rr}{\text{cyclic}}\arrow[red,degil, shift right=1.5ex, swap]{ru}{non-factor}\arrow[blue,shift left=1.5ex]{uu}{\text{Causal decomp.?}}&&\begin{tikzpicture}\node[green!70!black] at (0,0.6) {$d$-sep};\node[fill=green] at (0,0) {\fbox{BLO-QCM$_{acyc}$ \cite{Barrett_2019}}};\end{tikzpicture}\arrow[shift right=1.5ex]{ll}{}\arrow[red, degil, shift right=1.5ex, swap]{u}{\text{non-factor}}\\
\end{tikzcd}}
    \caption{Relationships between different causality frameworks and ours. Each vertex represents a causality framework, a directed (black) arrow from vertex A to vertex B represents that all objects in framework A can be mapped to and studied within framework B in a manner that preserves the relevant causal information. A red arrow from A to B represents that not all scenarios in A can be captured within B, the annotation indicates the reason why. For instance, our framework can describe all classical functional causal models on discrete variables with finite cardinality (fCM$_{\textup{finite}}$) but the converse is not true as our framework contains quantum causal models that are not captured in fCM$_{\textup{finite}}$. Whether the Barrett-Lorenz-Oreshkov formalism for cyclic quantum causal models can be faithfully mapped to ours is an open question (represented by a blue arrow) that relies on another prominent open question in field relating to causal decompositions of quantum channels (see \cref{sec: BLOtoOurs} for further discussion).
    The green boxes highlight the graph-separation property (annotated above the box) whose soundness and completeness is known to hold within that formalism. For further explanation of this diagram, see \cref{sec: frameworks}.}
    \label{fig:frameworks}
\end{figure}

\Cref{fig:frameworks} illustrates relationships between our framework and several other causality frameworks proposed within the classical statistics and quantum information communities. Below, we highlight the key features distinguishing these frameworks and briefly compare them to ours to discuss the generality of our formalism. For further details, see also \cref{sec:quantum aspects,sec:classical aspects}. We note that this is not an exhaustive list of such frameworks, but refers to notable and commonly used ones from different types of approaches.

\begin{enumerate} \item {\bf Classical vs Non-classical}

\begin{enumerate}
    \item {\bf Classical: finite-cardinality vs continuous variables, unique vs non-uniquely solvable.}
    Functional causal models (fCM, reviewed in  \cref{sec:fCMs}, and also referred to as \emph{structural equation models} or \emph{structural causal models} in the classical literature) and their restricted variants, such as on finite-cardinality variables or modular functional causal models (fCM$_{\textup{finite}}$, mfCM, mfCM$_{\textup{finite}}$), are purely classical frameworks. 
    The relationships among these frameworks are evident from their specified restrictions. Modular fCMs (mfCMs) \cite{Forre_2017, forre_2018} correspond to cases where the functional dependencies or structural equations possess a particular solvability property, which is stricter than unique solvability\footnote{Unique solvability corresponds to having a unique solution of the functional dependences for every valuation of the exogenous vertices.}. 
    In \cite{Forre_2017}, a graph-separation property for cyclic fCMs, called $\sigma$-separation was proposed and proven to be sound and complete for all modular fCMs (referred there as modular structural equation models) which cover finite and infinite cardinality, discrete and continuous variables.
    However, the soundness of $\sigma$-separation (and $d$-separation) can fail in non-modular fCMs even in the discrete and finite-cardinality case: this includes uniquely solvable models (such as \cite{Neal_2000}) as well non-uniquely solvable models (see section 4 of~\cite{Sister_paper} for more details on $\sigma$-separation). Our framework, when restricted to the classical domain, can fully describe fCM$_{\textup{finite}}$ (\cref{sec:fCM_to_CM}) and, by extension, its subset mfCM$_{\textup{finite}}$ (see also the companion paper \cite{Sister_paper}). However, since we work with finite-dimensional Hilbert spaces, we cannot capture continuous-variable or infinite-cardinality fCMs.

    \item {\bf Non-classical: quantum vs post-quantum.} 
    The remaining frameworks in \cref{fig:frameworks} are inherently quantum, except for the Henson-Lal-Pusey (HLP) framework \cite{Henson_2014} and the non-classical cyclic causal modelling framework of \cite{VilasiniColbeckPRA, VilasiniColbeckPRL}, which apply to general operational probabilistic theories. These theories can include non-quantum and post-quantum models not captured by other frameworks, with the latter accommodating cyclic and post-quantum scenarios. Our framework can describe many of the quantum frameworks, as shown in the diagram, but it does not generically capture post-quantum causal models.
\end{enumerate}

\item {\bf Cyclic vs acyclic: satisfaction of $d$-separation property and correspondence with valid process operators.} 
The four frameworks in the bottom-most row are defined on directed acyclic graphs, where $d$-separation is known to be sound and complete. The other frameworks, by contrast, do not assume acyclicity and allow for certain types of cycles. The framework of \cite{VilasiniColbeckPRA, VilasiniColbeckPRL} defines a class of cyclic causal models applicable in classical, quantum and post-quantum theories but restricted to those respecting the soundness of $d$-separation, while other formalisms, like \cite{Barrett_2019,Barrett_2021}, focus on cyclic models associated with valid process operators. Process operators/process matrices/ quantum supermaps \cite{Oreshkov_2012, Chiribella_2013} correspond to a linear subset of post-selected closed timeline curves (P-CTCs) \cite{Lloyd_2011, Lloyd_2011_2}. Our framework allows for classical and quantum models that violate the soundness of $d$-separation (\cref{sec:failure_dsep}) as well as those with cycles that do not necessarily correspond to valid process operators. 

\item {\bf Unitary vs general quantum channels.} 
The Barrett-Lorenz-Oreshkov (BLO) approach \cite{Barrett_2019, Barrett_2021} focuses on causal influence within unitary quantum channels, while the quantum restriction of the HLP approach \cite{Henson_2014} (and our extension to the cyclic case) allows for general quantum channels. While there are scenarios that can be modelled in both approaches, there are notable differences in the definition of causal influence conditions and the associated concept of causal structure between these two approaches (see also \cref{sec: BLOtoOurs} for further discussion).

\item {\bf Existence/Non-existence of tensor factorization across multiple edges.} 
A related point to the above is whether or not a separate tensor factor is assigned to each quantum system for each edge. For instance, when $A$ is a cause of both $B$ and $C$, whether $B$ and $C$ receive distinct quantum systems in tensor product with each other. In the BLO formalism, non-factor constructions are allowed, meaning a single system that cannot be split into an $n$-fold tensor product can still be a common parent for $n$ vertices. All BLO models that involve such tensor factorization can be faithfully reproduced within our framework as shown in \cref{sec: BLOtoOurs}. Whether non-factor BLO models can be reproduced directly in our framework depends on solving the open problem of causal decompositions of unitary channels in terms of arbitrary quantum channels, as discussed in \cref{sec: BLOtoOurs}.

\item {\bf Modelling of open wires.} 
Many causal modelling approaches, including ours, do not include open in/output wires; wires are fully contracted with a state preparation, channel, or POVM. 
Causality frameworks that are not strictly about causal modelling, such as process matrices/quantum supermaps \cite{Oreshkov_2012, Chiribella_2013}, cyclic quantum networks \cite{VilasiniRennerPRA, VilasiniRennerPRL}, and category-theoretic approaches based on string diagrams (e.g., \cite{Coecke_Kissinger_2017, coecke2016categoricalquantummechanicsi}) allow for open in/output wires. These can still be modelled within our framework by introducing a “slot” for state preparation or POVM measurement on an open in/output wire. This slot can be filled with a variable preparation/measurement box, which, depending on an input variable, prepares an arbitrary state or measures an arbitrary POVM. Thus, there is no strict loss of generality in this case, though explicitly modelling open wires may be useful in certain scenarios.

\item {\bf Faithful vs unfaithful mappings and coarse-graining.} The BLO QCM formalism \cite{Barrett_2021} provides detailed causal semantics for several unitary process matrices including those with indefinite causal order, by decomposing the process operator into marginal channels. When mapping such processes back from the BLO framework to process matrices, the causal semantics given by the causal graph is not maintained and the process operator/matrix is treated as one (coarse-grained) entity that interacts through feedback loops with each local operation. As discussed in \cref{sec:process_matrices}, such a mapping between formalisms that involves coarse-graining no longer faithfully reflects the causal connectivity, although the operational predictions can be recovered in both pictures. \Cref{fig:frameworks} shows that although direct path from BLO-QCMs to our framework is an open question (blue arrow), there exists a directed path of black arrows from BLO-QCMs to our framework via process matrices. The former says that the existence of a faithful mapping (in the sense of \cref{def: faithful_mapping}) is open while the latter says that there exists an unfaithful mapping nevertheless, as shown in \cref{sec:process_matrices}. 

\item {\bf Choice of post-selected teleportation} The Post-selected Closed Timelike Curves (P-CTCs) framework \cite{Lloyd_2011, Lloyd_2011_2} defines cyclic causal structures through post-selected quantum teleportation. While P-CTCs are specifically associated with post-selected teleportation using maximally entangled states (\cref{def:bell tele}), our work generalizes this by allowing for arbitrary post-selected teleportation protocols that do not necessarily involve maximally entangled pre- and post-selections (\cref{def:ps teleportation}). Notably, we have demonstrated in \cref{sec:probability rule CM} that the operationally accessible probabilities can always be computed via a cyclic composition operation, irrespective of the chosen teleportation protocol. This establishes a direct connection between our framework and prior works based on cyclic composition, such as \cite{VilasiniRennerPRA, VilasiniColbeckPRL, Portmann_2017}. Additionally, similar ideas linking between cyclic composition and pre- and post-selection exist in category-theoretic approaches that employ string diagrams \cite{Coecke_Kissinger_2017, coecke2016categoricalquantummechanicsi}. In these approaches, special pairs of bipartite states and effects, called cups and caps, combine to simulate an identity channel used in cyclic compositions. Although maximally entangled states are frequently used as canonical examples in the literature, the string-diagrammatic framework accommodates the broader class of teleportation protocols in \cref{def:ps teleportation}. While we have focused on a more concise representation of our causal models via causal graphs, which is suitable for studying open problems relating to graph-separation properties in cyclic models, one can equivalently represent our models through string diagrams defined in a category instantiated by quantum channels.\footnote{Generally, the diagrammatic and category-theoretic approach can be applied to non-classical and non-quantum theories as well.}
\end{enumerate}

To summarize, our formalism exhibits generality by encompassing both quantum and classical causal models. It allows for non-unitary quantum channels, cyclic models that potentially violate the soundness of $d$-separation, and models that are not necessarily associated with valid process operators (but may involve non-linear post-selected closed timelike curves). The graph separation property $p$-separation, introduced here, is sound and complete for all causal models within this formalism.

Two potential limitations in its generality arise from the assumption of finite-dimensional Hilbert spaces and the requirement of associating a separate tensor factor for each edge. The former restricts the classical limit of the formalism to discrete and finite-cardinality variables. While no quantum causal modelling frameworks for infinite-dimensional Hilbert spaces are known, continuous-variable and infinite-cardinality classical models are have been studied formally even in cyclic scenarios (see e.g., \cite{Bongers_2021}). Whether the restriction to requiring a separate subsystem (tensor factor) to each edge leads to any loss of generality is an open question, and we are not aware of any counter examples that relax this restriction and cannot be modelled faithfully within our framework.

\section{Discussion and outlook}
\label{sec: conclusions}
The main contributions of our work are summarised in \cref{sec: intro}. There are a number of interesting future directions for building on this framework and results, and open questions that remain to be explored. We discuss these below.
\begin{itemize}
    \item {\bf Causal discovery and causal compatibility in cyclic causal structures.} In the case of acyclic models, the soundness and completeness of $d$-separation plays a central role in causal discovery algorithms which infer underlying causal explanations in terms of directed acyclic graphs from observed data. Such algorithms exist in the classical and quantum literature for acyclic models (e.g., \cite{Spirtes2016, Giarmatzi2018}). A natural question is whether $p$-separation can be applied to generalize such causal discovery algorithms in a useful way to cyclic classical and quantum causal models. Similarly, the $d$-separation criterion underpins causal compatibility questions and witnessing non-classical correlations, such as in Bell's inequalities and network non-locality, where one asks whether a given causal structure supports a certifiable gap between sets of correlations realizable via classical vs non-classical causal models (see e.g., \cite{Henson_2014}).
Our results create possibilities for exploring analogous questions in cyclic graphs, such as whether meaningful separations between classical and non-classical correlations can be identified, or whether there exist novel cyclic quantum causal models that violate classical compatibility constraints in a genuinely cyclic graph. More conceptually, this may bear insights for the question of whether in a world with causal loops, quantum theory can still be more powerful than classical.

\item {\bf Accounting for interventions.} We have focussed our attention on observed correlations. However, data resulting from active interventions is also central for identifying causal influences. We have shown how objects from the Barrett-Lorenz-Oreshkov \cite{Barrett_2021}, Costa-Shrapnel \cite{Costa_2016} and process matrix/quantum supermaps \cite{Oreshkov_2012, Chiribella_2013} frameworks can be described in our framework. These frameworks allow for arbitrary interventions, as their vertices correspond to empty ``slots’’ where any quantum channel can be plugged in. For every such choice of intervention, one obtains a causal model in our framework. Explicitly modelling interventions in our framework and studying the causal inference implications of allowing different classes of interventions is a subject of an upcoming work based on the master's thesis~\cite{Master_thesis}. Such problems have been considered in the acyclic quantum literature \cite{Friend_2023}. Beyond this, there is scope to extend these works to systematically investigate causal inference (inferring the causal graph from available data) and causal identification (identifying the causal mechanisms, such as channels, from available data) in cyclic graphs using correlations and interventions.

\item {\bf Infinite-dimensional, continuous case and $\sigma$-separation} In this work, we have focused on causal models for finite-dimensional quantum systems, which embed classical functional models on finite-cardinality variables (including non-uniquely solvable models). A direction for future work would be to identify whether and to what extent our results can be generalised to the infinite-dimensional case. This would require considering generalisations of the post-selected teleportation protocol to the infinite-dimensional case and a careful, measure-theoretic treatment of the associated success probabilities.

While cyclic quantum causal modelling frameworks for infinite-dimensional systems have not been considered to our knowledge, the classical literature has explored functional models or structural equation models (fCMs)\footnote{In~\cite{Forre_2017}, these are denoted as SEMs instead of fCMs. The two acronyms refer to the same object for our present purposes.} with continuous variables, including cyclic graphs. 
In this context, a cyclic generalization of $d$-separation, known as $\sigma$-separation, was proposed in \cite{Forre_2017} and shown to be sound and complete for a subclass called modular fCMs (mfCMs) that satisfy a unique solvability property. The relationship between $p$-separation (introduced here) and $\sigma$-separation is discussed in detail in our companion paper \cite{Sister_paper}, which focuses on classical causal models. Notably, we find that neither concept is strictly more general than the other, although both extend $d$-separation to cyclic settings\footnote{This is perhaps to be expected as they consider different domains, $\sigma$-separation was constructed for uniquely solvable but possibly continuous variable classical fCMs while $p$-separation for possibly non-uniquely solvable finite cardinality/dimensional classical and quantum models.}. Moreover, this motivates open questions regarding gaps between different graph separation properties, e.g., under what conditions (on the graph and causal model) can we have different conditional independence constraints imposed by the soundness of two graph separation properties (among $d$, $\sigma$ and $p$ separation), where a causal model violates one of the constraints and satisfies the other?
The results and further examples of \cite{Sister_paper} present relevant insights and techniques for this question. 

Additionally, our results also indicate the possibility of correlation gaps between all finite-dimensional causal models (classical and quantum) and continuous variable classical models. This is because there exist scenarios (see \cite{Sister_paper}) where there is $p$-separation but $\sigma$-connection. The $p$-separation would imply corresponding conditional independence in all finite-dimensional causal models (by \cref{theorem: dsep theorem}), but there exist continuous variable models leading to conditional dependence as allowed by the $\sigma$-connection in this case \cite{forre_2018, Forre_2017}. In the broader context of quantum foundations, this may bear insights on correlation gaps between measurement statistics from finite and infinite dimensional systems in the presence of consistent causal loops.

\item {\bf Characterizing ``indefinite causal order’’ processes from graph-separation.} While we have introduced a generally sound graph separation criterion, $p$-separation, for cyclic models, it remains a relevant question to understand which subclass of cyclic causal models, in both classical and quantum settings, respects the soundness of $d$-separation. In the classical case, $d$-separation soundness holds for a class of models called ancestrally uniquely solvable \cite{Forre_2017}, though it is unknown whether these are the most general models respecting $d$-separation. On the other hand, indefinite causal order (ICO) processes, particularly those described by process matrix and quantum supermap formalisms \cite{Oreshkov_2012, Chiribella_2013}, are a specific subclass of cyclic causal structures that can be understood in the post-selected closed timelike curve framework \cite{Araujo_2017}. Interesting subsets of such ICO processes, such as unitary, causal, causally non-separable processes have been studied \cite{Araujo_2015, Oreshkov_2016,Araujo_2017purif,Wechs_2019}. It would be valuable to explore concrete relations between the subclass of cyclic quantum causal models respecting $d$-separation soundness and subclasses of models corresponding to valid process matrices, or a subset of these such as unitary, causal or causally non-separable processes. This would address open questions in classical causal modelling while providing an alternative characterization of quantum processes. 

There is already some evidence for such connections between the two domains, but a clear characterization is lacking. For example, process matrix protocols satisfy the conditions for probabilistic Markovianity identified in our framework (\cref{sec: quantum_Markov}), and in acyclic causal models soundness of $d$-separation and Markov factorizations of probabilities share close links \cite{Lauritzen_1990,Geiger_1990}. Furthermore, $d$-separation is tied to unique solvability properties in the classical case \cite{Forre_2017}, while valid classical process matrices are characterised via having unique fixed points under all choices of local interventions \cite{Baumeler_2016_fixedpoints}. This motivates concrete future research at the intersection of cyclic causal models, graph separation, and indefinite causality. Further discussions on the link between our framework and ICO protocols, and related open questions can be found in \cref{sec:process_matrices}.

\item {\bf Understanding space-time structure from operational causality.} Causal models define causation through information-theoretic mechanisms, whereas causality in relativistic physics is tied to space-time geometry. Understanding the  general interface of information-theoretic and spatio-temporal causality notions houses a number of intriguing and unexplored problems. A set of recent works connect causal models to relativistic principles by embedding them in space-time and considering the compatibility between the two notions of causal order \cite{VilasiniColbeckPRA, VilasiniColbeckPRL, VilasiniRennerPRA, VilasiniRennerPRL}. Conversely, one may ask whether certain properties of spatio-temporal causation can ``emerge’’ from information-theoretic models of causation, and can we thus understand them from more basic operational principles. With this motivation, in a follow-up work based on the master's thesis~\cite{Master_thesis}, we link the present causal modelling framework to tensor networks, which are used in the quantum information and relativistic physics communities for studying questions relating to the emergence of space-time from quantum correlations, to enable a transfer of techniques between the two research communities. From a fundamental perspective, it would be interesting to develop the initial insights of this and follow-up works towards the larger goal of understanding the emergence of acyclic causal structures and thus an operational arrow of time from information-theoretic principles, or analogues of space-time properties such as distance and curvature through graph-separation criteria.
\end{itemize}

\bigskip
\paragraph{Acknowledgements} We are grateful to Elie Wolfe for valuable feedback on our framework and its relation to results in classical cyclic causal models. We are thankful to Y\'{i}l\'{e} Y\={\i}ng for suggesting the name $p$-separation for the new graph separation property that we propose here. VV also thanks Tein van der Lugt for insightful discussions on the causal decomposition problem, which is related to the question regarding the link between our framework and the quantum causal modelling framework of Barrett, Lorenz and Oreshkov. VV's research has been supported by an ETH Postdoctoral Fellowship. VV acknowledges support from ETH Zurich Quantum Center, the Swiss National Science Foundation via project No.\ 200021\_188541 and the QuantERA programme via project No.\ 20QT21\_187724. VG and CF acknowledge support from NCCR SwissMAP, the ETH Zurich Quantum Center and the Swiss National Science Foundation via project No.\ 20QU-1\_225171. 
CF acknowledges support from the ETH Foundation.

\bibliographystyle{mybibstyle}
\newcommand{\etalchar}[1]{$^{#1}$}

\newpage
\appendix

\section{Post-selected teleportation}
\label{app:ps teleportation}

Throughout this section, we let $(\telepovm_{AB},\telestate_{BC})$ be an implementation of a postelected teleportation protocol according to \cref{def:ps teleportation}.
We start by confirming that the success probability of the teleportation protocol does not depend on the input state.
This will be needed for the upcoming results.
\begin{lemma}
\label{lemma:teleprob indep input}
    The teleportation success probability $\teleprob$ of \cref{def:ps teleportation} cannot depend on the state $\rho_A$ to be teleported.
\end{lemma}
\begin{proof}
    Suppose that $\teleprob$ is allowed to depend on the input state, so that 
    \begin{align}
        \Tr_{AB}[\telepovm_{AB} \rho_A \telestate_{BC}] = \teleprob(\rho) \rho_C.
    \end{align}
    Suppose that $\teleprob(\rho)$ was not constant, so that there exist $\rho \neq \rho'$ with $\teleprob(\rho) \neq \teleprob(\rho')$.
    Consider a mixture of $\rho$ and $\rho'$ to be teleported, for some $\lambda \in [0,1]$:
    \begin{align}
        \Tr_{AB}[\telepovm_{AB} (\lambda\rho_A + (1-\lambda)\rho'_A) \telestate_{BC}] 
        &= \teleprob\big( \lambda\rho + (1-\lambda)\rho' \big) \big( \lambda\rho_C + (1-\lambda)\rho'_C \big). \label{eq:proof tele prob 1}
    \end{align}
    By linearity, we also have
    \begin{align}
        \Tr_{AB}[\telepovm_{AB} (\lambda\rho_A + (1-\lambda)\rho'_A) \telestate_{BC}] 
        &= \lambda \Tr_{AB}[\telepovm_{AB} \rho_A \telestate_{BC}] + (1-\lambda) \Tr_{AB}[\telepovm_{AB} \rho'_A \telestate_{BC}] \nonumber\\
        &= \lambda \teleprob(\rho) \rho_C + (1-\lambda) \teleprob(\rho') \rho'_C. \label{eq:proof tele prob 2}
    \end{align}
    Equating \cref{eq:proof tele prob 1,eq:proof tele prob 2}, we get 
    \begin{align}
    \label{eq:proof tele prob 3}
        \teleprob\big( \lambda\rho + (1-\lambda)\rho' \big) \big( \lambda\rho_C + (1-\lambda)\rho'_C \big)
        &= \lambda \teleprob(\rho) \rho_C + (1-\lambda) \teleprob(\rho') \rho'_C.
    \end{align}
    Taking the trace of this equation yields
    \begin{align}
        \teleprob\big( \lambda\rho + (1-\lambda)\rho' \big)
        &= \lambda \teleprob(\rho) + (1-\lambda) \teleprob(\rho').
    \end{align}
    Inserting this identity back in \cref{eq:proof tele prob 3}, we obtain a quadratic expression in $\lambda$:
    \begin{align}
        (\teleprob(\rho)-\teleprob(\rho'))(\rho_C - \rho'_C)
        \lambda^2
        +
        ({\cdots}) \lambda + ({\cdots}) = 0.
    \end{align}
    This equation needs to hold for all $\lambda \in [0,1]$, which implies in particular that the term proportional to $\lambda^2$ must vanish.
    This is a contradiction to our assumption that $\rho \neq \rho'$ and $\teleprob(\rho) \neq \teleprob(\rho')$.
\end{proof}
We now prove a type of self-testing result: the pre- and post-selection implementing a post-selected teleportation protocol have to take a specific form.
\begin{lemma}[Self-testing of post-selected teleportation implementations]
\label{lemma:self test}
    Let $d_A = \dim(\hilmaparg A) = \dim(\hilmaparg C)$ and $d_B = \dim(\hilmaparg B)$.
    It holds that $(\telepovm_{AB},\telestate_{BC})$ is a post-selected teleportation protocol implementation as in \cref{def:ps teleportation}, with $\teleprob \in (0,1]$ the teleportation probability, if and only if $d_B \geq d_A$ and there exist coefficients $\{0\neq \telestate_k \in \mathbb{R}\}_{k=1}^{d_A}$ satisfying
    \begin{align}
        \label{eq:selftest coeffs}
        \sum_{k=1}^{d_A} \telestate_k^2 = 1, \qquad 
        \sum_{k=1}^{d_A} \frac{1}{\telestate_k^2} \leq \frac{1}{\teleprob},
    \end{align}
    and there exist an ancilla Hilbert space $\hilmaparg{B'}$, as well as orthonormal bases $\{\ket k_A\}_{k=1}^{d_A}$\footnote{This basis is also an orthonormal basis $\{\ket k_C\}_{k=1}^{d_A}$ for $\hilmaparg C = \hilmaparg A$.} for $\hilmaparg A$, $\{\ket k_{BB'}\}_{k=1}^{d_Bd_{B'}}$ for $\hilmaparg B \otimes \hilmaparg{B'}$ such that
    \begin{align}
        \telestate_{BC} = \Tr_{B'}[\ketbra{\telestate}_{BB'C}], \qquad
        \Pi_{BB'}\telepovm_{AB}\Pi_{BB'} = \ketbra{\telepovm}_{ABB'},
    \end{align}
    where we defined the pure state $\ket\telestate_{BB'C}$, the subnormalized state $\ket{\telepovm}_{ABB'}$ and the orthogonal projector $\Pi_{BB'}$ as follows:
    \begin{align}
        \ket{\telestate}_{BB'C} &= \sum_{k=1}^{d_A} \telestate_k \ket{k}_{BB'} \otimes \ket{k}_C, \\
        \ket{\telepovm}_{ABB'} &= \sqrt{\teleprob}\sum_{k=1}^{d_A} \frac{1}{\telestate_k} \ket{k}_{A} \otimes \ket{k}_{BB'}, \quad
        \Pi_{BB'} = \sum_{k=1}^{d_A} \ketbra{k}_{BB'}.
    \end{align}
\end{lemma}
\begin{proof}
    The ``if'' direction is trivial.
    We now prove the ``only if'' direction: suppose that $(\telepovm_{AB},\telestate_{BC})$ implements a post-selected teleportation protocol.

    \noindent We first prove that $d_B \geq d_A$.
    Since the teleportation probability is independent of the input state (as shown in \cref{lemma:teleprob indep input}), by linearity, we must have that $R\circ O = \teleprob \mathcal I$, where $\mathcal I : \linops(\hilmaparg A) \mapsto \linops(\hilmaparg C)$ is the identity map, and $O : \linops(\hilmaparg A) \mapsto \linops(\hilmaparg B)$, $R : \linops(\hilmaparg B) \mapsto \linops(\hilmaparg C)$ are defined as
    \begin{align}
        O(\rho_A) = \Tr_A[\telepovm_{AB}\rho_A], \quad
        R(\rho_B) = \Tr_B[\telestate_{BC}\rho_B].
    \end{align}
    Since the rank of the identity map $\mathcal I$ is equal to $d_A^2$ (the dimension of $\linops(\hilmaparg A)$, the complex vector space of linear operators on $\hilmaparg A)$), we must have that the rank of $R$ and $O$ is at least equal to $d_A^2$.
    However, the rank of a linear map is at most the minimum of the input and output vector space dimension.
    In particular, the rank of $R$ and $O$ are at most the vector space dimension of $\linops(\hilmaparg B)$, which is equal to $d_B^2$.
    We thus obtain $d_B \geq d_A$ as a necessary condition.

    \noindent We now prove the bulk of the self-testing result.
    We let $\hilmaparg{B'}$ be an ancilla Hilbert space so that there exists a pure state $\ket{\telestate}_{BB'C}$ that purifies $\telestate_{BC}$, i.e.,
    \begin{align}
        \telestate_{BC} = \Tr_{B'}[\ketbra{\telestate}_{BB'C}],
    \end{align}
    This implies the equality of \cref{eq:selftest coeffs}.
    
    We consider the Schmidt decomposition of the state $\ket\telestate_{BB'C}$ with respect to the $BB'|C$ partition: this implies that there exists an orthornormal basis of $\hilmaparg{B}\otimes\hilmaparg{B'}$ denoted $\{\ket{k}_{BB'}\}_{k=1}^{d_Bd_{B'}}$, an orthonormal basis of $\hilmaparg C$ denoted $\{\ket{k}_C\}_{k=1}^{d_A}$, and coefficients $\{0\neq \telestate_k \in \mathbb{R}\}_{k=1}^{d_\telestate}$, where $d_\telestate \leq \min(d_Bd_{B'},d_A)$, such that
    \begin{align}
    \label{eq:selftest 1}
        \ket\telestate_{BB'C} = \sum_{k=1}^{d_\telestate} \telestate_k \ket{k}_{BB'} \otimes \ket{k}_C.
    \end{align}
    The teleportation condition can be rewritten as follows, tracing out over the ancilla: for all density matrix $\rho_A \in \linops(\hilmaparg A)$,
    \begin{align}
    \label{eq:selftest 45}
        \Tr_{BB'}[\telepovm_{AB} \rho_A \telestate_{BB'C}] = \teleprob \rho_C.
    \end{align}
    Now, we define the orthogonal projector
    \begin{align}
        \Pi_{BB'} := \sum_{k=1}^{d_\telestate} \ketbra{k}_{BB'}.
    \end{align}
    This projector is such that $\Pi_{BB'} \ket{\telestate}_{BB'C} = \ket\telestate_{BB'C}$.
    We then define
    \begin{align}
        \telepovm_{ABB'} := (\id_A \otimes \Pi_{BB'}) (\telepovm_{AB} \otimes \id_{B'}) (\id_A \otimes \Pi_{BB'}),
    \end{align}
    which can be checked to be a valid POVM element.
    Using the ``partial cyclicity'' of the partial trace, \cref{eq:selftest 45} is equivalent to 
    \begin{align}
    \label{eq:selftest 98}
        \Tr_{BB'}[\telepovm_{ABB'} \rho_A \telestate_{BB'C}] = \teleprob \rho_C.
    \end{align}
    We let $\{\ket{k}_A\}_{k=1}^{d_A}$ be the orthonormal basis of $\hilmaparg A$ that corresponds to the orthonormal basis $\{\ket{k}_C\}_k$ of $\hilmaparg C$ through the identity that the teleportation protocol assumes, namely, $\hilmaparg A = \hilmaparg C$.
    Since the teleportation probability $\teleprob$ does not depend on $\rho$, we can use the linearity of \cref{eq:selftest 98}, together with the fact that the linear span of density matrices consists of all (possibly non-Hermitian) linear operators on the Hilbert space, to substitute $\ketbraa{i}{j}_{A}$ instead of $\rho_A$.
    Then, we have for all $i,j \in \{1,\dots,d_A\}$,
    \begin{align}
        \Tr_{BB'}[\telepovm_{ABB'} \ketbraa{i}{j}_A \telestate_{BB'C}] = \teleprob \ketbraa{i}{j}_C.
    \end{align}
    Introducing the Schmidt decomposition of \cref{eq:selftest 1}, we obtain
    \begin{align}
    \label{eq:selftest 2}
        \sum_{k,l=1}^{d_\telestate} \telestate_k \telestate_l \bra{j}_A \bra{l}_{BB'} \telepovm_{ABB'} \ket{i}_A \ket{k}_{BB'} \ketbraa{k}{l}_{C} = \teleprob \ketbraa{i}{j}_C.
    \end{align}
    This directly implies that $d_\telestate = d_A$, which in turn implies that $d_Bd_{B'} \geq d_A$.
    Furthemore, we can solve \cref{eq:selftest 2} for $\telepovm_{ABB'}$: for all $i,j \in \{1,\dots,d_A\}$, for all $l,k \in \{1,\dots,d_A\}$,
    \begin{align}
        \bra{j}_A \bra{l}_{BB'} \telepovm_{ABB'} \ket{i}_A \ket{k}_{BB'} = \frac{\teleprob}{\telestate_k\telestate_l} \delta_{ik}\delta_{jl}.
    \end{align}
    This yields
    \begin{align}
        \telepovm_{ABB'} = \Pi_{BB'} \telepovm_{ABB'} \Pi_{BB'} = \teleprob \sum_{k,l=1}^{d_A} \frac{1}{\telestate_k\telestate_l} \ketbraa{k}{l}_A \otimes \ketbraa{k}{l}_{BB'} =: \ketbra{\telepovm}_{ABB'},
    \end{align}
    where we defined the following sub-normalized state:
    \begin{align}
        \ket{\telepovm}_{ABB'} = \sqrt{\teleprob} \sum_{k=1}^{d_A} \frac{1}{\telestate_k} \ket{k}_A \ket{k}_{BB'}.
    \end{align}
    Note that for $\telepovm_{ABB'}$ to be a valid POVM element, we need $\telepovm_{ABB'} \leq \id_{ABB'}$, which implies that the norm of the above state must be at most one, implying the inequality of \cref{eq:selftest coeffs}.
\end{proof}
We can now prove that the optimal teleportation probability is given by $1/d_A^2$, which is achieved by the Bell state preparation and measurement of \cref{def:bell tele}.
\begin{lemma}[Maximal teleportation probability]
    Given $\hilmaparg A = \hilmaparg C$, the maximal teleportation probability $\teleprob$ that can be achieved within \cref{def:ps teleportation} is $1/d_A^2$.
\end{lemma}
\begin{proof}
    First, we have seen that the teleportation probability $\teleprob = 1/d_A^2$ is always achievable by the Bell state implementation of \cref{def:bell tele}.
    Indeed, we can put this implementation in the form of \cref{lemma:self test}: we have $\telestate_k = \frac{1}{\sqrt{d_A}}$, and the pure state $\ket{E}_{BB'C} = \ket{\bellstate}_{BB'C}$ defining the POVM element of the implementation can be written as
    \begin{align}
        \ket{E}_{BB'C} = \sqrt{\frac{1}{d_A^2}} \sum_{k=1}^{d_A} \sqrt{d_A} \ket{k}_{BB'} \otimes \ket k_{C},
    \end{align}
    thus showing that for this protocol, $\teleprob = 1/d_A^2$.

    \noindent We now wish to show that any teleportation probability is upper-bounded by $1/d_A^2$.
    For a vector $\vec\telestate = (\telestate_k)_k$, we denote with $\twonorm{\vec\telestate}^2 = \sum_k \telestate_k^2$ the Euclidean norm squared.
    First, we argue that the maximal teleportation probability is given by
    \begin{align}
        &\max\Biglset{\teleprob}{\exists\vec\telestate\in\mathbb{R}^{d_A} \st \telestate_k \neq 0, \twonorm{\vec\telestate}^2 = 1 \textup{ and } \textstyle\sum_{k=1}^{d_A} \frac{1}{\telestate_k^2} \leq \frac{1}{\teleprob}} \\
        = &\max\Biglset{\teleprob}{\exists\vec\telestate\in\mathbb{R}^{d_A} \st \telestate_k \neq 0, \twonorm{\vec\telestate}^2 = 1 \textup{ and } \textstyle\sum_{k=1}^{d_A} \frac{1}{\telestate_k^2} = \frac{1}{\teleprob}} \label{eq:max tele 1} \\
        \leq &\max\Biglset{\teleprob}{\exists\vec\telestate, \vec\psi\in\mathbb{R}^{d_A} \st \twonorm{\vec\telestate}^2 = 1 \textup{ and } \vec\telestate\cdot\vec\psi = d_A \textup{ and } \twonorm{\vec\psi}^2 = \textstyle\frac{1}{\teleprob}} \label{eq:max tele 2},
    \end{align}
    where the last inequality follows from seeing that any $\teleprob$ achievable for some $\vec\telestate$ in \cref{eq:max tele 1} can be achieved by using the same $\vec\telestate$ in \cref{eq:max tele 2} together with setting $\psi_k = 1/\telestate_k$.
    To maximize $\teleprob$, we need to minimize the norm of $\vec\psi$.
    However, subject to the constraints $\vec\telestate\cdot\vec\psi = d_A$ and $\twonorm{\vec\telestate} = 1$, this means that we need to pick $\vec\psi = d_A\vec\telestate$, so that $\twonorm{\vec\psi}^2 = d_A^2$, and the maximum teleportation probability is indeed $1/d_A^2$.
\end{proof}

We now prove the following lemma, which will allow us to prove that the probabilities associated to cyclic causal models according to our proposed probability rule do not depend on the choice of post-selected teleportation implementation.

\newcommand{\proofmap}{\mathcal M}

\begin{lemma}[Cyclic composition is independent of post-selected teleportation implementation]
\label{lemma:cyclic indep of tele implementation}
    It holds that for all post-selected teleportation implementation $(\telepovm_{AB},\telestate_{BC})$ with associated teleportation probability $\teleprob$ (see \cref{def:ps teleportation}),
    \begin{align}
        \Tr_{AB}[\telepovm_{AB} \proofmap_{A|C}(\telestate_{BC})] = \teleprob\,\selfcycle(\proofmap_{A|C}),
    \end{align}
    where $\selfcycle(\proofmap_{A|C})$ is defined in~\cref{def:selfcycle_v3}.
\end{lemma}
\begin{proof}
    Consider an arbitrary post-selected teleportation implementation $(\telepovm_{AB},\telestate_{BC})$ with associated teleportation probability $\teleprob$, and let $\hilmaparg{B'}$, $\{\ket{k}_A\}_{k=1}^{d}$, $\{\ket{k}_{BB'}\}_{k=1}^{d}$, $\ket\telestate_{BB'C}$ and $\ket\telepovm_{ABB'}$ be as in \cref{lemma:self test}.
    We then have
    \begin{align}
        &\Tr_{AB}[\telepovm_{AB}\proofmap_{A|C}(\telestate_{BC})] \nonumber\\
        &= 
        \Tr_{ABB'}[\ketbra\telepovm_{ABB'} \proofmap_{A|C}(\ketbra{\telestate}_{BB'C})] \nonumber\\
        &= 
        \sum_{k,l,m,n=1}^{d} 
        \Tr_{ABB'}\left[\frac{\teleprob}{\telestate_m\telestate_n} \ketbraa{m}{n}_A \otimes \ketbraa{m}{n}_{BB'} \telestate_k\telestate_l \ketbraa{k}{l}_{BB'} \otimes \proofmap_{A|C}(\ketbraa{k}{l}_{C})\right]
        \nonumber\\
        &= 
        \teleprob\,\selfcycle(\proofmap_{A|C}). \qedhere
    \end{align}
\end{proof}
\section{Links to other quantum causality frameworks}
\label{sec:quantum aspects}

\subsection{Barrett Lorenz Oreshkov cyclic causal models}

In this section, we outline relationships and distinctions between our framework and the cyclic quantum causal modelling formalism of Barrett, Lorenz and Oreshkov (BLO) \cite{Barrett_2021}. We show that whenever the output space of each party $A_i$ factorises into tensor factors, one for each child of the vertex $A_i$, then such causal models in the BLO framework can be recovered faithfully within our framework. This tensor factor restriction of the BLO framework can be seen as a natural cyclic generalization of the acyclic causal modelling framework of Costa and Shrapnel \cite{Costa_2016}, thus in particular, the latter is also recovered in our framework. We then discuss the prospects for faithfully mapping the more general case (without the tensor factor restriction) 
into our framework and its links to the causal decomposition problem. 

\subsubsection{Review of the BLO formalism}
We begin by outlining how cyclic causal models are defined in the BLO formalism, the main ingredients are the Choi-Jamiolkowski isomorphism~\cite{Choi_1975,Jamiolkowski_1972} and the concept of process operators. 
\paragraph{Choi-Jamiolkowski isomorphism.}
With every linear completely positive (CP) map, $\chanmap_{A\mapsto B}: \linops(\hilmap_A)\mapsto \linops(\hilmap_B)$, we can associate a corresponding Choi-Jamiolkowski (CJ) operator, 
\begin{equation}
  \rho_{B|A}=\sum_{i,j}\chanmap_{A\mapsto B}\left(\ket{i}\bra{j}_A\right)\otimes \ket{i}\bra{j}_{A^*}\in \linops(\hilmap_B)\otimes\linops(\hilmap_{A}^*),  
\end{equation} where $\{\ket{i}_A\}$ is an orthonormal basis of $\hilmap_A$ and  $\{\ket{i}_{A^*}\}$ is the corresponding dual basis (an orthonormal basis of the dual space $\hilmap_{A}^*$). When the map is completely positive and trace preserving (CPTP), we have $\Tr_B(\rho_{B|A})=\id_{A^*}$. The inverse isomorphism is given as follows and allows us to compute the action of the channel $\mathcal{E}_{A\mapsto B}$ on an input state $\rho_A$: $\mathcal{E}_{A\mapsto B}(\rho_A)=\Tr[\tau^{\textup{id}}_{AA^*}\rho_{B|A}\rho_A]$, where $\tau^{\textup{id}}_{AA^*}=\sum_{i,j}\ket{i}\bra{j}_A\otimes \ket{i}\bra{j}_{A^*}$ is known as the linking operator. This equation has a striking similarity to classical conditional probabilities (which describe the action of classical channels) $\undefprob(Y)=\sum_X\undefprob(Y|X)\undefprob(X)$ \cite{Leifer_2013}, where $\undefprob$ denotes an arbitrary probability distribution over two random variables $X$ and $Y$.

There are multiple versions of the CJ isomorphism (see for instance \cite{Frembs_2024}), the above version is chosen in the BLO formalism as it is basis independent and yields a positive CJ operator. In particular, the more commonly used version is 
\begin{equation}
\label{eq: choi1}
    \rho^{(1)}_{B|A}=\sum_{i,j}\chanmap_{A\mapsto B}(\ket{i}\bra{j}_A)\otimes \ket{i}\bra{j}_{A}\in \linops(\hilmap_B)\otimes\linops(\hilmap_{A}),
\end{equation} 
which is positive but basis dependent, and another version is given by 
\begin{equation}
\rho_{B|A}^{(2)}:=\sum_{i,j}\chanmap_{A\mapsto B}(\ket{i}\bra{j}_A)\otimes \ket{j}\bra{i}_{A}\in \linops(\hilmap_B)\otimes\linops(\hilmap_{A}),
\end{equation} 
which is basis-independent but not necessarily positive. The inverse isomorphisms in these cases are given as $\mathcal{E}_{A\mapsto B}(\rho_A)=\Tr[\rho^{(1)}_{B|A}(\id_B\otimes \rho_A^T)]=\Tr[\rho^{(2)}_{B|A}(\id_B\otimes \rho_A)]$.

\paragraph{Process operators and probabilities.}
The next ingredient is the concept of process operators, which was previously used in the process matrix~\cite{Oreshkov_2012}
and higher-order quantum process~\cite{Chiribella_2013} frameworks. The idea is to consider multiple laboratories $A_1,...,A_N$, each $A_i$ being associated with a pair of in and output spaces $A_i^{\textup{in}}:=\linops\left(\hilmap_{A_i^{\textup{in}}}\right)$, $A_i^{\textup{out}}:=\linops\left(\hilmap_{A_i^{\textup{out}}}\right)$. An agent may perform a local quantum operation within each such lab, which corresponds to a quantum instrument, i.e., a set of CP maps $\{\mathcal{M}_{\outcome_i}:A_i^{\textup{in}}\mapsto A_i^{\textup{out}}\}_{\outcome_i}$ such that $\mathcal{M}:=\sum_{\outcome_i}\mathcal{M}_{\outcome_i}$ is CPTP. Generally, the local operations can be parametrised by a classical setting choice $\setting_i$ and we have a corresponding quantum instrument $\{\mathcal{M}_{\outcome_i|\setting_i}\}_{\outcome_i}$. The process operator $\sigma^{A_1,...,A_N}\in \linops\left(\bigotimes_i\hilmap_{A_i^{\textup{in}}}\bigotimes_i\hilmap_{A_{i}^{\textup{out}}}^*\right)$ describes the environment of these labs, and encodes information on how they are connected, for instance whether or not there is a channel connecting the output $A_1^{\textup{out}}$ of lab $A_1$ to the input $A_2^{\textup{in}}$ of the next lab $A_2$. 
The outcome probabilities can be computed as follows
\begin{equation}   
\label{eq: BLO_probability}\textup{Pr}_{\textup{BLO}}(\outcome_1,...,\outcome_N|\setting_1,...,\setting_N)=\Tr\left[\sigma^{A_1,...,A_N}\Big(\bigotimes_{i=1}^n\rho_{A^{\textup{out}}_i|A^{\textup{in}}_i}^{\outcome_i|\setting_i}\Big)^T\right],
\end{equation}
where $\rho_{A^{\textup{out}}_i|A^{\textup{in}}_i}^{\outcome_i|\setting_i}$ denotes the CJ operator of the local operation $\mathcal{M}_{\outcome_i|a_i}$.

The conditions for $\sigma^{A_1,...,A_N}$ being a valid process operator\footnote{These include that $\sigma^{A_1,...,A_N}>0$ and $\Tr\Big[\sigma^{A_1,...,A_N}\Big(\bigotimes_{i=1}^n\rho_{A^{\textup{out}}_i|A^{\textup{in}}_i}\Big)^T\Big]=1$ when $\rho_{A^{\textup{out}}_i|A^{\textup{in}}_i}$ are CJ operators of CPTP maps.} ensure that the distribution defined above is a valid, normalised probability distribution. See \cite{Araujo_2015} and \cite{Oreshkov_2016} for further details on necessary and sufficient constraints on valid process operators. In the above, all output spaces appearing in the CJ operators $\rho_{A^{\textup{out}}_i|A^{\textup{in}}_i}$ and the process operator $\sigma^{A_1,...,A_N}$ are the dual spaces $\hilmap^*_{A_i^{\textup{out}}}$. 

We could alternatively use the CJ operators according to the variant $\rho^{(1)}_{A^{\textup{out}}_i|A^{\textup{in}}_i}$ of the CJ isomorphism where the dual spaces do not feature, and define the process operator on the joint in and output spaces of all parties. One can show that the probability expression \cref{eq: BLO_probability} will be identical. 
However, in the latter case, one must be careful to express all operators and take the transpose in the same basis. We are now ready to review the definition of causal models in this framework.

\begin{definition}[BLO-QCM \cite{Barrett_2021}]
\label{def:BLO-QCM}
  A BLO-QCM is given by
  \begin{myitem}
      \item A causal structure which corresponds to a directed graph $\graphname$ with vertices $A_1,...,A_N$ 
      \item For each $A_i$, a quantum channel (in CJ representation) $$\sigma_{A_i|\parnodes{A_i}}\in \linops\left(\hilmap_{A_i^{\textup{in}}}\otimes\left( \bigotimes_{A_k\in \parnodes{A_i}} \hilmap^*_{A_k^{\textup{out}}}\right)\right)$$ where $\parnodes{A_i}$ denotes the set of all parents of $A_i$ in $\graphname$, such that $$[\sigma_{A_i|\parnodes{A_i}},\sigma_{A_j|\parnodes{A_j}}]=0$$ for all $i,j$ and $\sigma^{A_1,...,A_N}=\prod_{i=1}^N\sigma_{A_i|\parnodes{A_i}}$ is a valid process operator.
  \end{myitem}
\end{definition}

Note that the original papers \cite{Barrett_2019, Barrett_2021} use $\rho_{A_i|\parnodes{A_i}}$ rather than $\sigma_{A_i|\parnodes{A_i}}$ for the channels mentioned in the above definition. We employ the latter notation in order to make it clear that these are internal channels of the process operator as opposed to the channels $\rho_{A_i^{\textup{out}}|A_i^{\textup{in}}}$ corresponding to the external operations performed by the parties.

The BLO framework focuses process operators $\sigma^{A_1,...,A_N}$ which correspond to the CJ operator of a unitary channel from the outputs of all parties to their inputs. Generally, given a channel $\chanmap$, the dependence of an output $O$ of a channel $\chanmap$ on an input $I$ is checked by whether a choice $\mathcal{M}_I$ of local operation on $I$ can lead to distinguishable states on $O$, i.e., we have signalling from $I$ to $O$ through $\chanmap$ when $\Tr_{\backslash O}\circ \chanmap\circ \mathcal{M}_I\neq \Tr_{\backslash O}\circ \chanmap$, where $\backslash O$ denotes all outputs of $\chanmap$ except $O$ (see \cite{Ormrod_2023} for equivalent definitions).
If the channel is unitary, the signalling relations between in and outputs form a directed acyclic graph where the in and output systems are the vertices \cite{Barrett_2019}. This is not the case for general CPTP maps, where we can have non-trivial signalling relations between subsets of systems without signalling between individual systems, e.g., an input $I$ of a (non-unitary) CPTP map $\chanmap$ can signal to outputs $O_1$ and $O_2$ jointly but not individually.  

Therefore, for unitary channels one can construct the edges of the graph through signalling relations, i.e., the parents of every output system $S$, $\parnodes{S}$, correspond to inputs that signal to $S$ through the unitary channel. 
If we apply such construction to the unitary channel of the process operator, the factorisation $$\sigma^{A_1,...,A_N}=\prod_{i=1}^N\sigma_{A_i|\parnodes{A_i}}$$ can be derived as a theorem (Theorem 4.3 of \cite{Barrett_2019}). 
The factors are marginals of the full channel and commute with each other as required by~\cref{def:BLO-QCM}. This motivates using such signalling relations to define causal relations in the case of unitary channels as processes --- which is the case in the BLO framework. For further details on causality in unitary channels, we refer the reader to \cite{Allen_2017, Barrett_2019, Barrett_2021}. For details on the inequivalence between causation and signalling in non-unitary channels, we refer the reader to the above and to \cite{VilasiniRennerPRA}.

\subsubsection{Tensor restriction of the BLO formalism and mapping to our framework}
\label{sec: BLOtoOurs}
At this point, we can observe an important distinction between causal models in ours and the BLO formalism. Consider a simple common cause graph where we have vertices $A$, $B$ and $C$ and directed edges $(A,B)$ and $(A,C)$. In our formalism, any causal model on this graph will have a channel $\chanmaparg{A}$ which has two outgoing wires, one for each outgoing edge and no incoming wires as $A$ is exogenous (i.e., a bipartite state preparation), while $B$ and $C$ will be associated with channels $\chanmaparg{B}$ and $\chanmaparg{C}$ acting on the sub-systems of the bipartite state sent along the respective edges, i.e., \begin{equation}
    \centertikz{
        \node (A) at (0, 0) {$A$};
        \node (B) at (-1, 0.5) {$B$};
        \node (C) at (1, 0.5) {$C$};
        \draw[->,-stealth] (A) -- (B);
        \draw[->,-stealth] (A) -- (C);
    }.
\end{equation}
In the BLO formalism, a common cause structure emanating from $A$ does not generally imply a channel at $A$ with multiple outgoing wires (each associated with a separate system). 
More specifically, considering the three vertices as three labs and using trivial in/output spaces for exogenous or childless vertices, we notice that the BLO framework associates a single global channel from the output space $A^{\textup{out}}$ of $A$ to the inputs $B^{\textup{in}}$ and $C^{\textup{in}}$ whose CJ operator is $\sigma_{BC|A}\in \linops(\hilmap_{B^{\textup{in}}}\otimes \hilmap_{C^{\textup{in}}}\otimes \hilmap^*_{A^{\textup{out}}} )$ (this would in fact be the process operator, if we ignore the trivial spaces), with commuting marginals $[\sigma_{B|A},\sigma_{C|A}]=0$ (keeping factors of identity implicit), $\sigma_{B|A}=\Tr_C(\sigma_{BC|A})$ and $\sigma_{C|A}=\Tr_B(\sigma_{BC|A})$. Generally, $A$ could be a qubit and need not factorise into independent tensor factors corresponding to two wires. 

In the BLO formalism, a vertex $A$ corresponds to (pairs of) systems, $\hilmap_{A^{\textup{in}}}$ and $\hilmap_{A^{\textup{out}}}$ and the edge structure specifies the channels $\sigma_{A|\parnodes{A}}$ that connect the in/output systems between vertices. By contrast, in our formalism vertices $\vertname$ are associated with channels $\chanmaparg\vertname$ while edges $\edgename$ are associated with systems $\hilmaparg{\edgename}$. How should we compare and map between the frameworks in a faithful manner that preserves the relevant information?

\paragraph{Faithful mappings and tensor-restriction of BLO models.} We now show that if we require a tensor factorisation of output spaces in the BLO framework in terms of the structure of outgoing edges, then there is a faithful mapping into out framework for every choice of local operations. More formally, we first define these concepts, starting with defining when a mapping between two quantum causality formalisms can be considered faithful. As different frameworks formalize causal models in distinct ways, we will use a more general terminology that captures the essential features of such models: a \emph{quantum causal description} consists of a directed graph together with an assignment of causal mechanisms (quantum channels) to that graph, and a rule for computing observed correlations arising from those mechanisms. 

\begin{definition}[Faithful mapping between quantum causal descriptions]
\label{def: faithful_mapping}
We say that a mapping from one quantum causal description to another is faithful if
\begin{enumerate}
    \item[(a)] all the vertices and associated causal mechanisms of the former are preserved in the image of the mapping,
    \item[(b)] if there is no directed path from $A_i$ to $A_j$ in the former, then there will be no directed path from $A_i$ to $A_j$ in the latter,
    \item[(c)]  the probability rules associated with both descriptions lead to the same observed correlations.
\end{enumerate}   
\end{definition}

This usage of the term faithful is also consistent with its formal usage in causal modelling and inference parlance. There, a causal model is said to be faithful if whenever there is connectivity in the graph, there is a corresponding conditional dependence in the causal model. If we mapped a BLO-QCM to our framework in a way that reproduces the operational predictions but there are additional directed paths in the causal graph of the image than the original model specifies, then the final causal model will generally not be faithful to the graph as it will contain additional observable independences not reflected in the connectivity (see \cref{sec:process_matrices} for further details). Note that any reasonable graph-separation notion would regard $A_i$ and $A_j$ as connected when there is a directed path from one to the other, i.e., when one is a cause of the other ($d$, $\sigma$ and $p$ separation satisfy this).

Next, we define the tensor-restriction of BLO quantum causal models as follows. 
\begin{definition}[BLO-QCM$_{\otimes}$]
\label{def: BLOTensor}We define a tensor product restricted BLO-QCM, denoted BLO-QCM$_{\otimes}$ as being specified by the following.
\begin{myitem}
 \item A causal structure which corresponds to directed graph $\graphname$ with vertices $A_1,...,A_N$.
 \item For each $A_i$, a tensor factorisation of its output space as $\hilmap_{A_i^{\textup{out}}}=\bigotimes_{A_k\in \childnodes{A_i}}\hilmap_{A_{ik}^{\textup{out}}}$ where $\childnodes{A_i}$ denotes the set of all children of $A_i$ in $\graphname$ and $\hilmap_{A_{ik}^{\textup{out}}}$ are arbitrary finite dimensional Hilbert spaces.
 \item For each $A_i$, a quantum channel $$\sigma_{A_i|\parnodes{A_i}}\in \linops\left(\hilmap_{A_i^{\textup{in}}}\otimes\left(\bigotimes_{A_k\in \parnodes{A_i}} \hilmap^*_{A_{ki}^{\textup{out}}}\right)\right)$$ such that $\sigma^{A_1,...,A_N}=\prod_{i=1}^N\sigma_{A_i|\parnodes{A_i}}$ is a valid process operator. 
 \end{myitem}
\end{definition}
Since each channel $\sigma_{A_i|\parnodes{A_i}}$ acts on a distinct Hilbert space (distinct tensor factors of the parental Hilbert space), the commutation condition $[\sigma_{A_i|\parnodes{A_i}},\sigma_{A_j|\parnodes{A_j}}]=0$ is satisfied for all $i,j$.
Hence, every BLO-QCM$_{\otimes}$ is indeed an instance of a general BLO-QCM.

\paragraph{Faithfully mapping a tensor restricted BLO-QCM to our formalism.} We now show that 
any BLO-QCM$_{\otimes}$ (\cref{def: BLOTensor}) can be faithfully mapped (\cref{def: faithful_mapping}) into our framework. Given a BLO-QCM$_{\otimes}$ together with a choice of quantum instrument $\mathcal{M}_{\setting_i}:=\{\mathcal{M}_{\outcome_i|\setting_i}: \linops(\hilmap_{A_i^{\textup{in}}})\mapsto \linops(\hilmap_{A_i^{\textup{out}}})\}_{\outcome_i}$, one for each agent $A_i$, we can obtain a causal model in our formalism (\cref{def:causal model}) through the mapping we describe below.
Therefore, each BLO-QCM$_{\otimes}$ maps to a family of QCMs in our framework, with each element in the family corresponding to a fixed choice of instruments for the $N$ parties, labelled by the settings \{$\setting_i\}_i$. A previous causal modelling framework by Costa and Shrapnel \cite{Costa_2016} can be seen as a further restriction of BLO-QCM$_{\otimes}$ above to the case of directed acyclic graphs $\graphname$. Hence, the mapping we give here shows that the Costa and Shrapnel formalism can also be faithfully recovered within ours.

\begin{itemize}
    \item {\bf Causal graph} If $\graphname$ is the graph associated with the given BLO-QCM$_{\otimes}$, by construction it has $N$ vertices $\{A_1,...,A_N\}$. Without loss of generality, let $\{A_1,...,A_k\}$ denote the vertices associated with a trivial (1-dimensional) input Hilbert space. 
    
    Define a new causal graph $\graphname'$ constructed from $\graphname$ as follows with $2N+k$ vertices: $\{A_1,...,A_N\}\cup \{\sigma_1,...,\sigma_k\}\cup\{X_1,...,X_N\}$, where $\{X_1,...,X_N\}$ are observed while the rest are unobserved vertices. Whenever $(A_i,A_j)$ is a directed edge in $\graphname$, $A_j$ cannot have a trivial input space, and we include a corresponding directed edge $(A_i,\sigma_j)$ in $\graphname'$. Further, include directed edges $(\sigma_j,A_j)$ for every $j\not\in\{1,...,k\}$ and $(A_j,X_j)$ for all $j\in\{1,...,N\}$ in $\graphname'$. This construction ensures that 
    absence of directed paths in $\graphname$ implies the same in $\graphname'$ as required by \cref{def: faithful_mapping}.
    \item {\bf State spaces} We now associate a state space to each edge on the causal graph $\graphname'$ as required by \cref{def:causal model}. Consider the edges $(A_i,\sigma_j)$ of $\graphname'$, which correspond to edges $(A_i,A_j)$ of $\graphname$. Since each BLO-QCM$_{\otimes}$ assigns a tensor factor $\hilmap_{A_{ij}^{\textup{out}}}$ of $A_i$'s output space $\hilmap_{A^{\textup{out}}_i}$ to every child $A_j$ of $A_i$, we can assign this space to the edge $(A_i,\sigma_j)$ of $\graphname'$. The edges $(\sigma_j,A_j)$ are incoming to the lab $A_j$ and will be associated with the input space $\hilmap_{A_j^{\textup{in}}}$. The edges $(A_j,X_j)$ model the influence of each party's action on their classical outcome, and will be associated a Hilbert space $\hilmap(\outcomemap_j)$ which encodes the outcome set of the $j$-th party $\outcomemap_j:=\{\outcome_j\}_{\outcome_j}$ in a preferred basis.
        \item {\bf Quantum channels} In this mapping, we will ignore the dual spaces of the BLO framework. This is not an issue, as discussed earlier, as long as we are consistent about taking transposes in the same basis in which the CJ operators are expressed. Each of the $\sigma_i$ vertices is assigned the quantum channel $\hat{\sigma}_{i}: \linops(\hilmap_{A_i^{\textup{in}}})\mapsto \linops(\bigotimes_{A_k\in \parnodes{A_i}}\hilmap_{A_{ki}^{\textup{out}}})$ whose CJ operator is the operator $\sigma_{A_i|\parnodes{A_i}}$ of the BLO framework\footnote{According to the basis-dependent CJ representation without the dual spaces (\cref{eq: choi1}). We referred to this with a superscript 1 before but drop this to avoid clutter, since the meaning is clear from context.}. Each of the $A_i$ vertices is assigned the corresponding local quantum instrument, $\mathcal{M}_{\setting_i}$ with a slight modification to make explicit the outcome as a separate system. The modified instrument is associated with a CPTP map $\tilde{\mathcal{M}}_{\setting_i}:\linops(\hilmap_{A_i^{\textup{in}}})\mapsto \linops(\hilmap_{A_i^{\textup{out}}}\otimes \hilmap(\outcomemap_i))$ defined as follows (where each term in the sum below is a CP map $\tilde{\mathcal{M}}_{\outcome_i|\setting_i}$, the set of which defines the instrument)
    \begin{equation}
        \tilde{\mathcal{M}}_{\setting_i}(\rho_{A^{\textup{in}}_i}):=\sum_{\outcome_i}\mathcal{M}_{\outcome_i|\setting_i}(\rho_{A_i^{\textup{in}}})\otimes \ket{\outcome_i}\bra{\outcome_i}.
    \end{equation}

    Finally, the $X_i$ vertices are associated with a set of maps $\{\mathcal{E}_{\outcome_i}\}_{\outcome_i}$ specified by a POVM which has elements $E_{\outcome_i}:=\ket{\outcome_i}\bra{\outcome_i}$, according to \cref{eq:def_onode} (we have used $\mathcal{E}$ for this map as opposed to $\mathcal{M}$ of the defining equation, in order to distinguish this from the maps assigned to the labs). Note that the in and output spaces of all channels in this step match with the state spaces assigned to the in/outgoing edges in the previous step.
    \item {\bf Probabilities} The above fully specifies a causal model according to \cref{def:causal model}, for every BLO-QCM$_{\otimes}$ and choice $\{\setting_i\}_{i=1}^N$ of local instruments. We can now show that the probabilities $\textup{Pr}_{\textup{BLO}}(\outcome_1,...,\outcome_N|\setting_1,...,\setting_N)$ as computed in the BLO formalism (\cref{eq: BLO_probability}) and using the probability rule of our formalism (\cref{sec:probability rule CM}) would be the same.
\end{itemize}

For simplicity and for sake of illustration, we make the argument for a particular example. However, the argument readily generalises to all graphs that support BLO-QCM$_{\otimes}$.
Consider the graph $\graphname$ in the BLO framework with 4 vertices $A$, $B$, $C$ and $D$ and edges $(D,B)$, $(A,B)$ and $(A,C)$ i.e., 
\begin{equation}
    \graphname=\centertikz{
        \node (D) at (-2,0) {$D$};
        \node (A) at (0,0) {$A$};
        \node (B) at (-1,1) {$B$};
        \node (C) at (1,1) {$C$};
        \draw[->,-stealth] (D)--(B); 
        \draw[->,-stealth] (A)--(B); 
        \draw[->,-stealth] (A)--(C);
    }.
\end{equation}
Let $D$ and $A$ have trivial input spaces. Then the corresponding causal graph $\graphname'$ in our formalism would be as follows (it has $10$ vertices as $N=4$, $k=2$):
\begin{equation}
    \graphname'=\centertikz{
    \node[unode] (D) at (0,0) {$D$};
     \node[unode] (C) at (3,0) {$A$};
      \node[unode] (sA) at (1,2) {$\sigma_B$};
       \node[unode] (sB) at (3,2) {$\sigma_C$};
        \node[unode] (A) at (0,4) {$B$};
         \node[unode] (B) at (3,4) {$C$};
    \node[onode] (xD) at (-1,1) {$X_D$};
       \node[onode] (xC) at (4,1) {$X_A$};
           \node[onode] (xA) at (-1,5) {$X_B$};
               \node[onode] (xB) at (4,5) {$X_C$};
        \draw[qleg] (D)--(sA); \draw[qleg] (C)--(sA); \draw[qleg] (C)--(sB); \draw[qleg] (sA)--(A); \draw[qleg] (sB)--(B);
\draw[qleg] (A)--(xA); \draw[qleg] (B)--(xB); \draw[qleg] (C)--(xC); \draw[qleg] (D)--(xD);
    }.
\end{equation}
Notice that the absence of directed paths from $D$ to $C$ in $\graphname$ is also reflected in $\graphname'$, in accordance with \cref{def: faithful_mapping}.
We now explain why the observed probabilities are the same under this mapping. 

Consider a causal model on the causal graph $\graphname'$ obtained from a BLO-QCM$_{\otimes}$ through the mapping that we have described above. We can construct an acyclic causal graph $\overline{\graphname'}\in \graphfamily{\graphname'}$ by replacing each of the edges $(A_i,\sigma_j)$ of $\graphname'$ (which arise from the edges $(A_i,A_j)$ of the BLO-QCM$_{\otimes}$) with a post-selected teleportation protocol. Suppose there are $k$ edges of the type $\centertikz{\node[unode] (a) at (0,0) {$A_i$}; \node[unode] (s) at (2,0) {$\sigma_j$}; \draw[qleg] (a)--(s);}$. We then have the pre and post-selection vertices $\{\prevertname_1,\dots,\prevertname_k\} = \prevertset$ and $\{\postvertname_1,\dots,\postvertname_k\} = \psvertset$, and each such edge is replaced with the following protocol for a pair $(\prevertname_l,\postvertname_l)$:

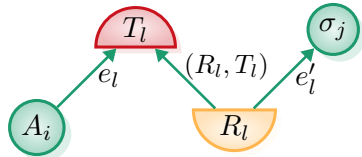
   \begin{align*}
        \centertikz{
        \begin{scope}[xscale=2.6,yscale=1.3]
        \node[unode] at (0,0) {$A_i$}; \node[unode] (s) at (1.5,1) {$\sigma_j$};
            \node[prenode] (pre) at (1,0) {$\prevertname_l$};
            \node[psnode] (post) at (0.5,1) {$\postvertname_l$};
            \draw[qleg] (pre) -- node[pos=0.7,right] {\small$\edgearg{\prevertname_l}{\postvertname_l}$} (post);
            \draw[qleg] (a) -- node[anchor=west]{$\edgename_l$} (post);
            \draw[qleg] (pre) -- node[anchor=west] {$\edgename_l'$} (s);
        \end{scope}
        }
    \end{align*}

    Then we can define a family of causal models on the acyclic causal graph $\overline{\graphname'}$ based on any given causal model on the original graph $\graphname'$: for all the original edges and vertices, all the causal mechanisms are the same and we assign a post-selection teleportation protocol specified by pre and post-selection mechanisms $(\telepovm^l,\telestate^l)$ to each pair of pre and post-selection vertices $(\prevertname_l,\postvertname_l)$ and associated edges. This is a family since we have a choice over the pre and post-selection mechanisms $(\telepovm^l,\telestate^l)$ that result in a post-selected teleportation protocol (\cref{def:ps teleportation}). 
Using the acyclic probability rule, the required probability $\prob(\outcome_1,...,\outcome_N|\setting_1,..,\setting_N)_{\graphname'}$ of the original causal model on $\graphname'$ is then given as follows.
    \begin{equation}
       \label{eq:prob_mapping_from_BLO}
       \begin{split}
           \prob(\outcome_1,...,\outcome_N|&\setting_1,..,\setting_N)_{\graphname'}=\\
       &\frac{\bigotimes_{l=1}^k\telepovm^l \bigcomp_{i=1}^N\mathcal{E}_{\outcome_i} \bigcomp_{i=1}^N\tilde{\mathcal{M}}_{\outcome_i|\setting_i}\bigcomp_{i=1}^N\hat{\sigma}_i(\bigotimes_{l=1}^k\telestate^l)}{\sum_{\outcomealt_1,...,\outcomealt_N}\bigotimes_{l=1}^k\telepovm^l \bigcomp_{i=1}^N\mathcal{E}_{\outcomealt_i} \bigcomp_{i=1}^N\tilde{\mathcal{M}}_{\outcomealt_i|\setting_i}\bigcomp_{i=1}^N\hat{\sigma}_i(\bigotimes_{l=1}^k\telestate^l)}.
       \end{split}
    \end{equation}

In the above, the composition rule is dictated by the acyclic graph $\overline{\graphname'}$ (\cref{def: acyclic probability}) and some factors of identity have been suppressed for brevity. The fraction form of the above expression is due to the usual conditional probability rule (Bayes rule). The numerator is the joint probability of obtaining the outcomes $\outcome_i$ and the post-selections given by $\telepovm^l$ succeeding, while the denominator is just the probability of the latter. The fraction thus gives us the probability of the outcomes conditioned on post-selection success, which is what yields the required probabilities of the cyclic model (\cref{def: probability distribution v3}). 

In \cref{prop:probs as self cycles_v3}, it is shown that this expression can be equivalently written in terms of self-cycle composition, which is independent of the choice of $\telepovm^l$ and $\telestate^l$. For this, define $\etot_{\outcome_1,...,\outcome_N|\setting_1,...,\setting_N}:=\bigcomp_{i=1}^N\mathcal{E}_{\outcome_i} \bigcomp_{i=1}^N\tilde{\mathcal{M}}_{\outcome_i|\setting_i}\bigcomp_{i=1}^N\hat{\sigma}_i$ and notice that
$\etot_{\outcome_1,..,\outcome_N} : \textstyle\linops(\bigotimes_{l=1}^k \hilmaparg{\edgename_l'}) \mapsto \linops(\bigotimes_{l=1}^k \hilmaparg{\edgename_l})$. Then it follows that the probabilities can be written as follows where $\selfcycle$ links the space of each edge $\edgename_l$ to that of $\edgename'_l$ through a self-cycle composition (\cref{def:selfcycle_v3}),
\begin{equation}
 \prob(\outcome_1,...,\outcome_N|\setting_1,..,\setting_N)_{\graphname'}=\frac{\selfcycle(\etot_{\outcome_1,...,\outcome_N|\setting_1,...,\setting_N})}{\sum_{\outcome_1,...,\outcome_N}\selfcycle(\etot_{\outcome_1,...,\outcome_N|\setting_1,...,\setting_N})}.  
\end{equation}

The composition operation $\selfcycle$ here is identical to loop composition as defined in \cite{Portmann_2017}, which takes on the form given in \cite{VilasiniRennerPRA} for finite dimensional Hilbert spaces.\footnote{Here, we use $\selfcycle$ only for compositions relevant for probability computations where the result of the composition is a number (a map with no in or outputs).} It follows from the results of \cite{VilasiniRennerPRA} (Lemma 2 of arXiv v4) that $\selfcycle(\etot_{\outcome_1,...,\outcome_N|\setting_1,...,\setting_N})$ equals the right hand side of \cref{eq: BLO_probability}, and it follows that the denominator above equals 1, by the normalisation condition for BLO probabilities. Although the particular example graph $\graphname$ considered here is acyclic, that the denominator equals 1 holds for BLO models on any (possibly cyclic) graph, as these models are associated with valid process operators or process matrices. The above-mentioned results of \cite{VilasiniRennerPRA} apply generically to cyclic graphs as well, and one can similarly establish the equivalence of the probabilities computed in the BLO framework and ours, in any BLO-QCM$_{\otimes}$. Finally, notice that all probability computations in our framework directly use the quantum channels and do not refer to the CJ operator, in contrast to the probability rule of the BLO formalism, \cref{eq: BLO_probability}.

\subsection{General BLO causal models and indefinite causal order processes}
\label{sec:process_matrices}

\paragraph{General BLO-QCMs and the causal decomposition problem.} Having discussed the mapping to our formalism for tensor restricted BLO causal models, let us consider the more general case. Consider the motivating example from the beginning of the section, where $A$ is a common cause of $B$ and $C$, i.e., 
\begin{equation}
    \centertikz{
        \node (A) at (0, 0) {$A$};
        \node (B) at (-1, 0.5) {$B$};
        \node (C) at (1, 0.5) {$C$};
        \draw[->,-stealth] (A) -- (B);
        \draw[->,-stealth] (A) -- (C);
    }.
\end{equation} The process operator corresponds to the CJ operator $\sigma_{BC|A}$ of a channel $\chanmap_{A\mapsto BC}$ from $A^{\textup{out}}$  to $B^{\textup{in}}$ and $C^{\textup{in}}$. The model specifies the commuting marginals $\sigma_{B|A}$ and $\sigma_{C|A}$, which are CJ operators of channels $\chanmap_{A\mapsto B}$ and $\chanmap_{A\mapsto C}$, which non-trivially overlap on $A$. Generally, $A^{\textup{out}}$ is not a classical system (i.e., states on $A$ can not be perfectly copied) and need not factorise into two factors, one for each of the marginal channels $\sigma_{B|A}$ and $\sigma_{C|A}$. It is also not immediate how the overall channel $\chanmap_{A\mapsto BC}$ can be written as a composition of the marginal channels $\chanmap_{A\mapsto B}$ and $\chanmap_{A\mapsto C}$.

In an earlier paper \cite{Allen_2017}, it is shown that whenever the overall channel factorises into a product of its (commuting) marginals $\sigma_{BC|A}=\sigma_{B|A}\sigma_{C|A}$, then  $\hilmap_{A^{\textup{out}}}=\bigoplus_i \hilmap_{A^{\textup{out},i}_L}\otimes \hilmap_{A_R^{\textup{out},i}}$ and $\rho_{BC|A}=\sum_i \sigma_{B|A_L^i}\otimes\sigma_{C|A^i_R}$ where the marginals $\sigma_{B|A_{L}^i}$ and $\sigma_{C|A_{R}^i}$ are CJ operators of channels  $\chanmap_{A_L^i\mapsto B}:\linops(\hilmap_{A^{\textup{out},i}_{L}}) \mapsto \linops(\hilmap_{B^{\textup{in}}})$
and $\chanmap_{A_R^i\mapsto C}:\linops(\hilmap_{A^{\textup{out},i}_{R}}) \mapsto \linops(\hilmap_{C^{\textup{in}}})$. This gives $\chanmap_{A\mapsto BC}$ a circuit decomposition in terms of physical operations: measure $A^{\textup{out}}$ through a von Neumann measurement $\mathcal{M}$ and depending on the outcome $i$, split $A^{\textup{out}}$ into two subspaces $A^{\textup{out},i}_L$ and $A^{\textup{out},i}_R$ sending one to $\chanmap_{A_L^i\mapsto B}$ and the other to $\chanmap_{A_R^i\mapsto C}$. This can be mapped to our framework using three vertices $A$, $B$ and $C$ for the local operations of the three parties and 3 more vertices $M$, $\sigma_{A}$ and $\sigma_{C}$ for the internal decomposition of the process channel into the measurement $\mathcal{M}$ (that can possibly involve an additional ancilla) and the two marginal channels $\chanmap_{A_L^i\mapsto B}$ and $\chanmap_{A_R^i\mapsto C}$. Further, if we add a fourth vertex $D$ with directed edge $(D,B)$ to the common cause graph above
i.e., 
\begin{equation}
    \centertikz{
        \node (A) at (0, 0) {$A$};
        \node (B) at (-1, 0.5) {$B$};
        \node (C) at (1, 0.5) {$C$};
        \node (D) at (-2, 0) {$D$};
        \draw[->,-stealth] (A) -- (B);
        \draw[->,-stealth] (A) -- (C);
        \draw[->,-stealth] (D) -- (B);
    },
\end{equation}
and apply a similar channel decomposition for the common cause edges from $A$ to $B$ and $C$, the total unitary channel from $AD$ to $BC$ given by the BLO formalism, would decompose into sub-channels such that the absence of a directed edge/path from $D$ to $C$ in the graph is reflected in the absence of a directed path of wires in the decomposition. This would allow a faithful mapping into our framework. It is shown in 
 \cite{Allen_2017} that such direct-sum decompositions of the Choi states exist also when $A$ is a common cause of any number $n$ of vertices, allowing this case also to be mapped faithfully to our framework by considering the corresponding decomposition at the level of the channels. Thus, it is in principle possible to construct faithful mappings to our framework for BLO causal models that do not belong to BLO-QCM$_\otimes$.

An interesting open question is whether such faithful decompositions (where the connectivity of the wires in the circuit decomposition matches the connectivity of the causal graph) exist for all valid causal models in the BLO formalism. This question of causal decomposition is an important one in the quantum causal modelling literature, and the question can be refined by specifying the subset of channels with respect to which one seeks a decomposition of the original channel. In particular, causal decompositions of unitary channels in terms of smaller unitary channels has been studied and has led to interesting possibility and potential impossibility results in a range of directed acyclic graphs (defined relative to the signalling structure of the original unitary, as in the BLO framework) \cite{Lorenz_2021}. For there to exist a faithful mapping of every BLO causal model to our formalism, it is sufficient if the overall unitary channel defined in the BLO formalism admits a faithful causal decomposition in terms of general (not necessarily restricted to unitary) quantum channels. This remains an open problem even in acyclic causal models, and we are not aware of any counter-examples to show that this is not possible.

We leave a further investigation of this for future work. We note that unfaithful mappings into our formalism still exist for all process matrices and therefore for BLO-QCMs (based on unitary process matrices), and discuss this in the next subsection.

\paragraph{Process matrix protocols within our framework.}
In usual quantum protocols, operations occur in a well-defined and acyclic order. Frameworks for modelling so-called indefinite causal order processes have been proposed (e.g., \cite{hardy2005probabilitytheoriesdynamiccausal,Hardy_2009,Oreshkov_2012, Chiribella_2013}), where the order in which quantum operations are applied can be subject to quantum uncertainty or may not be acyclic. This includes scenarios such as the quantum switch \cite{Chiribella_2013}, which involves a quantum superposition of Alice acting before Bob and Bob acting before Alice on a target quantum system (possibly depending on the state of a control system), as well as purely classical scenarios where there is a cyclic dependence between the operations of the parties (such as the Baumeler-Wolf process \cite{Baumeler_2016}). 
One such prominent framework is that of process matrices \cite{Oreshkov_2012}. Here one considers $N$ parties associated with labs $A_1,...,A_N$ where just as in the BLO framework each lab is associated with a pair of in and output spaces $\hilmap_{A_i^{\textup{in}}}$ and $\hilmap_{A_i^{\textup{out}}}$. The process matrix $\sigma\in \linops(\bigotimes_i\hilmap_{A_i^{\textup{in}}})\otimes \linops(\bigotimes_i\hilmap_{A_i^{\textup{out}}})$ models the environment of these labs and tells us how they can be connected. 
The process operator $\sigma^{A_1,...,A_N}$ of the BLO formalism, when expressed in the basis-dependent convention for CJ operators (without the duals) is precisely the process matrix $\sigma$. Given such a process operator or process matrix together with a set of local quantum channels $\mathcal{M}_{\outcome_i|\setting_i}$, and the same probability rule \cref{eq: BLO_probability} can be used to compute the joint probability $\textup{Pr}_{\textup{BLO}}(\outcome_1,...,\outcome_N|\setting_1,...,\setting_N)$. Recall that the form of this rule is independent of whether we use the basis-independent version of the BLO framework (with the dual spaces) or the standard basis-dependent version of the CJ isomorphism.

\begin{figure}
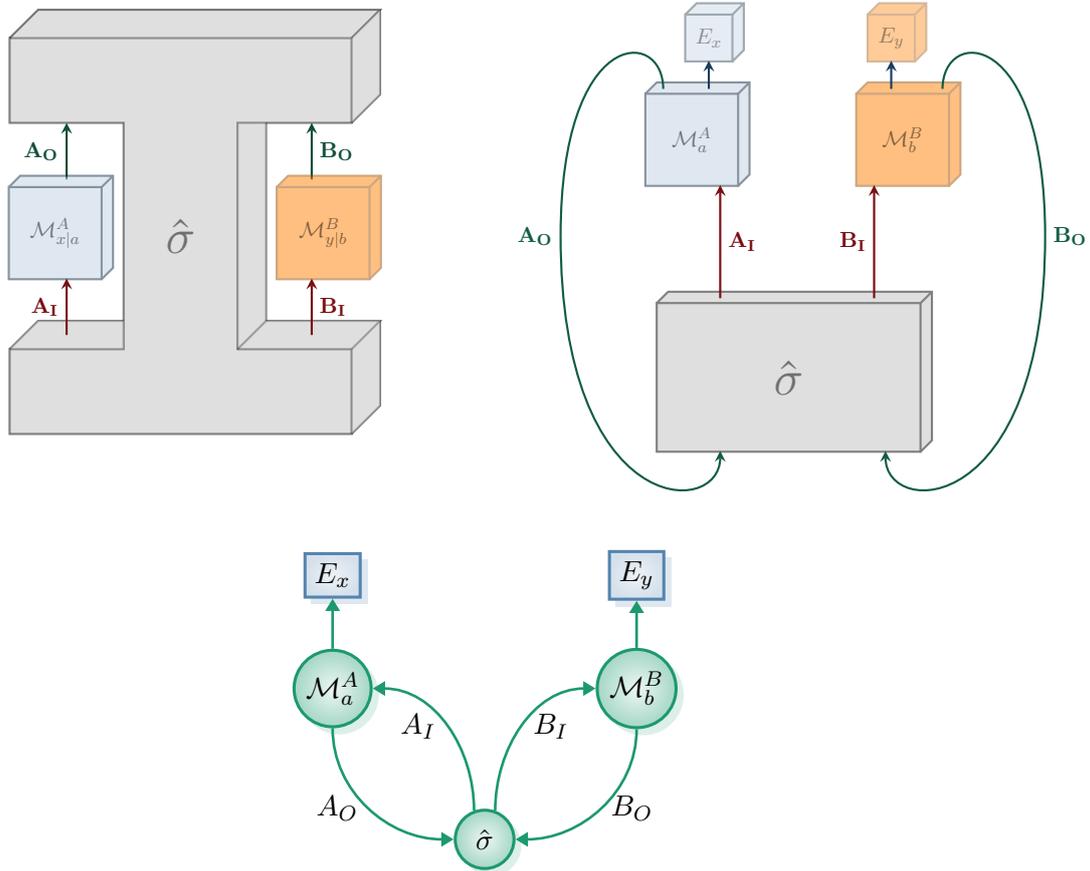

\begin{equation*}
    \centertikz{
    \begin{scope}[shift = {(-4,0)},scale=0.75, transform shape]
        \draw[thick,black, opacity=0.5,fill=lightgray] (-3,3.5)--(3,3.5)--(3,2)--(1,2)--(1,-2)--(3,-2)--(3,-3.5)--(-3,-3.5)--(-3,-2)--(-1,-2)--(-1,2)--(-3,2)--cycle; \draw[fill opacity=0.5,fill=lightgray, draw=none] (-3,3.5)--(-2.5,4)--(3.5,4)--(3,3.5)--cycle;
        \draw[thick,black, opacity=0.5] (-3,3.5)--(-2.5,4)--(3.5,4)--(3,3.5);
        \draw[thick,black, opacity=0.5] (3,2)--(3.5,2.5)--(3.5,4); \draw[fill opacity=0.5,fill=lightgray, draw=none] (3,2)--(3.5,2.5)--(3.5,4)--(3,3.5)--cycle;
        \draw[thick,black, opacity=0.5] (-3,-2)--(-2.5,-1.5)--(-1,-1.5);
        \draw[fill opacity=0.5,fill=lightgray, draw=none] (-3,-2)--(-2.5,-1.5)--(-1,-1.5)--(-1,-2)--cycle;
        \draw[thick,black, opacity=0.5] (1.5,-1.5)--(1.5,2); \draw[thick,black, opacity=0.5] (1.5,-1.5)--(1,-2);
        \draw[fill opacity=0.5,fill=lightgray, draw=none] (1.5,-1.5)--(1.5,2)--(1,2)--(1,-2)--cycle;
        \draw[thick,black, opacity=0.5] (3,-2)--(3.5,-1.5)--(1.5,-1.5);
        \draw[fill opacity=0.5,fill=lightgray, draw=none] (3,-2)--(3.5,-1.5)--(1.5,-1.5)--(1,-2)--cycle;
        \draw[thick,black, opacity=0.5] (3.5,-1.5)--(3.5,-3)--(3,-3.5);
        \draw[fill opacity=0.5,fill=lightgray, draw=none] (3.5,-1.5)--(3.5,-3)--(3,-3.5)--(3,-2)--cycle;
        \node[thick,opacity=0.5] at (0,0){\Huge{$\hat{\sigma}$}};
        \draw[thick,green!50!black,->,>=stealth] (-2,1)--node[anchor=east]{$\mathbf{A_O}$}(-2,2); 
        \draw[thick,green!50!black,->,>=stealth] (2.3,1)--node[anchor=west]{$\mathbf{B_O}$}(2.3,2); 
        \draw[thick,red!60!black,->,>=stealth] (-2,-1.75)--node[anchor=east]{$\mathbf{A_I}$}(-2,-0.75); 
        \draw[thick,red!60!black,->,>=stealth] (2.3,-1.75)--node[anchor=west]{$\mathbf{B_I}$}(2.3,-0.75); 
        \node[cube,thick,draw=blue!40!black,fill=blue!30!white, opacity=0.5, minimum width=1.6cm,minimum height=1.6cm] at (-2.2,0.05){$\mathcal{M}^A_{x|a}$};
        \node[cube,thick,draw=orange!60!black,fill=orange, opacity=0.5, minimum width=1.6cm,minimum height=1.6cm] at (2.5,0.05){$\mathcal{M}^B_{y|b}$};

    \end{scope}
   \begin{scope}[shift = {(4,0)},scale=0.75, transform shape]
        \node[cube,thick,draw=black, opacity=0.5,fill=lightgray, minimum width=4.6cm,minimum height=2.6cm] at (0,-2.5){\Huge{$\hat{\sigma}$}}; 
        \node[cube,thick,draw=blue!40!black,fill=blue!30!white, opacity=0.5, minimum width=1.6cm,minimum height=1.6cm] at (-1.7,1.7){$\mathcal{M}_\setting^A$};
        \node[cube,thick,draw=blue!40!black,fill=blue!30!white,opacity=0.4, minimum width=0.8cm,minimum height=0.8cm] at (-1.4,3.5){$\povmel{x}$};
        \node[cube,thick,draw=orange!60!black,fill=orange, opacity=0.5, minimum width=1.6cm,minimum height=1.6cm] at (2.,1.7){$\mathcal{M}_{b}^B$};
        \node[cube,thick,draw=orange!60!black,fill=orange, opacity=0.4, minimum width=0.8cm,minimum height=0.8cm] at (1.8,3.5){$\povmel{y}$};

        \draw[thick, green!60!black, bend left=50, distance=1.2cm] (2.7,2.6) .. controls (2.7,3.7) and (4.5,3.7) .. (4.5,0) node[anchor=west]{$\mathbf{B_O}$};
        \draw[ thick, green!60!black, bend left=50, distance=1.2cm, ->,>=stealth] (4.5,0) .. controls (4.5,-5) and (1.7,-5) .. (1.7,-3.8);
        \draw[thick,red!60!black,->,>=stealth] (-1.2,-1.1)--node[anchor=west]{$\mathbf{A_I}$}(-1.2,0.9); 
        \draw[thick,red!60!black,->,>=stealth] (1.5,-1.1)--node[anchor=east]{$\mathbf{B_I}$}(1.5,0.9); 
        \draw[thick, green!60!black, bend left=50, distance=1.2cm] (-2.2,2.6) .. controls (-2.2,3.7) and (-4.,3.7) .. (-4.,0) node[anchor=east]{$\mathbf{A_O}$};
        \draw[ thick, green!60!black, bend left=50, distance=1.2cm, ->,>=stealth] (-4.,0) .. controls (-4.,-5) and (-1.2,-5) .. (-1.2 ,-3.8) ; 

        \draw[thick,blue!60!black,->,>=stealth] (-1.4,2.6)--(-1.4,3.1); 
        \draw[thick,blue!60!black,->,>=stealth] (1.8,2.6)--(1.8,3.1);
        
     \end{scope}

     \begin{scope}[shift={(0,-8)}]
         \node[unode, minimum size=22pt] (W) at (0,0) {$\hat{\sigma}$};
         \node[unode] (A)at (-2,2) {$\mathcal{M}_{a}^A$};
         \node[unode] (B)at (2,2) {$\mathcal{M}_{b}^B$};
         \node[onode] (X)at (-2,3.5) {$\povmel{x}$};
         \node[onode] (Y)at (2,3.5) {$\povmel{y}$};
         \draw[qleg] (W.110) to [in=0, out=90] node[anchor=east]{$A_I$} (A);
         \draw[qleg] (A) to [in=180, out=270] node[anchor=east]{$A_O$} (W);
         \draw[qleg] (W.70) to [in=180, out=90] node[anchor=west]{$B_I$} (B);
         \draw[qleg] (B) to [in=0, out=270] node[anchor=west]{$B_O$} (W);
         \draw[qleg] (A) -- (X);
         \draw[qleg] (B) -- (Y);
     \end{scope}
    }   
\end{equation*}
    \caption{On the top left is an example of bipartite process matrix \cite{Oreshkov_2012, Chiribella_2013}. Here, the joint probabilities of two agents obtaining outcomes $x$ and $y$ conditioned on inputs $a$ and $b$ is associated with the local operations $\mathcal{M}_{x|a}^A$ and $\mathcal{M}_{y|b}^B$ of the two agents. The process $\hat{\sigma}$ describes the global behaviour of the protocol and defines the information-flow between the local operations. On the top right, the process is equivalently viewed in terms of a network formed by the loop composition of the process map $\hat{\sigma}$ with the local operations of the agents as shown on the top right, where we have modelled the classical outcomes explicitly in terms of an additional system on which the corresponding POVM acts. The top two figures are adapted from Figure 9 of~\cite{VilasiniRennerPRA}.
    The bottom figure shows the causal graph corresponding to the cyclic causal model (\cref{def:causal model}) specified by the top right cyclic network). Here, the process and each local operation correspond to unobserved nodes with quantum in and output edges while the outcomes correspond to observed classical nodes. }
    \label{fig:process matrix}
\end{figure}

In the BLO formalism, the process operator $\sigma^{A_1,...,A_N}$ is associated with a unitary channel on the parties outputs to their inputs, and this is linked to the process operator having further internal structure as a factorisation in terms of the channels $\sigma_{A_i|\parnodes{A_i}}$. While the BLO formalism includes certain classes of indefinite causal order process such as the quantum switch, it cannot model more general non-unitary processes (for instance the causal inequality violating process proposed by Oreshkov, Costa and Brukner \cite{Oreshkov_2012}).

In the earlier section, we discussed how to faithfully map tensor restricted BLO causal models to our framework. More generally, one
 can also formulate all process matrices within our formalism through a generic mapping, which is similar to mappings proposed from process matrices to other cyclic quantum causality frameworks  \cite{Lloyd_2011,Lloyd_2011_2} and \cite{VilasiniRennerPRA, VilasiniRennerPRL} (we will discuss such frameworks shortly). While this gives alternative and equivalent ways to study processes matrices (and in particular the resulting probabilities), it will generally not be a faithful mapping. The mapping is rather simple, it involves including $2N+1$ vertices $\{A_1$,...,$A_N\}\cup\{X_1$,...,$X_N\}\cup \{\sigma\}$ in the causal graph where all the $A_i$ vertices and $\sigma$ are unobserved and all $X_i$ vertices are observed. The vertex $\sigma$ associated with the channel whose CJ operator is the process matrix $\sigma$, while for the remaining vertices the channels associated with them are given by the chosen local operations in the lab through the same procedure as we used  the previous section for BLO-QCM$_\otimes$s. The generic graph for all processes would be the one where we have a directed cycle $(\sigma,A_i)$ and $(A_i,\sigma)$ for each $A_i$, along with an edge $(A_i,X_i)$ for every agent (see \cref{fig:process matrix} for an example).

 The generic graph constructed above has a high connectivity, a process where neither $A$ nor $B$ causally influence the other would also be mapped to the same graph where the lack of causal influence between the parties cannot be discerned from the graphical representation alone. This implies that the model is unfaithful: the graph shows connectivity although we will observe independences in the joint probability distribution. There is scope to employ additional procedures to fine-grain the $\sigma$ vertex of this graph into multiple vertices and thereby fine-grain the causal model (while preserving the overall behaviour of the channel $\sigma$) to discern more information on the causal structure. For instance such a notion of fine-graining for causal networks has been formally developed in \cite{VilasiniRennerPRA, VilasiniRennerPRL}. Fine-graining has the potential to transform cyclic into acyclic causal models \cite{VilasiniRennerPRA, VilasiniRennerPRL} as well as unfaithful into faithful causal models \cite{Grothus2024} while preserving the overall action of the channels involved. We leave an exploration of fine-graining in the context of causal inference in cyclic quantum causal models for future work.

\section{Functional causal models and connection with~\cite{Sister_paper}}
\label{sec:classical aspects}
In a related work~\cite{Sister_paper}, we define the probability rule for arbitrary cyclic functional causal models. 
There, we use the same ideas of mapping a functional models on a cyclic graph to a functional causal model on an acyclic graph with post-selection, making use of a classical analogue of post-selected teleportation. 
For further details on the classical framework see~\cite{Sister_paper}. 
In this section, we show that classical functional models can be embedded in our quantum causal modelling framework and prove that the probability rule given in definition 12 of~\cite{Sister_paper} is equivalent to~\cref{def: probability distribution v3} applied to the embedded casual model.
First, we briefly review functional causal models. Then, we define a mapping from functional causal models to causal models (\cref{def:causal model}). Finally, we prove consistency between the results of~\cite{Sister_paper} and the quantum formalism presented here. 

\subsection{Review of functional models}
\label{sec:fCMs}
Classical causal models~\cite{Pearl_2009}, also known as functional or structural causal models, describe how variables are related to each other through deterministic functions and are used for classical causal inference between classical random variables. 
The general definition of functional causal models allows for continuous variables. However, in what follows we restrict the discussion to finite random variables. 
Therefore, all functional causal models considered in this section belong to the subset fCM$_{\textup{finite}}$~(see \cref{fig:frameworks}). For brevity, we denote a finite functional causal model on a graph $\graphname$ as $\fcm_{\graphname}$.

\begin{definition}[Finite functional causal model]
    \label{def:functional_CM}
    Given a directed graph $\graphname=\graphexpl$, a finite functional model \textup{($\fcm_{\graphname}$)} is given by associating the following specifications to each vertex $\vertname\in\vertset$:
    \begin{myitem}
        \item A random variable $X_\vertname$ taking values $x_\vertname$ 
        from a non-empty finite set $\outcomemaparg{\vertname}$. We will use a notation where if $\vertset'\subseteq\vertset$ is a non-empty subset of vertices,
        \begin{equation}
            \outcomemaparg{\vertset'}=\prod_{\vertname\in\vertset'} \outcomemaparg{\vertname}
        \end{equation}
        where $\prod$ here denotes the Cartesian product.
        \item An error random variable $\errorrvarg\vertname$ taking values $u_\vertname$ 
        from a finite set $\errormaparg{\vertname}$, distributed as $\probex{\vertname}: \errormaparg{\vertname}\mapsto [0,1]$.
        \item A function 
            $\funcarg{\vertname}: \outcomemaparg{\parnodes{\vertname}}\times \errormaparg{\vertname} \mapsto \outcomemaparg{\vertname}$.
    \end{myitem} 
\end{definition}
In literature, the probability is considered to be well-defined for all functional models on \textit{acyclic} graphs~\cite{Pearl_2009} and for a restricted set of \textit{cyclic} functional models~\cite{Forre_2017, Bongers_2021}. 
In definition 12 of~\cite{Sister_paper}, we provided a probability rule over the vertices of any --- except for a handful of pathological cases --- cyclic functional models. In the special case of acyclic models, the probability is defined as follows in the literature (and recovered by the general rule of~\cite{Sister_paper}).

\begin{definition}[Probability distribution of an acyclic functional model] 
    \label{def: distribution_functional_cm}
    \hypertarget{probfacyc}{Consider} a functional model $\fcm_{\graphname}$ on an \textup{acyclic} graph $\graphname=\graphexpl$ and a global observed event $\outcome :=\{\outcome_{\vertname}\in\outcomemaparg{\vertname}\}_{\vertname\in\vertset}$. The probability $\probfacyc\left(\outcome\right)_{\graphname}\in [0,1]$ is defined as 

    \begin{equation}
        \probfacyc\left(\outcome\right)_{\graphname} = \sum_{u}\prod_{\vertname\in\vertset}\probex{\vertname}(u_\vertname) \delta_{\outcome_{\vertname}, \funcarg{\vertname}\left(\outcome_{\parnodes{\vertname}},u_{\vertname}\right)},
    \end{equation}
    where the sum $\sum_{u}$ runs over $u=\{u_\vertname\in\errormaparg{\vertname}\}_{\vertname\in\vertset}$.
\end{definition}
If a subset of vertices is unobserved, the distribution over the remaining observed vertices is obtained through marginalising~\cref{def: distribution_functional_cm} over the unobserved variables.

\subsection{Mapping a functional model to a causal model}
\label{sec:fCM_to_CM}
Classical systems are naturally embedded into quantum mechanical systems. Indeed, they correspond to systems whose state belongs to a subset of density operators, namely states which are diagonal in a given fixed basis, the so-called \textit{computational} basis.
Classical causal modelling, involving classical states and functional causal mechanisms is also naturally embedded into quantum causal modelling. 
In this section, we define a map from functional causal models to causal models which preserves the acyclic probability~\cref{def: distribution_functional_cm} through~\cref{def: acyclic probability}.

\paragraph{Causal graph.} Functional models are defined over graphs, not causal graphs. Hence, we first need to define causal graphs associated to functional models. Since these involve classical random variables, all edges of the causal graph are classical. Without loss of generality, we consider all vertices to be observed. If some vertices are unobserved, the probability distribution over the remaining ones is obtained though marginalising the joint distribution over the unobserved variables. Hence, it is not restrictive to consider all vertices as observed and eventually marginalise over unobserved ones.

\begin{definition}[Causal graph of a functional model]
\label{def:functional causal graph}
Given a functional model on a directed graph $\graphname= \graphexpl$, we define a causal graph by decorating $\graphname$ such that all vertices are observed, $\vertset=\overtset$, and all edges are classical, $\edgeset=\cledgeset$. Hence, all edges of the graph are represented as $\centertikz{
        \draw[cleg] (0,0) -- (0.6,0);
        }$, and all vertices are represented as $\centertikz{
            \node[onode] (q) at (0,0) {$\vertname$};
            }$.
\end{definition}
Notice that the outgoing edges of an observed vertex are necessarily classical, since all edges are classical, which proves that~\cref{def:functional causal graph} gives a well-defined causal graph.

\paragraph{Causal model.}
Mapping a functional causal model to a causal model as in~\cref{def:causal model}, involves describing classical objects, such as random variables and functional dependencies, as quantum ones.
Recall that, given a finite set $\outcomemap$, we associate to it a finite dimensional Hilbert space, such that $\dim(\hilmapsetarg{\outcomemap})=|\outcomemap|$ and label the elements of an orthonormal basis of $\hilmapsetarg{\outcomemap}$ as $\ket{\outcome}$ for $\outcome\in\outcomemap$, i.e., \[
        \hilmapsetarg{\outcomemap}= \text{\normalfont{span}}\left\{\; \ket{\outcome}\;  |  \; \outcome\in\outcomemap\right\}.
   \]
   This allows us to describe classical random variables.
In addition, functional dependencies can be described using quantum channels as follows.
\begin{definition}
    \label{def:POVM from functions}
    Let us consider a function $ \func: \outcomemapalt \mapsto \outcomemap$, where $\outcomemap$ and $\outcomemapalt$ are finite sets.
    We define a POVM $\povmarg{\func}=\{\povmelarg{\outcome}{\func} \}_{\outcome\in\outcomemap} \subset \linops\left(\hilmapsetarg{\outcomemapalt}\right)$ , as
    \begin{equation}
       \povmelarg{\outcome}{\func} = \sum_{\outcomealt\in\outcomemapalt} \delta_{\outcome, \func(\outcomealt)} \ketbra{\outcomealt}.
    \end{equation}
    \end{definition}
    
For all functions $\func: \outcomemapalt \mapsto \outcomemap$, we have
\begin{equation}
\label{eq:action povm func}
    \Tr\left[\povmelarg{\outcome}{\func} \ketbra{\outcomealt}\right] =\delta_{\outcome, \func(\outcomealt)}.
\end{equation}
In addition, $\{\povmelarg{\outcome}{\func}\}_{\outcome\in\outcomemap}$ is a POVM, since $\povmelarg{\outcome}{\func}$ are positive and it holds
    \begin{equation}
        \sum_{\outcome\in\outcomemap} E^\outcome_{\func} = \sum_{\substack{\outcome\in\outcomemap \\ \outcomealt \in\outcomemapalt}} \delta_{\outcome, \func(\outcomealt)} \ketbra{\outcomealt} =\sum_{\substack{\outcomealt\in\outcomemapalt \\ \exists \outcome\in \outcomemap : \func(\outcomealt)=\outcome}} \ketbra{\outcomealt}=\sum_{\outcomealt\in\outcomemapalt} \ketbra{\outcomealt}= \id.
    \end{equation}
With these results in mind, we proceed to describe functional models within our causal modelling framework.

\begin{definition}[Causal model associated to a functional model]
\label{def:CM of a fCM}
    Consider a functional model $\fcm_{\graphname}$ (\cref{def:functional_CM}) on a directed graph $\graphname= \graphexpl$. Using the same notation of~\cref{def:functional_CM}, we define a causal model $\cm(\fcm)_{\graphname}$ on the causal graph given by~\cref{def:functional causal graph}, through:
    \begin{myitem}
    \item To each edge $\edgename \in \edgeset = (\vertname,\vertname')$ associate the finite dimensional Hilbert space $\hilmaparg{\edgename}=\hilmapsetarg{\outcomemaparg{\vertname}}$, where $\hilmapsetarg{\outcomemaparg{\vertname}}$ is defined as in~\cref{eq:Hxspace}; 
    \item To each error variable $\errorrvarg{\vertname}$ of $\vertname\in\vertset$ associate the finite dimensional Hilbert space $\hilmapsetarg{\errormaparg{\vertname}}$, where $\hilmapsetarg{\errormaparg{\vertname}}$ is defined as in~\cref{eq:Hxspace}, and the state
    \begin{equation}
    \label{eq:error state}
        \sigma^{\vertname} = \sum_{u_{\vertname}\in\errormaparg{\vertname}}\probex{\vertname}(u_{\vertname}) \ketbra{u_\vertname}_{\hilmapsetarg{\errormaparg{\vertname}}};
    \end{equation}
     \item To each vertex $\vertname\in\vertset$ associate the finite set $\outcomemaparg{\vertname}$ as outcome set; 
    \item To each vertex $\vertname\in\vertset$ associate the set of CP maps
    \begin{equation}
        \left\{\measmaparg{\outcome}{\vertname}: \linops\left(\hilmaparg{\inedges{\vertname}}\right) \mapsto \linops\left(\hilmaparg{\outedges{\vertname}}\right)\right\}_{\outcome\in\outcomemaparg{\vertname}}
    \end{equation}
   defined as follows: for all $\rho\in\linops(\hilmaparg{\inedges{\vertname}})$,
    \begin{equation}
        \measmaparg{\outcome}{\vertname}(\rho) = \Tr\left[\povmelarg{\outcome}{\funcarg{\vertname}}(\rho\otimes\sigma^{\vertname})\right] \bigotimes_{\edgename\in\outedges{\vertname}} \ketbra{\outcome}_{\hilmaparg{\edgename}},
    \end{equation}
    where $\{\povmelarg{\outcome}{\funcarg{\vertname}}\}_{\outcome\in\outcomemaparg{\vertname}}$ is the POVM obtained through applying~\cref{def:POVM from functions} to the function $\funcarg{\vertname}: \outcomemaparg{\parnodes{\vertname}}\times \errormaparg{\vertname} \mapsto \outcomemaparg{\vertname}$ and $\sigma^{\vertname}$ is defined in~\cref{eq:error state}. 
\end{myitem}
\end{definition}

The following lemma proves that given an acyclic functional causal model $\fcm_{\graphname}$, the probability distribution evaluated using~\cref{def: distribution_functional_cm} equals the probability distribution of the image causal model~(\cref{def:CM of a fCM}), $\cm(\fcm)_{\graphname}$ evaluated using~\cref{def: acyclic probability}.

\begin{lemma}[Equivalence of ayclic probabilities under~\cref{def:CM of a fCM}]
\label{lem: equivalence acyclic probabilities fcm and cm(fcm)}
    Consider a functional model on a directed and \textit{acyclic} graph $\graphname=\graphexpl$, $\fcm_{\graphname}$, and its image under the mapping of~\cref{def:CM of a fCM}, $\cm(\fcm)_{\graphname}$.
    Let $\outcome := \{\outcome_\vertname \in \outcomemaparg{\vertname}\}_{\vertname\in\vertset}$ be a joint observed event, then we have:
    \begin{equation}
        \probfacyc(\outcome)_{\graphname}=\probacyc(\outcome)_{\graphname},
    \end{equation}
    where $\probfacyc_{\graphname}$ is the acyclic probability distribution of $\fcm_{\graphname}$, evaluated with~\cref{def: distribution_functional_cm}, and $\probacyc_{\graphname}$ the acyclic probability distribution of $\cm(\fcm)_{\graphname}$, evaluated using~\cref{def: acyclic probability}.
\end{lemma}

\begin{proof}
     For all $\vertname\in\vertset$, let us consider the action of $\measmaparg{\outcome_{\vertname}}{\vertname}$ on states of the form
     \begin{equation}
         \ketbra{\outcome_{\parnodes{\vertname}}}:=  \bigotimes_{\vertname'\in\parnodes{\vertname}}\ketbra{\outcome_{\vertname'}},
     \end{equation}
     where $\outcome_{\vertname'}\in\outcomemaparg{\vertname'}$ for all $\vertname'\in\parnodes{\vertname}$. Using~\cref{eq:action povm func} and linearity of the trace, we have
\begin{equation}
 \label{eq:action of cp maps}
        \measmaparg{\outcome_{\vertname}}{\vertname} \left(\ketbra{\outcome_{\parnodes{\vertname}}}\right) =\sum_{u_\vertname\in\errormaparg{\vertname}} \probex{\vertname}(u_\vertname)\delta_{\outcome_{\vertname}, \funcarg{\vertname}(\outcome_{\parnodes{\vertname}},u_{\vertname})} \bigotimes_{\edgename\in\outedges{\vertname}} \ketbra{\outcome_{\vertname}}_{\hilmaparg{\edgename}}.
 \end{equation} 
    The maps associated to vertices are composed according to the graph connectivity (see~\cref{def: acyclic probability} for more details) taking as input $1\in\mathbb{C}$. In addition, the output of any CP map $\measmaparg{\outcome}{\vertname}$ is always decohered in the basis specified by the values of $\outcome_\vertname$ (this is always the case for maps associated to observed vertices). Thus, the input of all maps is always of the form considered in~\cref{eq:action of cp maps}, up to a constant factor. 
    Since the graph is finite and acyclic, the output of the maps composition will eventually be in $\mathbb{C}$.
    By linearity and because of trace-preservation, such composition will equal the product of the proportionality factors mentioned above, i.e., $\sum_{u_\vertname} \probex{\vertname}(u_\vertname)\delta_{\outcome_{\vertname}, \funcarg{\vertname}(\outcome_{\parnodes{\vertname}},u_{\vertname})}$ for each $\vertname\in\vertset$.

 Thus, applying~\cref{def: acyclic probability} gives
 \begin{equation}
 \begin{split}
     \prob(\outcome)_{\graphname} = \bigcomp_{\vertname \in \vertset}
        \measmaparg{\outcome_{\vertname}}{\vertname} 
         =& \prod_{\vertname\in\vertset}\sum_{u_\vertname\in\errormaparg{\vertname}} \probex{\vertname}(u_\vertname)\delta_{\outcome_{\vertname}, \funcarg{\vertname}(\outcome_{\parnodes{\vertname}},u_{\vertname})}\\
         =&\sum_{u_\vertset}\prod_{\vertname\in\vertset}\probex{\vertname}(u_\vertname) \delta_{\outcome_{\vertname}, \funcarg{\vertname}\left(\outcome_{\parnodes{\vertname}},u_{\vertname}\right)}=\probfacyc\left(\outcome\right)_{\graphname},
 \end{split}
    \end{equation}
    where the sum $\sum_{u_\vertset}$ runs over $u_\vertset=\{u_\vertname\in\errormaparg{\vertname}\}_{\vertname\in\vertset}$ and $\outcome=\{\outcome_{\vertname}\}_{\vertname\in\vertset}$. 
 Thus proving the equivalence of this definition and~\cref{def: distribution_functional_cm}. 
\end{proof}

\paragraph{Remark.}\label{Remark:unobserved in fcm} We described functional models within our framework as causal models on a causal graph having only observed vertices. In the functional modelling literature, if a vertex is unobserved the probability distribution over the remaining ones is obtained through marginalising the joint distribution. Within our framework, unobserved variables can be naturally represented as unobserved vertices of the causal graph. Consider a functional model as in~\cref{def:functional_CM}, and assume that a subset of vertices are unobserved, $\uvertset\subseteq\vertset$. Then, we associate to this graph a causal graph where all edges are classical, $\edgeset=\cledgeset$, the unobserved vertices are $\uvertset$ and the observed are $\overtset=\vertset\setminus\uvertset$. We define a causal model on this causal graph by preserving the same associations of~\cref{def:CM of a fCM} to edges, observed vertices $\vertname\in\overtset$ and all error variables. 
However, to each unobserved vertex $\vertname\in\uvertset$, we associate the CPTP map such that, for all $\rho\in\linops(\hilmaparg{\inedges{\vertname}})$,
\begin{equation}
    \chanmaparg{\vertname}(\rho) = \sum_{\outcome\in\outcomemaparg{\vertname}} \measmaparg{\outcome}{\vertname}(\rho) = \sum_{\outcome\in\outcomemaparg{\vertname}} \Tr\left[\povmelarg{\outcome}{\funcarg{\vertname}}(\rho\otimes \sigma^\vertname)\right] \bigotimes_{\edgename\in\outedges{\vertname}} \ketbra{\outcome}_{\hilmaparg{\edgename}},
\end{equation}
where $\rho_{\outcome}^\vertname$ and $\measmaparg{\outcome}{\vertname}$ are as in~\cref{def:CM of a fCM}.

The probability distribution of this causal model, evaluated through~\cref{def: acyclic probability}, gives for all $\outcome=\{\outcome_\vertname\in\outcomemaparg{\vertname}\}_{\vertname\in\vertset}$
\begin{equation}
\begin{split}
    \probacyc\left(\{\outcome_\vertname\}_{\vertname\in\overtset}\right)_{\graphname} &=\bigcomp_{\vertname \in \overtset}
    \measmaparg{\outcome_{\vertname}}{\vertname}
    \bigcomp_{\vertname' \in \uvertset} \chanmaparg{\vertname'}=\sum_{\{\outcome_\vertname\}_{\vertname\in\uvertset}}\bigcomp_{\vertname \in \vertset}\measmaparg{\outcome_{\vertname}}{\vertname} \\
    &= \sum_{\{\outcome_\vertname\}_{\vertname\in\uvertset}} \probfacyc\left(\outcome\right)_{\graphname}= \probfacyc\left(\{\outcome_\vertname\}_{\vertname\in\overtset}\right)_{\graphname}.
\end{split}
\end{equation}
where the sum runs over $\{\outcome_{\vertname}\in\outcomemaparg{\vertname}\}_{\vertname\in\uvertset}$.
Therefore, the distribution of the causal model which accounts for unobserved vertices in the causal graph is equal to the distribution obtained through marginalising~\cref{def: distribution_functional_cm} over variables associated to unobserved vertices.

\subsection{Equivalence of $p$-separation definitions}
In a companion paper~\cite{Sister_paper}, we provide a causal modelling framework for cyclic classical causal models. This follows the same ideas of this work in mapping a given cyclic functional model to a family of acyclic models with post-selections. However, the mapping in section 3 of~\cite{Sister_paper} is constructed on a slightly different family of acyclic graphs which simplifies the mapping in the classical case without relying on the quantum notation and formalism used in this paper\footnote{The simplification is possible as classical variables can be copied and broadcast to multiple parties, i.e., the children vertices of the variable.}. 
The graph family of~\cref{def: graph_family_v3} plays a crucial role in defining $p$-separation~(\cref{def: p-separation}), while in~\cite{Sister_paper}, $p$-separation defined with respect to this slightly modified family. Here, we show that the two definitions of $p$-separation, which only depend on the graph, are equivalent. 
Let us first comment on the two different graph families.

\paragraph{One graph, two acyclic families.}
Given a (causal) graph $\graphname$, we constructed a family of acyclic causal graph as in~\cref{def: graph_family_v3}. This is denoted as $\graphfamily{\graphname}$, and defined above. For completeness, let us repeat the definition of $\graphfamily{\graphname}$.
\graphfamq*
From the same underlying graph $\graphname$, in definition 7 of~\cite{Sister_paper} we introduce a different family of acyclic graphs. We provide the definition of this family, which is denoted as $\graphfamilysn{\graphname}$, here. 

\begin{definition}[Family of acyclic classical teleportation graphs $\graphfamilysn{\graphname}$]
    \label{def:sn_graph_family}
    Given a directed graph $\graphname=\graphexpl$,
    we define an associated family $\graphfamilysn{\graphname}$ of acyclic graphs, where each element $\graphname\sn\in\graphfamilysn{\graphname}$ is obtained from the graph $\graphname$ as follows.
    \begin{myitem}
        \item Choose any subset of vertices $\splitvert{\graphname\sn}\subseteq\vertset$, such that the subgraph $\graphname'=(\vertset',\edgeset')$ of $\graphname$ with $\vertset'=\vertset$ and $\edgeset'=\edgeset\setminus\left(\bigcup_{\vertname\in\splitvert{\graphname\sn}}\outedges{\vertname}\right)$ is acyclic. We refer to $\splitvert{\graphname\sn}$ as the split vertices of $\graphname\sn$.
        \item The vertices $\vertset\sn$ of $\graphname\sn$ consist of the vertices $\vertset$ of the original graph $\graphname$ together with new vertices $\prevertname_\vertname, \postvertname_\vertname$ for each split vertex $\vertname \in \splitvert{\graphname\sn}$, i.e.,
        \begin{align}
            \vertset\sn = \vertset 
            \cup \{ \prevertname_\vertname \}_{\vertname\in\splitvert{\graphname\sn}} 
            \cup \{ \postvertname_\vertname \}_{\vertname\in\splitvert{\graphname\sn}}.
        \end{align}
        \item The edges $\edgeset\sn$ of $\graphname\sn$ consist of the edges $\edgeset'$ of the subgraph $\graphname'$ together with the following new edges:
        \begin{align}
            \edgeset\sn = \edgeset' 
            \cup \{\edgearg{\vertname}{\postvertname_\vertname}\}_{\vertname\in\splitvert{\graphname\sn}}
            &\cup \{\edgearg{\prevertname_\vertname}{\postvertname_\vertname}\}_{\vertname\in\splitvert{\graphname\sn}} \nonumber\\
            &\cup \{\edgearg{\prevertname_\vertname}{\vertname'}\}_{\vertname\in\splitvert{\graphname\sn}, \vertname' \in \childnodes{\vertname}_{\graphname}},
        \end{align}
        where $\childnodes{\vertname}_{\graphname}$ refers to the children vertices of $\vertname$ in the graph $\graphname$.     
    \end{myitem}
    For each $\vertname\in\splitvert{\graphname\sn}$, we refer to the vertices $\prevertname_\vertname$ and $\postvertname_\vertname$ as pre- and post-selection vertices respectively and depict them with distinct vertex styles  $\centertikz{\node[psnode] {$\postvertname_\vertname$};}$ and $\centertikz{\node[prenode] {$\prevertname_\vertname$};}$, as these will play a special role in our framework.
    This construction makes $\graphname\sn$ identical to $\graphname$ up to replacing all vertices $\vertname\in\splitvert{\graphname\sn}$ with the following structure
    \begin{equation}
\label{eq:split_node}
    \centertikz{
        \node[onode] (v) {$\vertname$};
        \draw[cleg] ($(v)-(0.5,0.8)$) -- (v.240);
        \draw[cleg] ($(v)-(-0.5,0.8)$) -- (v.300);
        \draw[cleg] (v.60) -- ($(v)+(0.5,0.8)$);
        \draw[cleg] (v.120) -- ($(v)+(-0.5,0.8)$);
        \node [above = 0.4*\chanvspace of v] {\indexstyle{\dots}};
        \node [below = 0.4*\chanvspace of v] {\indexstyle{\dots}};
        \draw [thick, decoration={brace},decorate]($(v)-(-0.5,1)$) -- ($(v)-(0.5,1)$) node [pos=0.5,anchor=north] {\indexstyle{\parnodes{\vertname}}};
        \draw [thick, decoration={brace},decorate]($(v)+(-0.5,1)$) -- ($(v)+(0.5,1)$) node [pos=0.5,anchor=south] {\indexstyle{\childnodes{\vertname}}}; 
    } \mapsto
    \centertikz{
        \node[onode] (vin) at (0,0) {$\vertname$};
        \draw[cleg] ($(vin)-(0.5,0.8)$) -- (vin.240);
        \draw[cleg] ($(vin)-(-0.5,0.8)$) -- (vin.300);
         \node[prenode] (vout) at (2.5,0) {$\prevertname_\vertname$};
        \draw[cleg] (vout.40) -- ($(vout)+(0.7,0.8)$);
        \draw[cleg] (vout.100) -- ($(vout)+(-0.3,0.8)$);
        \node [above = 0.4*\chanvspace of vout.60] {\indexstyle{\dots}};
        \node [below = 0.4*\chanvspace of vin] {\indexstyle{\dots}};
        \draw [thick, decoration={brace},decorate]($(vin)-(-0.5,1)$) -- ($(vin)-(0.5,1)$) node [pos=0.5,anchor=north] {\indexstyle{\parnodes{\vertname}_{\graphname}}};
        \draw [thick, decoration={brace},decorate]($(vout)+(-0.3,1)$) -- ($(vout)+(0.7,1)$) node [pos=0.5,anchor=south] {\indexstyle{\childnodes{\vertname}_{\graphname}}}; 
        \node[psnode] (ps) at (1,2) {$\postvertname_\vertname$};
        \draw[cleg] (vout.140) -- (ps.300);
        \draw[cleg] (vin.90) -- (ps.240);
    }.
\end{equation}
Each element of the family $\graphname\sn\in\graphfamilysn{\graphname}$ is called a classical \textup{teleportation graph}. The set of all pre- and post-selection vertices in $\graphname\sn$ are denoted as $\prevertset = \{\prevertname_\vertname\}_{\vertname\in\splitvert{\graphname\sn}}$ and $\psvertset = \{\postvertname_\vertname\}_{\vertname\in\splitvert{\graphname\sn}}$.
\end{definition}
\Cref{def: graph_family_v3} is formulated in terms of causal graphs while~\cref{def:sn_graph_family} is not. However, we can equivalently consider a causal graph where all vertices are observed in \cref{def:sn_graph_family} (consistently with the notation used in  \cref{eq:split_node})and compare the two definitions. 

With this in mind, the family of acyclic graphs of~\cref{def:sn_graph_family} is constructed in a similar way to the family of acyclic teleportation graphs in~\cref{def: graph_family_v3}. 
Specifically, both definitions involve considering an acyclic subgraph of $\graphname$. The subgraphs that are used to construct~\cref{def:sn_graph_family} are a subset of those considered for~\cref{def: graph_family_v3}. Indeed, in the latter, one can remove any subset of edges from $\graphname$ to make graph acyclic, while in~\cref{def:sn_graph_family} one removes all outgoing edges from a subset of vertices $\splitvert{\graphname\sn}$. 
In addition, in the classical case all outgoing edges are jointly substituted with one pair of pre- and post-selection vertices, while, in the teleportation graph family of~\cref{def: graph_family_v3}, each split edge is associated to a pair of pre- and post-selection vertices.

To understand the difference between these two families, consider the cyclic graph, $\graphname_{\textup{eg}}$, and its subgraphs, $\graphname_{\textup{eg}}'$ and $\graphname_{\textup{eg}}''$.
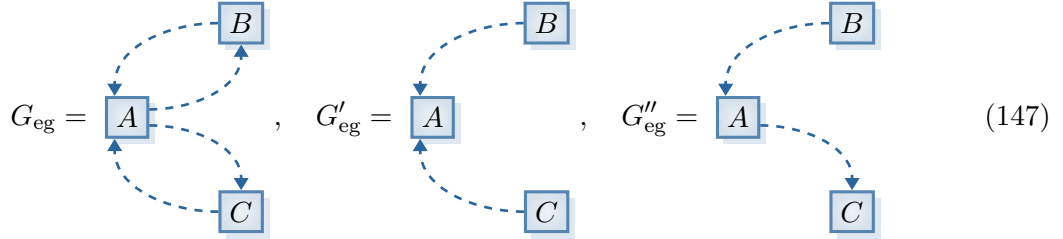
\begin{equation}
    \graphname_{\textup{eg}}=\centertikz{
        \node[onode] (A) at (0,0) {$A$};
        \node[onode] (B) at (1.5,1.25) {$B$};
        \node[onode] (C) at (1.5,-1.25) {$C$};
        \draw[cleg] (A.20)to[in=270,out=0](B.south);
        \draw[cleg] (B.180)to[in=90,out=180](A.120);
        \draw[cleg] (A.340)to[in=90,out=0](C.north);
        \draw[cleg] (C.180)to[in=270,out=180](A.240);
    }, \quad \graphname_{\textup{eg}}'= \centertikz{
        \node[onode] (A) at (0,0) {$A$};
        \node[onode] (B) at (1.5,1.25) {$B$};
        \node[onode] (C) at (1.5,-1.25) {$C$};
        \draw[cleg] (B.180)to[in=90,out=180](A.120);
        \draw[cleg] (C.180)to[in=270,out=180](A.240);
    }, \quad \graphname_{\textup{eg}}''= \centertikz{
        \node[onode] (A) at (0,0) {$A$};
        \node[onode] (B) at (1.5,1.25) {$B$};
        \node[onode] (C) at (1.5,-1.25) {$C$};
        \draw[cleg] (B.180)to[in=90,out=180](A.120);
        \draw[cleg] (A.340)to[in=90,out=0](C.north);
    }
\end{equation}
The acyclic sugraph $\graphname_{\textup{eg}}'$ is a valid acyclic subraph for both~\cref{def: graph_family_v3,def:sn_graph_family}, respectively with set of split edges $\splitedges{\graphname_{\textup{eg}}'}=\{(A,B),(A,C)\}$ and split vertices $\splitvert{\graphname_{\textup{eg}}}=\{A\}$. While, the acyclic subgraph $\graphname_{\textup{eg}}''$, which has set of split edges $\splitedges{\graphname_{\textup{eg}}''}=\{(A,B),(C,A)\}$, is not valid for~\cref{def:sn_graph_family}.

The member of the families in~\cref{def: graph_family_v3,def:sn_graph_family} constructed from $\graphname_{\textup{eg}}'$, are, respectively,
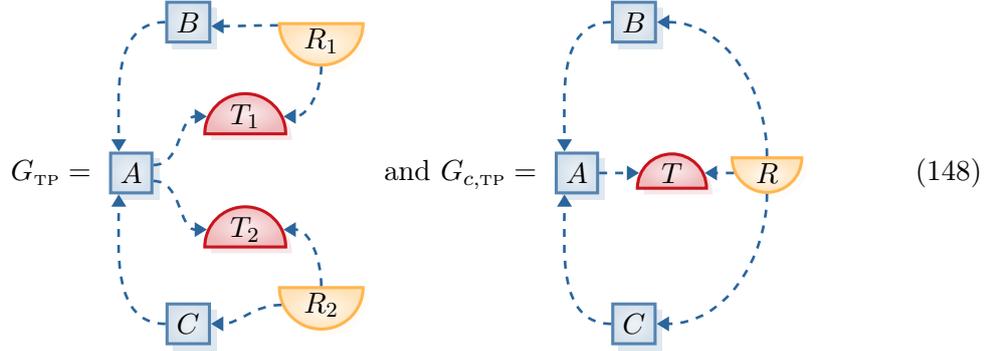
\begin{equation}
    \graphtele= \centertikz{
        \node[onode] (A) at (0,0) {$A$};
        \node[onode] (B) at (0.75,2) {$B$};
        \node[onode] (C) at (0.75,-2) {$C$};
        \node[prenode] (pre1) at (2.5,1.75) {$\prevertname_1$};
        \node[prenode] (pre2) at (2.5,-1.75) {$\prevertname_2$};
        \node[psnode] (ps1) at (1.5,0.75) {$\postvertname_1$};
        \node[psnode] (ps2) at (1.5,-0.75) {$\postvertname_2$};
        \draw[cleg] (A.20)to[in=180,out=0](ps1.180);
        \draw[cleg] (pre1.270)to[in=0,out=270](ps1.0);
        \draw[cleg] (pre1.160)to[in=0,out=180](B.350);
        \draw[cleg] (B.180)to[in=90,out=180](A.120);
        \draw[cleg] (A.340)to[in=180,out=0](ps2.180);
        \draw[cleg] (pre2.90)to[in=0,out=90](ps2.0);
        \draw[cleg] (pre2.180)to[in=0,out=180](C.0);
        
        \draw[cleg] (C.180)to[in=270,out=180](A.240);
        } \text{ and } \graphname\sn= \centertikz{
        \node[onode] (A) at (0,0) {$A$};
        \node[onode] (B) at (0.75,2) {$B$};
        \node[onode] (C) at (0.75,-2) {$C$};
        \node[prenode] (pre) at (2.5,0) {$\prevertname$};
        \node[psnode] (ps) at (1.25,0) {$\postvertname$};
        \draw[cleg] (pre.90)to[in=0,out=90](B.0);
        \draw[cleg] (pre.270)to[in=0,out=270](C.0);
        \draw[cleg] (B.180)to[in=90,out=180](A.120);
        \draw[cleg] (C.180)to[in=270,out=180](A.240);

        \draw[cleg] (A.0) to[in=180,out=0](ps.180);
        \draw[cleg] (pre.180) to[in=0,out=180](ps.0);
        }
\end{equation}
which differ in the amount of pre- and post-selection vertices. Such difference reflects the fact that each vertex of a functional model broadcasts its value to all its children, thus allowing for simplification compared to the quantum definition.

\paragraph{Equivalence of $p$-separation definitions.}
Let us now show that $p$-separation can be equivalently defined using either of these graph families. This result is necessary to establish consistency of the present work with the results of~\cite{Sister_paper}. First we recall the two definitions of $p$-separation by repeating the definition given in~\cref{sec: introducing pseparation} and in~\cite{Sister_paper}. 

\psep*

In section 4 of~\cite{Sister_paper}, $p$-separation was defined with respect to the graph family~\cref{def:sn_graph_family}.\footnote{In what follows we modify the notation of $p$-separation with respect to $\graphfamilysn{\graphname}$ compared to~\cite{Sister_paper}, which is denoted there as $(V_1\perp^{p} V_2|V_3)_{\graphname}$. Here, $(V_1\perp^{p} V_2|V_3)_{\graphname}$ denotes $p$-separation with respect to $\graphfamily{\graphname}$, while $(V_1\perp^{p,c} V_2|V_3)_{\graphname}$ is with respect to $\graphfamilysn{\graphname}$. This distinction is necessary to show equivalence of the two definitions. Once this is established the notation $(V_1\perp^{p} V_2|V_3)_{\graphname}$ will denote either definition.}

\begin{definition}[$p$-separation in \cite{Sister_paper}]
\label{def: cl_p-separation}
Let $\graphname$ be a directed graph and $V_1$, $V_2$ and $V_3$ denote any three disjoint subsets of the vertices of $\graphname$ with $V_1$ and $V_2$ being non-empty. Then, we say that $V_1$ is $p$-separated from $V_2$ given $V_3$ in $\graphname$, denoted $(V_1\perp^{p,c} V_2|V_3)_{\graphname}$, if and only if there exists $\graphname\sn\in \graphfamilysn{\graphname}$ (\cref{def:sn_graph_family}) such that $(V_1\perp^d V_2|V_3\cup\psvertset)_{\graphname\sn}$, where $\perp^d$ denotes $d$-separation and $\psvertset$ denotes the set of all post-selection vertices in the teleportation graph $\graphname\sn\in \graphfamilysn{\graphname}$. 
 Otherwise, we say that $V_1$ is $p$-connected to $V_2$ given $V_3$ in $\graphname$, and we denote it $(V_1\not\perp^{p,c} V_2|V_3)_{\graphname}$.
 To summarize,
 \begin{equation}
 \begin{aligned}
     \text{$p$-separation: } (V_1\perp^{p,c}V_2|V_3)_\graphname \equiva \exists \graphname\sn \in \graphfamilysn{\graphname} \st (V_1\perp^d V_2|V_3 \cup \psvertset)_{\graphname\sn}, \\
     \text{$p$-connection: } (V_1\not\perp^{p,c}V_2|V_3)_\graphname \equiva \forall \graphname\sn \in \graphfamilysn{\graphname} \st (V_1\not\perp^d V_2|V_3 \cup \psvertset)_{\graphname\sn}.
 \end{aligned}
 \end{equation}
\end{definition}

The following lemma shows equivalence of the two definitions of $p$-separation.

\begin{lemma}[Equivalence of $p$-separation in the two graph families]
\label{lemma: pseparation between families}
    Let us consider a directed graph $\graphname$ and three disjoint subsets of vertices $\vertset_1,\vertset_2,\vertset_3\subseteq\vertset$ with $\vertset_1,\vertset_2$ being not empty. Then, it holds
    \begin{equation}
        (V_1\perp^{p,c} V_2|V_3)_{\graphname} \iff (V_1\perp^{p} V_2|V_3)_{\graphname},
    \end{equation}
    where $(V_1\perp^{p,c} V_2|V_3)_{\graphname}$ denotes $p$-separation with respect to the graph family $\graphfamilysn{\graphname}$ (\cref{def: cl_p-separation}) and $(V_1\perp^{p} V_2|V_3)_{\graphname}$ denotes $p$-separation with respect to the graph family $\graphfamily{\graphname}$ (\cref{def: p-separation}).
\end{lemma}
\begin{proof}

\noindent {\textbf{(}$\perp^{p,c}\implies\perp^{p}$\textbf{)}}: By \cref{def: cl_p-separation}, $(V_1\perp^{p,c} V_2|V_3)_{\graphname}$ implies that $\exists$ $\graphname\sn\in \graphfamilysn{\graphname}$ (\cref{def:sn_graph_family}) such that $(V_1\perp^d V_2|V_3\cup\psvertset)_{\graphname\sn}$, where $\psvertset$ is the set of all post-selection vertices in $\graphname\sn$. We now define a teleportation graph $\graphtele\in\graphfamily{\graphname}$ (\cref{def: graph_family_v3}) by choosing 
$\splitedges{\graphtele}:=\outedges{\splitvert{\graphname\sn}}$.
Then the graphs $\graphtele$ and $\graphname\sn$ are equivalent
up to the following replacement (see equation 20 in~\cite{Sister_paper}) 
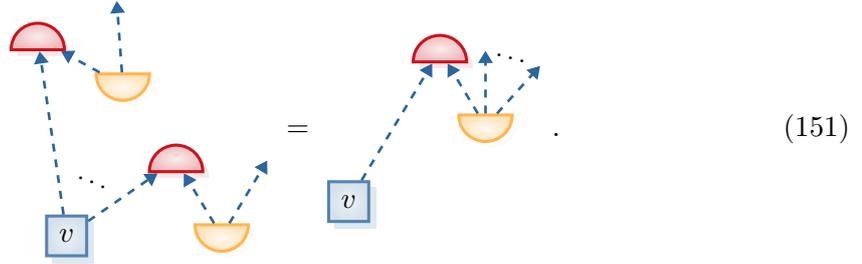
\begin{equation}
    \label{eq: q-c_telegraphs}
    \centertikz{
    \begin{scope}[xscale=1.2]
        \node[rotate=-20] at (0.55, 1.2) {\small$\dots$};
    \end{scope}
    \node[onode] (vi) at (0.25*1.2,0.5) {$\vertname$};
    \begin{scope}[xscale=1.2,shift={(0,2)}, rotate=30]
        \node[prenode] (pre) at (1,0) {};
        \node[psnode] (post) at (0.5,1) {};
        \draw[cleg] (pre) -- (post);
        \draw[cleg] (pre) -- (1.5,1);
    \end{scope}
    \begin{scope}[xscale=1.2]
         \draw[cleg] (vi) -- (post);
    \end{scope}
    \begin{scope}[xscale=1.2,shift={(0.95,0.5)}, rotate=0]
        \node[prenode] (pre) at (1,0) {};
        \node[psnode] (post) at (0.5,1) {};
        \draw[cleg] (pre) -- (post);
        \draw[cleg] (pre) -- (1.5,1);
    \end{scope}
    \begin{scope}[xscale=1.2]
         \draw[cleg] (vi) -- (post);
    \end{scope}
    } = \centertikz{
    \begin{scope}[xscale=1.2]
        \node[prenode] (pre) at (1,0) {};
        \node[psnode] (post) at (0.5,1) {};
        \draw[cleg] (pre) -- (post);
        \node[onode] (vi) at (-0.5,-1) {$\vertname$};
        \draw[cleg] (vi) -- (post);
        \draw[cleg] (pre) -- (1,1);
        \node[rotate=-15] at (1.3, 0.9) {\small$\dots$};
        \draw[cleg] (pre) -- (1.6,0.8);
    \end{scope}
    }.
\end{equation}
The left hand side represents the structure of $\graphtele$ and the right hand side, $\graphname\sn$, and where the red vertices correspond to the post-selection vertices in both cases.
Then it is easy to see that $(V_1\perp^d V_2|V_3\cup\psvertset)_{\graphname\sn}$ implies $(V_1\perp^d V_2|V_3\cup\psvertset')_{\graphtele}$, where $\psvertset'$ is the set of all post-selection vertices in $\graphtele$. By \cref{def: p-separation} of $p$-separation, $(V_1\perp^d V_2|V_3\cup\psvertset')_{\graphtele}$ implies $(V_1\perp^p V_2|V_3)_{\graphname}$ as required.\\

\noindent {\textbf{(}$\perp^{p}\implies\perp^{p,c}$\textbf{)}}: If $(V_1\perp^p V_2|V_3)_{\graphname}$, then there exists $\graphtele\in\graphfamily{\graphname}$ such that $(V_1\perp^d V_2|V_3\cup\psvertset)_{\graphtele}$ where $\psvertset$ denotes the set of all post-selection vertices in $\graphtele$ (\cref{def: p-separation}). Let $\splitedges{\graphtele}$ be the set of split edges of  $\graphtele$ and define 
\begin{equation}
    \splitvert{\graphname\sn}:=\{\vertname\in\vertset \text{ s.t. }\exists \edgename = (\vertname,\vertname')\in\splitedges{\graphtele}\}.
\end{equation}
Consider the graph $\graphname'_{\textsc{tp}}\in\graphfamily{\graphname}$ such that $\splitedges{\graphname'_{\textsc{tp}}}=\bigcup_{\vertname\in\splitvert{\graphname\sn}} \outedges{\vertname}$. We want to prove that  $(V_1\perp^d V_2|V_3\cup\psvertset')_{\graphname'_{\textsc{tp}}}$, where $\psvertset'$ denotes the set of post-selection vertices in $\graphname'_{\textsc{tp}}$. 

Let us consider $\vertname_1\in\vertset_1$ and $\vertname_2\in\vertset_2$, if these are $d$-separated conditioned on $\vertname_3,\psvertset$ in $\graphtele$, then one of the following conditions is satisfied: 
\begin{itemize}
    \item There is no path (even undirected) between $\vertname_1$ and $\vertname_2$ in $\graphtele$. $\graphname'_{\textsc{tp}}$ can be obtained through substituting some edges of $\graphtele$ with a pre- and post-selection edges and vertices. Such substitution does not create new paths, hence there is also no path between $\vertname_1$ and $\vertname_2$ in $\graphname'_{\textsc{tp}}$.
    \item All paths between $\vertname_1$ and $\vertname_2$ are blocked by $\vertset_3\cup\psvertset$. Hence, for each path in $\graphtele$, one of the following holds:
    \begin{itemize}
        \item[$\circ$] There exists $W\in\vertset_3\cup\psvertset$ such that $A\leftarrow W \rightarrow B$ for some $A$ and $B$ in the path. Adding pre and post-selection edges and vertices does not change the path, hence, the difference between $\graphtele$ and $\graphname'_{\textsc{tp}}$ is that some edges of the path might be substituted with such vertices and edges. If this substitution does not involve the subgraph $A\leftarrow W \rightarrow B$, then the same $W$ blocks the path in $\graphname'_{\textsc{tp}}$. If the edge $A\leftarrow W$ is substituted, we have $A\leftarrow Q \rightarrow P \leftarrow W \rightarrow B$. Since, the path between $\vertname_1$ and $\vertname_2$ contains $P \leftarrow W \rightarrow B$ and $W\in\vertset_3\cup\psvertset'$, then the path is blocked in $\graphname'_{\textsc{tp}}$. The same considerations hold if the edge $W\rightarrow B$ is substituted.
        \item[$\circ$] There exists $W\in\vertset_3\cup\psvertset$ such that $A\rightarrow W \rightarrow B$ for some $A$ and $B$ in the path. Adding pre and post-selection edges and vertices does not change the path, hence, the difference between $\graphtele$ and $\graphname'_{\textsc{tp}}$ is that some edges of the path might be substituted with a teleportation subgraph. If this substitution does not involve the subgraph $A\rightarrow W \rightarrow B$, then the same vertex $W$ blocks the path also in $\graphname'_{\textsc{tp}}$. If the edge $A\rightarrow W$ is substituted, we have $A\rightarrow P \leftarrow Q \rightarrow W \rightarrow B$. Since the path between $\vertname_1$ and $\vertname_2$ contains $Q \rightarrow W \rightarrow B$ and $W\in
        \vertset_3\cup\psvertset'$, then the path is blocked in $\graphname'_{\textsc{tp}}$. The same considerations hold if the edge $W\rightarrow B$ is substituted.
        \item[$\circ$] There exists $U$ such that $A\rightarrow U \leftarrow B$ for some $A$ and $B$ in the path and $U$ nor any descendent of $U$ are in $\vertset_3\cup\psvertset$. Assume first that the substitution of edges with pre and post-selection vertices and edges does not involve the subgraph $A\rightarrow U \leftarrow B$, but that there is a descendent of $U$ in the conditioning set, i.e., $W\in\vertset_3\cup\psvertset'$ such that 
        \begin{equation}
        \centertikz{
            \node (U) at (0,0) {$U$};
            \node (A) at (-1,0) {$A$};
            \node (B) at (1,0) {$B$};
            \node (ds) at (0,-0.75) {$\dots$};
            \node (W) at (0,-1.5) {$W$};
            \draw[->] (A)-- (U);
            \draw[->] (B)-- (U);
            \draw[->] (U)-- (ds);
            \draw[->] (ds)-- (W);
        }.
        \end{equation}
        Since, $U$ had no descendents in the set $W\in\vertset_3\cup\psvertset$, $W$ is a post-selection vertex that was not in $\graphtele$, i.e., $W\in\psvertset'\setminus\psvertset$. However, by construction of $\splitedges{\graphname'_{\textsc{tp}}}$, this means that $\vertname\in\inedges{W}$ is also a split vertex, $\vertname\in\splitvert{\graphname\sn}$. Therefore, there exists $\edgename=(\vertname,\vertname')\in\splitedges{\graphtele}$. In $\graphname'_{\textsc{tp}}$ we have
         \begin{equation}
        \centertikz{
            \node (U) at (0,0) {$U$};
            \node (A) at (-1,0) {$A$};
            \node (B) at (1,0) {$B$};
            \node (ds) at (0,-0.75) {$\dots$};
            \node (ds) at (0,-0.75) {$\dots$};
            \node (v) at (0,-1.5) {$\vertname$};
            \node (W) at (0,-2.25) {$W$};
            \node (Wp) at (1,-1.5) {$W'$};
            \draw[->] (A)-- (U);
            \draw[->] (B)-- (U);
            \draw[->] (U)-- (ds);
            \draw[->] (ds)-- (v);
            \draw[->] (v)-- (W);
            \draw[->] (v)-- (Wp);
        } \textup{ and in $\graphtele$ } \centertikz{
            \node (U) at (0,0) {$U$};
            \node (A) at (-1,0) {$A$};
            \node (B) at (1,0) {$B$};
            \node (ds) at (0,-0.75) {$\dots$};
            \node (ds) at (0,-0.75) {$\dots$};
            \node (v) at (0,-1.5) {$\vertname$};
            \node (W) at (0,-2.25) {$\dots$};
            \node (Wp) at (1,-1.5) {$W'$};
            \draw[->] (A)-- (U);
            \draw[->] (B)-- (U);
            \draw[->] (U)-- (ds);
            \draw[->] (ds)-- (v);
            \draw[->] (v)-- (W);
            \draw[->] (v)-- (Wp);
        }.
        \end{equation}
        Hence, $W'\in\vertset_3\cup\psvertset$ and it is a descendent of $U$. This is not possible by hypothesis, meaning that even in $\graphname'_{\textsc{tp}}$, $U$ has no descendents in the conditioning set. If the edge $A\rightarrow U$ is substituted, we have $A\rightarrow P \leftarrow Q \rightarrow U \leftarrow B$. The path between $\vertname_1$ and $\vertname_2$ contains $Q \rightarrow U \leftarrow B$, hence the path is blocked in $\graphname'_{\textsc{tp}}$. The same considerations hold if the edge $U\leftarrow B$ is substituted.
    \end{itemize}
\end{itemize}
This proves that if $(\vertset_1\perp^d\vertset_2|\vertset_3\cup\psvertset)_{\graphtele}$, then $(\vertset_1\perp^d\vertset_2|\vertset_3\cup\psvertset')_{\graphname'_{\textsc{tp}}}$. Using \cref{eq: q-c_telegraphs} to relate $\graphname'_{\textsc{tp}}$ and $\graphname\sn$, it holds that $(\vertset_1\perp^d\vertset_2|\vertset_3\cup\psvertset')_{\graphname'_{\textsc{tp}}}$ implies $(\vertset_1\perp^d\vertset_2|\vertset_3\cup\psvertset^{\textup{c}})_{\graphname\sn}$.
\end{proof}

Thus, $p$-separation relations of a given graph $\graphname$ can be equivalently deduced using~\cref{def: p-separation} or~\cref{def: cl_p-separation}.

\subsection{Equivalence of cyclic probabilities under the mapping in~\cref{def:CM of a fCM}}
\label{sec:equivalence_of_defs_fCM}

In~\cref{sec:fCM_to_CM}, we showed that applying~\cref{def:CM of a fCM} to map an acyclic functional model to an acyclic causal model preserves the probability distribution over observed events.
In this section, we consider probabilities over cyclic functional models and show that the probability rule given in definition 12 of~\cite{Sister_paper} is equivalent to applying the probability rule of~\cref{def: probability distribution v3} to the image causal models through~\cref{def:CM of a fCM}. 

\begin{lemma}
    Consider a functional model on a directed graph $\graphname=\graphexpl$, $\fcm_{\graphname}$, and its image under the mapping of~\cref{def:CM of a fCM}, $\cm(\fcm)_{\graphname}$. Let $\outcome=\{\outcome_{\vertname}\in\outcomemaparg{\vertname}\}_{\vertname\in\vertset}$ be a joint observed event, then we have:
    \begin{equation}
        \probf(\outcome)_{\graphname} = \prob(\outcome)_{\graphname}
    \end{equation}
    where $\probf_{\graphname}$ is the probability distribution of $\fcm_{\graphname}$ evaluated with definition 12 of~\cite{Sister_paper} and $\prob_{\graphname}$ is the probability distribution given by~\cref{def: probability distribution v3} of $\cm(\fcm)_{\graphname}$.
\end{lemma}

\begin{proof}
In the proof the following two results are used.
\begin{lemma}
\label{lem: decoherence in fcm}
    Consider a functional model on a directed graph $\graphname=\graphexpl$ with functional dependencies $\funcarg{\vertname}:\outcomemaparg{\parnodes{\vertname}}\times \errormaparg{\vertname}\mapsto\outcomemaparg{\vertname}$, associated to each vertex $\vertname\in\vertset$. The maps $\measmaparg{\outcome}{\funcarg{\vertname}}$ obtained though~\cref{def:CM of a fCM}, satisfy
    \begin{equation}
        \mathcal{D}_{\edgeset_{\textup{out}}}\circ \measmaparg{\outcome}{\funcarg{\vertname}} \circ  \mathcal{D}_{\edgeset_{\textup{in}}} = \measmaparg{\outcome}{\funcarg{\vertname}}  
    \end{equation}
    for any subset of outgoing and ingoing edges to $\vertname$, $\edgeset_{\textup{out}}\subseteq \outedges{\vertname}$ and $\edgeset_{\textup{in}}\subseteq \inedges{\vertname}$. The decohering channel $\mathcal{D}_{\edgeset'}: \linops(\hilmaparg{\edgeset'}) \mapsto \linops(\hilmaparg{\edgeset'})$ is defined for any subset of edges $\edgeset'\subseteq \edgeset$ as $\mathcal{D}_{\edgeset'} = \bigotimes_{\edgename\in\edgeset} \mathcal{D}_\edgename$ acting as $\mathcal{D}_\edgename(\rho) = \sum_{\outcome \in \outcomemaparg\edgename} \ketbra{\outcome} \rho \ketbra{\outcome}$ for any $\rho\in\linops(\hilmaparg{\edgename})$.
\end{lemma}
\begin{proof}
    Let us prove that 
    \begin{equation}
        \mathcal{D}_{\edgename}\circ \measmaparg{\outcome}{\funcarg{\vertname}}  = \measmaparg{\outcome}{\funcarg{\vertname}} =  \measmaparg{\outcome}{\funcarg{\vertname}} \circ \mathcal{D}_{\edgename'} 
    \end{equation}
    for any $\edgename\in\outedges{\vertname}$ and $\edgename'\in\inedges{\vertname}$. Once this is proven the general results follows from iterative application of this equality.
    The first equality follows immediately from the form of $\measmaparg{\outcome}{\vertname}$, since for all $\rho\in\linops(\hilmaparg{\inedges{\vertname}})$, we have
    \begin{equation}
        \measmaparg{\outcome}{\vertname}(\rho) = \Tr\left[\povmelarg{\outcome}{\funcarg{\vertname}}(\rho\otimes\sigma^\vertname)\right] \bigotimes_{\edgename\in\outedges{\vertname}} \ketbra{\outcome}_{\hilmaparg{\edgename}},
    \end{equation}
    which is already coherence-free.
    The second follows from the form of $\povmelarg{\outcome}{\funcarg{\vertname}}$ and of the state $\sigma^\vertname$ (see~\cref{eq:error state}), indeed
    \begin{equation}
    \begin{split}
        \Tr\left( \povmelarg{\outcome}{\funcarg{\vertname}} (\mathcal{D}_{\edgename'}(\rho)\otimes \sigma^\vertname)\right)&=\Tr\left( \povmelarg{\outcome}{\funcarg{\vertname}} \mathcal{D}_{\edgename'\cup U_\vertname}(\rho\otimes \sigma^\vertname)\right)\\
        &= \Tr\left( \sum_{\outcomealt} \delta_{\outcome, \funcarg{\vertname}(\outcomealt)} \ketbra{\outcomealt} \mathcal{D}_{\tilde{\edgename}}(\tilde{\rho})\right)\\
        &= \sum_{\outcomealt} \sum_{z} \delta_{\outcome, \funcarg{\vertname}(\outcomealt)} \Tr\left(\ketbra{\outcomealt}  \ketbra{z} \tilde{\rho} \ketbra{z}\right)\\
         &= \sum_{\outcomealt} \delta_{\outcome, \funcarg{\vertname}(\outcomealt)} \Tr\left(\ketbra{\outcomealt}  (\rho \otimes \sigma^\vertname)\right) = \Tr\left( \povmelarg{\outcome}{\funcarg{\vertname}} (\rho \otimes \sigma^\vertname)\right).
    \end{split}
    \end{equation} 
    where we denoted $\tilde{e}=e'\cup U_\vertname$ and $\tilde{\rho} = \rho\otimes\sigma^\vertname$ for brevity, and the sums run over $\outcomealt = \{\outcomealt_{\vertname'}\in\outcomemaparg{\vertname'}\}_{\vertname'\in\parnodes{\vertname}}\cup \{u_\vertname\in\errormaparg{\vertname}\}$ and $z = \{z_{\vertname'}\in\outcomemaparg{\vertname'}\}_{\vertname'\in\parnodes{\vertname}}\cup \{u'_\vertname\in\errormaparg{\vertname}\}$.
\end{proof}

\begin{lemma}
\label{lem:trace and decoherence}
    Consider a Hilbert space $\hilmaparg{\inedges{\vertname}} = \bigotimes_{\edgename\in\inedges{\vertname}} \hilmaparg{\edgename}$\footnote{The lemma holds for any Hilbert space with such partition. We call the space $\hilmaparg{\inedges{\vertname}}$ since the lemma will be applied to subsystems in a causal model.}. For any $\povmel{\outcome},\rho\in\linops\left(\hilmaparg{\inedges{\vertname}}\right)$ and $\edgeset_{\textup{in}}\subseteq \inedges{\vertname}$, it holds
    \begin{gather}
    \Tr\left[\povmel{\outcome}(\mathcal{D}_{\edgeset_{\textup{in}}}(\rho))\right] = \Tr\left[(\mathcal{D}_{\edgeset_{\textup{in}}}(\povmel{\outcome}))\rho\right].
\end{gather}
\end{lemma}
\begin{proof}
    The proof follows from the cyclicity of the trace
    \begin{gather}
        \Tr\left[\povmel{\outcome}(\mathcal{D}_{\edgeset_{\textup{in}}}(\rho)\right] = \sum_{y}\Tr\left[ \povmel{\outcome}(\ketbra{y}_{\edgeset_{\textup{in}}}\otimes \id)\rho(\ketbra{y}_{\edgeset_{\textup{in}}}\otimes \id)\right]\\
        =\sum_{y}\Tr\left[(\ketbra{y}_{\edgeset_{\textup{in}}}\otimes \id) \povmel{\outcome}(\ketbra{y}_{\edgeset_{\textup{in}}}\otimes \id)\rho\right]= \Tr\left[(\mathcal{D}_{\edgeset_{\textup{in}}}(\povmel{\outcome}))\rho\right].
    \end{gather}
\end{proof}
Given these two results, we can proceed with the proof.\\

\textbf{Overview of the proof and relevant notation.}
The goal of this proof is to connect the probabilities of a functional causal model on a possibly cyclic graph $\graphname$, namely $\fcm_{\graphname}$, and those of its image causal model through the embedding~\cref{def:CM of a fCM}, namely $\cm(\fcm_{\graphname})$ (which, for brevity, we will simply denote as $\cm_{\graphname}$ in what follows). To achieve this, we define auxiliary functional models (\cref{def:functional_CM}) and causal models (\cref{def:causal model}). Although the idea behind the proof is straightforward, given our formalism, it involves relating a number of different causal graphs and causal models. We therefore first sketch the main steps and notation of the proof, before detailing it.
\begin{itemize}
    \item[\textbf{Step 0:}] We start considering a functional causal model on $\graphname$: $\fcm_{\graphname}$ (\cref{def:functional_CM}).
    
    The cyclic probability distribution for the observed event $\outcome$ relative to this functional model is $\probf(\outcome)_{\graphname}$ (definition 12 in~\cite{Sister_paper}).
    \item[\textbf{Step 1:}] We define a classical teleportation causal model on a teleportation graph $\graphname\sn$ induced by $\fcm_{\graphname}$: $\fcm_{\graphname\sn}$ (definition 9 of~\cite{Sister_paper}). 
    
    Let $\psvertset$ be the set of post-selection vertices in $\graphname\sn$, which is acyclic by construction, and $\splitvert{\graphname\sn}$ the set of split vertices of $\graphname\sn$.
    The probability distribution of the event $\outcome$ and successful post-selections $\{\postoutcome_\vertname=1\}_{\postvertname_\vertname\in\psvertset}$ relative to this acyclic 
functional model is $\probfacyc(\outcome,\{\postoutcome_\vertname=1\}_{\postvertname_\vertname\in\psvertset})_{\graphname\sn}$ (\cref{def: distribution_functional_cm}).
    
    By definition, the probability distribution of the cyclic model $\fcm_\graphname$ equals the conditional distribution of the acyclic model $\fcm_{\graphname\sn}$, i.e.,
    \begin{equation}
        \probf(\outcome)_{\graphname}=\probfacyc(\outcome|\{\postoutcome_\vertname=1\}_{\postvertname_\vertname\in\psvertset})_{\graphname\sn}
    \end{equation}
    for all observed events $x:=\{\outcome_\vertname\in\outcomemaparg{\vertname}\}_{\vertname\in\vertset}$.
    \item[\textbf{Step 2:}] We define an intermediate functional model on a graph $\graphname\sn'$ which is constructed from the classical teleportation graph $\graphname\sn$: $\fcm_{\graphname\sn'}$. The idea is to use the correspondence between a single broadcasted classical teleportation protocol and multiple copies of regular classical teleportation, in going from $\graphname\sn$ to $\graphname\sn'$ (\cref{eq: q-c_telegraphs}).

    Let $\psvertset'$ be the set of post-selection vertices in $\graphname\sn'$.
    The acyclic probability distribution of the event $\outcome$ and successful post-selections $\{\postoutcome^\edgename_\vertname=1\}_{\postvertname^\edgename_\vertname\in\psvertset'}$ relative to this functional model is $\probfacyc(\outcome,\{\postoutcome^\edgename_\vertname=1\}_{\postvertname^\edgename_\vertname\in\psvertset})_{\graphname\sn'}$ (\cref{def: distribution_functional_cm}).
    
    We prove that the conditional probabilities in $\fcm_{\graphname\sn}$ (from \textbf{Step 1}) and $\fcm_{\graphname\sn'}$ are equal, i.e.,
    \begin{equation}
        \probfacyc(\outcome|\{\postoutcome^\edgename_\vertname=1\}_{\postvertname^\edgename_\vertname\in\psvertset'})_{\graphname\sn'}=\probfacyc(\outcome|\{\postoutcome_\vertname=1\}_{\postvertname_\vertname\in\psvertset})_{\graphname\sn}.
    \end{equation}

    \item[\textbf{Step 3:}] We map the acyclic functional model $\fcm_{\graphname\sn'}$ to a causal model (\cref{def:CM of a fCM}): $\cm_{\graphname\sn'}$.

    The acyclic probability distribution of the event $\outcome$ and successful post-selections $\{\postoutcome^\edgename_\vertname=1\}_{\postvertname^\edgename_\vertname\in\psvertset'}$ relative to this functional model is $\probacyc(\outcome|\{\postoutcome^\edgename_\vertname=1\}_{\postvertname^\edgename_\vertname\in\psvertset'})_{\graphname\sn'}$.

    As a consequence of~\cref{lem: equivalence acyclic probabilities fcm and cm(fcm)}, the acyclic probability rule is preserved through the mapping, i.e.,
    \begin{equation}
        \probfacyc(\outcome,\{\postoutcome^\edgename_\vertname=1\}_{\postvertname^\edgename_\vertname\in\psvertset'})_{\graphname\sn'}=\probacyc(\outcome,\{\postoutcome^\edgename_\vertname=1\}_{\postvertname^\edgename_\vertname\in\psvertset'})_{\graphname\sn'}.
    \end{equation}
    \end{itemize}
Before proceeding, let us sketch the results of the first three steps. The following diagram summarises the (functional) causal models that we introduced so far and their conditional probabilities. Each step establishes equivalence between the (conditional) distributions of these causal models.
\begingroup
\crefname{equation}{eq.}{eqs.}
\crefname{definition}{def.}{Def.}
\crefname{lemma}{lem.}{Lem.}

\begin{equation}
\label{eq:steps123}
    \centertikz{
    \node (fCM) at (-3,0) {$\fcm_\graphname$};
    \node (snCM) at (0,0) {$\fcm_{\graphname\sn}$};
    \node (snCMnb) at (3,0) {$\fcm_{\graphname\sn'}$};
    \node (CMsnCMnb) at (6,0) {$\cm_{\graphname\sn'}$};
    
    \node (fCMp) at (-3,-1) {\small$\probf(x)$};
    \node (snCMp) at (0,-1) {\small$\probfacyc(x|t)$};
    \node (snCMnbp) at (3,-1) {\small$\probfacyc(x|\{t^\edgename\}_\edgename)$};
    \node (CMfCMnbp) at (6,-1) {\small$\probacyc(x|\{t^e\}_e)$};
    
    \draw[thick, ->] (fCM) -- node [midway,above] {\small\textbf{Step 1}} (snCM);
    \draw[thick, ->] (snCM) -- node [midway,above] {\small\textbf{Step 2}} (snCMnb);
    \draw[thick, ->] (snCMnb) -- node [midway, above]   {\small\textbf{Step 3}} (CMsnCMnb);    
    
    \node[color=red] at (-1.75,-1) {:=} ;
    \node[color=red] at (1.25,-1) {=} ;
    \node [color=red] at (4.5,-1) {=};
}
\end{equation} 
\endgroup
    \begin{itemize}
    \item[\textbf{Step 4:}] We map the initial functional model $\fcm_{\graphname}$ (from \textbf{Step 0}) to a causal model~(\cref{def:causal model}): $\cm_{\graphname}$ (as specified by the mapping in~\cref{def:CM of a fCM}).

    The cyclic probability distribution for the observed event $\outcome$ relative to this causal model is $\prob(\outcome)_{\graphname}$ (\cref{def: probability distribution v3}).

    \item[\textbf{Step 5:}] We define a teleportation causal model on a specific teleportation graph $\graphtele$ induced by $\cm_{\graphname}$: $\cm_{\graphtele}$ (\cref{def: graph_family_v3} and \cref{def:causal model of graphfamily_v3}).
    Specifically, we choose $\graphtele$ with split edges correspond to the outgoing edges of split vertices of $\graphname\sn$ (see \textbf{Step 1}), i.e., $\splitedges{\graphtele}=\cup_{\vertname\in\splitvert{\graphname\sn}} \outedges{\vertname}$.

    Let us denote the set of post-selection vertices of $\graphtele$ as $\tilde{V}_{\text{post}}$.
    The acyclic probability distribution of the event $\outcome$ and successful post-selections $\{\postoutcome^\edgename_\vertname=\checkmark\}_{\postvertname^\edgename_\vertname\in\tilde{V}_{\text{post}}}$ relative to this causal model is $\probacyc\left(\outcome,\{\postoutcome^\edgename_\vertname=\checkmark\}_{\postvertname^\edgename_\vertname\in\tilde{V}_{\text{post}}}\right)_{\graphtele}$ (\cref{def: acyclic probability}). 
    
    By definition, the cyclic probability of $\cm_{\graphname}$ equals the acyclic conditional distribution of $\cm_{\graphtele}$, i.e.,
    \begin{equation}
        \prob(\outcome)_{\graphname}=\probacyc\left(\outcome|\{\postoutcome^\edgename_\vertname=\checkmark\}_{\postvertname^\edgename_\vertname\in\tilde{V}_{\text{post}}}\right)_{\graphtele}.
    \end{equation}

\end{itemize}
Let us sketch what the previous two steps achieve.
The diagram summarises the causal models introduced in \textbf{Step 4-5} together with the original functional causal model introduced in \textbf{Step 0}, and their (conditional) probabilities. \textbf{Step 5} established the equivalence between the two probabilities of the newly introduced causal models. The equivalence between $\probf(\outcome)_\graphname$ and $\prob(\outcome)_{\graphname}$ is what we aim to prove.
\begingroup
\crefname{equation}{eq.}{eqs.}
\crefname{definition}{def.}{Def.}
\crefname{lemma}{lem.}{Lem.}

\begin{equation}
\label{eq:steps45}
    \centertikz{
    \node (fCM) at (-6,-2) {$\fcm_\graphname$};
    \node (CMfCM) at (-3,-2) {$\cm_{\graphname}$};
    \node (CMt) at (0,-2) {$\cm_{\graphtele}$};
    
    \node (fCMp) at (-6,-3) {\small$\probf(x)$};
    \node (CMfCMp) at (-3,-3) {\small$\prob(x)$}; 
     \node (CMfCMp) at (0,-3) {\small$\probacyc(x|t)$}; 
     
    \draw[thick, ->] (fCM) -- node [midway,above] {\small\textbf{Step 4}} (CMfCM);  
    \draw[thick,->] (CMfCM) -- node [midway,above] {\small\textbf{Step 5}} (CMt);
    \node[color=red] at (-1.75,-3) {:=};
}
\end{equation}
\endgroup
\begin{itemize}
    \item[\textbf{Step 6:}] We prove that the conditional probability distribution of the causal models $\cm_{\graphname\sn'}$ (from \textbf{Step 3}) and $\cm_{\graphtele}$ (from \textbf{Step 5}) are equal, i.e.,
    \begin{equation}
        \probacyc\left(\outcome|\{\postoutcome^\edgename_\vertname=\checkmark\}_{\postvertname^\edgename_\vertname\in\tilde{\vertset}_{\textup{post}}}\right)_{\graphtele}=\probacyc(\outcome|\{\postoutcome^\edgename_\vertname=1\}_{\postvertname^\edgename_\vertname\in\psvertset'})_{\graphname\sn'}.
    \end{equation}
\end{itemize}
Combining together all the equalities established concludes the proof. The following diagram combines those of~\cref{eq:steps123,eq:steps45} through \textbf{Step 6}.
\begingroup
\crefname{equation}{eq.}{eqs.}
\crefname{definition}{def.}{Def.}
\crefname{lemma}{lem.}{Lem.}

\begin{equation}
\centertikz{
    \node (fCM) at (-3,0) {$\fcm_\graphname$};
    \node (snCM) at (0,0) {$\fcm_{\graphname\sn}$};
    \node (snCMnb) at (3,0) {$\fcm_{\graphname\sn'}$};
    \node (CMsnCMnb) at (6,0) {$\cm_{\graphname\sn'}$};
    
    \node (fCMp) at (-3,1) {\small$\probf(x)$};
    \node (snCMp) at (0,1) {\small$\probfacyc(x|t)$};
    \node (snCMnbp) at (3,1) {\small$\probfacyc(x|\{t^\edgename\}_\edgename)$};
    \node (CMfCMnbp) at (7,1) {\small$\probacyc(x|\{t^e\}_e)$};
    
    \draw[thick, ->] (fCM) -- node [midway,above] {\small\textbf{Step 1}} (snCM);
    \draw[thick, ->] (snCM) -- node [midway,above] {\small\textbf{Step 2}} (snCMnb);
    \draw[thick, ->] (snCMnb) -- node [midway, above]   {\small\textbf{Step 3}} (CMsnCMnb);    
    \node[color=red] at (-1.75,1) {:=} ;
    \node[color=red] at (1.25,1) {=} ;
    \node [color=red] at (5,1) {=};
    \begin{scope}[shift = {(6,0)}]
        \node (CMfCM) at (-3,-2) {$\cm_{\graphname}$};
        \node (CMt) at (0,-2) {$\cm_{\graphtele}$};
        \node (CMfCMp) at (-3,-3) {\small$\prob(x)$}; 
        \node (CMfCMp) at (1,-3) {\small$\probacyc(x|t)$}; 
         
        \draw[thick, ->] (fCM) -- node [midway,below] {\small\textbf{Step 4}} (CMfCM);  
        \draw[thick,->] (CMfCM) -- node [midway,above] {\small\textbf{Step 5}} (CMt);
        \node[color=red] at (-1,-3) {:=};
    \end{scope} 
    \draw[thick] (CMsnCMnb) -- node [midway,left] {\small\textbf{Step 6}} (CMt);
    \draw[color=red] (CMfCMnbp) -- node [midway,fill=white] {\small\textbf{=}} (CMfCMp);
}
\end{equation}
\endgroup

\textbf{Details of the proof.} In what follows, we prove in detail the results of \textbf{Steps 1-6}.

\begin{itemize}
    \item[\textbf{Step 1:}] We define a classical teleportation causal model on a teleportation graph $\graphname\sn\in\graphfamilysn{\graphname}$ with split vertices $\splitvert{\graphname\sn}\subseteq\vertset$ (see~\cref{def:sn_graph_family}), as prescribed in definition 9 of~\cite{Sister_paper}.  
    For completeness, let us briefly remind what this amounts to. 
    The classical teleportation graph is obtained through performing the following substitution for all $\vertname\in\splitvert{\graphname\sn}$:
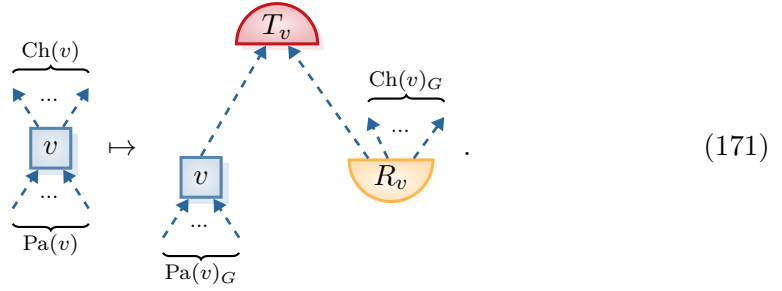
    \begin{equation}
    \label{eq: sub classical tele}
\centertikz{
        \node[onode] (v) {$\vertname$};
        \draw[cleg] ($(v)-(0.5,0.8)$) -- (v.240);
        \draw[cleg] ($(v)-(-0.5,0.8)$) -- (v.300);
        \draw[cleg] (v.60) -- ($(v)+(0.5,0.8)$);
        \draw[cleg] (v.120) -- ($(v)+(-0.5,0.8)$);
        \node [above = 0.4*\chanvspace of v] {\indexstyle{\dots}};
        \node [below = 0.4*\chanvspace of v] {\indexstyle{\dots}};
        \draw [thick, decoration={brace},decorate]($(v)-(-0.5,1)$) -- ($(v)-(0.5,1)$) node [pos=0.5,anchor=north] {\indexstyle{\parnodes{\vertname}}};
        \draw [thick, decoration={brace},decorate]($(v)+(-0.5,1)$) -- ($(v)+(0.5,1)$) node [pos=0.5,anchor=south] {\indexstyle{\childnodes{\vertname}}}; 
    } \mapsto
    \centertikz{
        \node[onode] (vin) at (0,0) {$\vertname$};
        \draw[cleg] ($(vin)-(0.5,0.8)$) -- (vin.240);
        \draw[cleg] ($(vin)-(-0.5,0.8)$) -- (vin.300);
        \node[prenode] (vout) at (2.5,0) {$\prevertname_\vertname$};
        \draw[cleg] (vout.40) -- ($(vout)+(0.7,0.8)$);
        \draw[cleg] (vout.100) -- ($(vout)+(-0.3,0.8)$);
        \node [above = 0.4*\chanvspace of vout.60] {\indexstyle{\dots}};
        \node [below = 0.4*\chanvspace of vin] {\indexstyle{\dots}};
        \draw [thick, decoration={brace},decorate]($(vin)-(-0.5,1)$) -- ($(vin)-(0.5,1)$) node [pos=0.5,anchor=north] {\indexstyle{\parnodes{\vertname}_{\graphname}}};
        \draw [thick, decoration={brace},decorate]($(vout)+(-0.3,1)$) -- ($(vout)+(0.7,1)$) node [pos=0.5,anchor=south] {\indexstyle{\childnodes{\vertname}_{\graphname}}}; 
        \node[psnode] (ps) at (1,2) {$\postvertname_\vertname$};
        \draw[cleg] (vout.140) -- (ps.300);
        \draw[cleg] (vin.90) -- (ps.240);
    }.
\end{equation}
Then, we define a classical teleportation functional model on $\graphname\sn$ through definition 9 of~\cite{Sister_paper}. In particular, this construction prescribes that the same finite sets, functions and probabilities of the original functional model are associated to vertices and edges that are preserved from $\graphname$ to $\graphname\sn$. To pre- and post-selection vertices we associate uniform prior classical teleportation protocols\footnote{Any other choice of classical post-selected teleportation protocol would yield the same probability as this canonical choice, as shown in \cite{Sister_paper} paper.}, i.e., for all $\vertname\in\splitvert{\graphname\sn}$ we associate to the corresponding pre-selection vertex $\prevertname_\vertname$ the finite set $\outcomemaparg{\prevertname_\vertname}=\outcomemaparg{\vertname}$, a uniform distribution to its error variable $\probex{\prevertname_\vertname}(u)=|\outcomemaparg{\vertname}|^{-1}$ for all $u\in\outcomemaparg{\vertname}$ and a function $\funcarg{\prevertname_\vertname}(u)=u$. Notice that this amounts to having the variable associated to pre-selection vertices $\prevertname_\vertname$ distributed as $\probex{\prevertname_\vertname}$ (for more details see the remark in section 2 of~\cite{Sister_paper}), thus in what follows we will refer to $\prevertname_\vertname$ and its distribution.
We associate to the corresponding post-selection vertex $\postvertname_\vertname$ the finite set $\outcomemaparg{\postvertname_\vertname} =\{0,1\}$ and a (deterministic) delta function $\funcarg{\postvertname_\vertname}(x,y)=\delta_{x,y}$ for all $x,y\in\outcomemaparg{\vertname}$. 

By definition 12 of~\cite{Sister_paper}, given an observed event $\outcome:=\{\outcome_{\vertname}\in\outcomemaparg{\vertname}\}_{\vertname\in\vertset}$, the probability of $\fcm_{\graphname}$ is defined as
\begin{equation}
    \probf(\outcome)_{\graphname}=\probfacyc\left(\outcome|\{\postoutcome_\postvertname=1\}_{\postvertname\in\psvertset}\right)_{\graphname\sn},
\end{equation}
where $\psvertset$ denotes the set of post-selection vertices of $\graphname\sn$.
\item[\textbf{Step 2:}] We define an intermediate functional causal model over a graph $\graphname\sn'$, $\fcm_{\graphname\sn}$. The graph $\graphname\sn'$ is constructed from $\graphname\sn$ through performing the following substitution of edges and vertices for $\vertname\in\splitvert{\graphname\sn}$ 
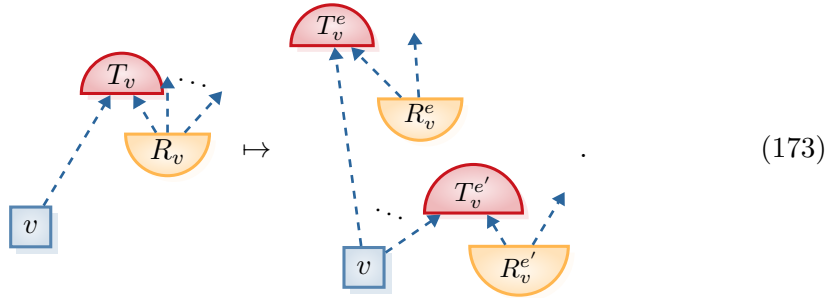
    \begin{equation}
    \label{eq: q-c_tele}
     \centertikz{
    \begin{scope}[xscale=1.2]
        \node[prenode] (pre) at (1,0) {$\prevertname_\vertname$};
        \node[psnode] (post) at (0.5,1) {$\postvertname_\vertname$};
        \draw[cleg] (pre) -- (post);
        \node[onode] (vi) at (-0.5,-1) {$\vertname$};
        \draw[cleg] (vi) -- (post);
        \draw[cleg] (pre) -- (1,1);
        \node[rotate=-15] at (1.3, 0.9) {\small$\dots$};
        \draw[cleg] (pre) -- (1.6,0.8);
    \end{scope}
    }\mapsto\centertikz{
    \begin{scope}[xscale=1.2]
        \node[rotate=-20] at (0.55, 1.2) {\small$\dots$};
    \end{scope}
    \node[onode] (vi) at (0.25*1.2,0.5) {$\vertname$};
    \begin{scope}[xscale=1.2,shift={(0,2)}, rotate=30]
        \node[prenode] (pre) at (1,0) {\small$\prevertname_\vertname^\edgename$};
        \node[psnode] (post) at (0.75,1.5) {\small$\postvertname_\vertname^\edgename$};
        \draw[cleg] (pre) -- (post);
        \draw[cleg] (pre) -- (1.5,1);
    \end{scope}
    \begin{scope}[xscale=1.2]
         \draw[cleg] (vi) -- (post);
    \end{scope}
    \begin{scope}[xscale=1.2,shift={(0.95,0.5)}, rotate=0]
        \node[prenode] (pre) at (1,0) {\small$\prevertname_\vertname^{\edgename'}$};
        \node[psnode] (post) at (0.5,1) {\small$\postvertname_\vertname^{\edgename'}$};
        \draw[cleg] (pre) -- (post);
        \draw[cleg] (pre) -- (1.5,1);
    \end{scope}
    \begin{scope}[xscale=1.2]
         \draw[cleg] (vi) -- (post);
    \end{scope}
    }.
\end{equation}
This substitution is motivated by equation 20 of~\cite{Sister_paper} which proves that $n$ copies of a classical teleportation protocol (on the right) are equivalent to one classical teleportation protocol which is broadcast to $n$ parties (on the left). For each split vertex $\vertname\in\splitvert{\graphname\sn}$ and each outgoing edge $\edgename\in\outedges{\vertname}$, let us denote with $\prevertname_\vertname^\edgename$ and $\postvertname_\vertname^\edgename$ the pre- and post-selection vertices in $\graphname\sn'$ associated to the outgoing edge $\edgename$ of $\vertname$ (see~\cref{eq: q-c_tele}). Let us denote with $\psvertset'$ the set of post-selection vertices of $\graphname\sn'$.

We define a functional causal model over $\graphname\sn'$, $\fcm_{\graphname\sn'}$, as follows:
\begin{enumerate}
    \item To each vertex present in both $\graphname\sn$ and $\graphname\sn'$ assign the same finite set, functional dependency (for endogenous vertices) or probability (for exogenous vertices), of the causal model $\fcm_{\graphname\sn}$\footnote{By construction, these are also the same association of the original causal model $\fcm_{\graphname}$ on $\graphname$.}.
    \item For all $\vertname\in\splitvert{\graphname\sn}$ and $\edgename\in\outedges{\vertname}$, associate to $\prevertname_\vertname^\edgename$ the finite set $\outcomemaparg{\prevertname_\vertname^\edgename}=\outcomemaparg{\vertname}$ and a uniform distribution to its error variable $\probex{\prevertname^\edgename_\vertname}(u)=|\outcomemaparg{\vertname}|^{-1}$ for all $u\in\outcomemaparg{\vertname}$ and a function $\funcarg{\prevertname^\edgename_\vertname}(u)=u$. Notice that this amounts to having the variable associated to pre-selection vertices $\prevertname^\edgename_\vertname$ uniformly distributed (for more details see the remark in section 2 of~\cite{Sister_paper}).
    \item For all $\vertname\in\splitvert{\graphname\sn}$ and $\edgename\in\outedges{\vertname}$, associate to $\postvertname_\vertname^\edgename$ the finite set $\outcomemaparg{\postvertname_\vertname^\edgename}=\{0,1\}$ and a (deterministic) delta function $\funcarg{\postvertname_\vertname^\edgename}(x,y) = \delta_{x,y}$ for all $x,y\in\outcomemaparg{\vertname}$.
\end{enumerate}
Hence, each outgoing edge from a vertex $\vertname\in\splitvert{\graphname\sn}$ in $\graphname$ has been replaced with a uniform prior classical teleportation protocol (definition 6 in~\cite{Sister_paper}).

Since this model is acyclic, we can use the acyclic probability rule~(\cref{def: distribution_functional_cm}) to evaluate probabilities. The equivalence between broadcasting a classical teleportation protocol to $n$ parties and $n$ copies of the same classical teleportation protocol (established in equation 20 of~\cite{Sister_paper}) implies that 
\begin{equation}
        \probfacyc(\outcome|\{\postoutcome^{\edgename}_\vertname=1\}_{\postvertname^{\edgename}_\vertname\in\psvertset'})_{\graphname\sn'}=\probfacyc(\outcome|\{\postoutcome_\vertname=1\}_{\postvertname_\vertname\in\psvertset})_{\graphname\sn},
\end{equation}
i.e., probabilities in $\fcm_{\graphname\sn}$ and $\fcm_{\graphname\sn'}$, conditioned on all post-selections being successful, are equal. Notice that the variables associated to pre-selection vertices are considered unobserved.
\item[\textbf{Step 3:}] We map the acyclic functional causal model $\fcm_{\graphname\sn'}$ to a causal model using~\cref{def:CM of a fCM}, $\cm_{\graphname\sn'}$.
As prescribed by the mapping, $\graphname\sn'$ is decorated to form a causal graph whose edges are all classical. In addition, we consider the pre-selection vertices $\prevertset'$ to be unobserved\footnote{In~\cref{def:functional causal graph}, all vertices are considered observed. In what follows, pre-selection vertices are always marginalised, hence unobserved. In the remark of~\cref{Remark:unobserved in fcm}, we showed that in this case we can consider such vertices unobserved and modify the associated maps and states.} and the remaining ones, $\vertset\cup\psvertset'$ to be observed. The associations prescribed by the mapping in~\cref{def:CM of a fCM} together with the specific form of uniform prior teleportation protocols, give the following assignments:
\begin{enumerate}
    \item To each vertex that is present in both $\graphname$ and $\graphname\sn'$, i.e., $\vertname\in\vertset$, we associate the set of CP maps acting on $\rho\in\linops(\hilmaparg{\inedges{\vertname}})$ 
    \begin{equation}
    \label{eq: map assiciated to preserved nodes}
        \measmaparg{\outcome}{\vertname}(\rho) = \Tr\left[\povmelarg{\outcome}{\funcarg{\vertname}}(\rho\otimes\sigma^\vertname)\right] \bigotimes_{\edgename\in\outedges{\vertname}} \ketbra{\outcome}_{\hilmaparg{\edgename}}
    \end{equation}
    where $\{\povmelarg{\outcome}{\funcarg{\vertname}}\}_{\outcome\in\outcomemaparg{\vertname}}$ is the POVM obtained through applying~\cref{def:POVM from functions} to the function $\funcarg{\vertname}$ associated to $\vertname$ in the original functional model $\fcm_{\graphname}$ and $\sigma^\vertname$ is the state associated to the error variable (see~\cref{eq:error state}).
    \item For all $\vertname\in\splitvert{\graphname\sn}$ and $\edgename\in\outedges{\vertname}$, associate to $\postvertname_\vertname^\edgename$ and outcome $\postoutcome^\edgename_\vertname=1$ the map acting on $\rho\in\linops\left(\hilmaparg{\inedges{\postvertname_\vertname^\edgename}}\right)\cong\linops\left(\hilmaparg{(\vertname,\postvertname_\vertname^\edgename)}\otimes\hilmaparg{(\prevertname_\vertname^\edgename,\postvertname_\vertname^\edgename)}\right)$
    \begin{equation}
        \label{eq:postsel_delta_new}
        \measmaparg{1}{\postvertname_\vertname^\edgename}(\rho) = \Tr\left[\povmelarg{\postoutcome=1}{\funcarg{\postvertname_\vertname^\edgename}}\rho\right]= \Tr\left[\sum_{\outcome\in\outcomemaparg{\vertname}}\left(\ketbra{\outcome}_{\hilmaparg{(\vertname,\postvertname^\edgename_\vertname)}}\otimes\ketbra{\outcome}_{\hilmaparg{(\prevertname^\edgename_\vertname,\postvertname^\edgename_\vertname)}}\right)\rho\right],
    \end{equation}
    where $\povmelarg{\postoutcome=1}{\funcarg{\postvertname_\vertname^\edgename}}$ is the POVM element obtained through applying~\cref{def:POVM from functions} to a delta function and outcome $1$, i.e., $\funcarg{\postvertname_\vertname^\edgename}(x,y) =\delta_{x,y}=1$.
    \item For all $\vertname\in\splitvert{\graphname\sn}$ and $\edgename=(\vertname,\vertname')\in\outedges{\vertname}$, associate to $\prevertname_\vertname^\edgename$, which is considered unobserved, the state
    \begin{equation}
    \label{eq:postsel_uniform_new}
    \sigma^{\prevertname_\vertname^\edgename} = \sum_{\outcome\in\outcomemaparg{\vertname}}\frac{1}{|\outcomemaparg{\vertname}|} \bigotimes_{\edgename\in\outedges{\prevertname_\vertname^\edgename}} \ketbra{\outcome}_{\hilmaparg{\edgename}} = \sum_{\outcome\in\outcomemaparg{\vertname}}\frac{1}{|\outcomemaparg{\vertname}|} \ketbra{\outcome}_{\hilmaparg{(\prevertname_\vertname^\edgename,\postvertname_\vertname^\edgename)}} \otimes \ketbra{\outcome}_{\hilmaparg{(\prevertname_\vertname^\edgename,v')}}.
    \end{equation}
\end{enumerate}
\Cref{lem: equivalence acyclic probabilities fcm and cm(fcm)} establishes that probabilities are preserved through mapping an acyclic functional model to an acyclic causal model through~\cref{def:CM of a fCM}.
    Since $\graphname\sn'$ is acyclic, this implies that
    \begin{equation}
        \probfacyc(\outcome,\{\postoutcome^\edgename_\vertname=1\}_{\postvertname^\edgename_\vertname\in\psvertset'})_{\graphname\sn'}=\probacyc(\outcome,\{\postoutcome^\edgename_\vertname=1\}_{\postvertname^\edgename_\vertname\in\psvertset'})_{\graphname\sn'}.
    \end{equation}
    \item[\textbf{Step 4:}] We map the original functional model $\fcm_{\graphname}$ to a causal model (\cref{def:causal model}), $\cm_{\graphname}$, using~\cref{def:CM of a fCM}. 
    As prescribed by the mapping, the graph $\graphname$ is decorated to form a causal graph such that all edges are classical and all vertices are observed. The causal model is then defined through the associations of~\cref{def:CM of a fCM}. In particular, to each vertex $\vertname\in\vertset$ we associate the set of CP maps acting on $\rho\in\linops(\hilmaparg{\inedges{\vertname}})$ 
    \begin{equation}
    \label{eq:tilde pers nodes map}
        \tilde{\mathcal{M}}_{\outcome}^{\vertname}(\rho) = \Tr\left[\povmelarg{\outcome}{\funcarg{\vertname}}(\rho\otimes\sigma^\vertname)\right] \bigotimes_{\edgename\in\outedges{\vertname}} \ketbra{\outcome}_{\hilmaparg{\edgename}}
    \end{equation}
    where $\{\povmelarg{\outcome}{\funcarg{\vertname}}\}_{\outcome\in\outcomemaparg{\vertname}}$ is the POVM obtained through applying~\cref{def:POVM from functions} to the function $\funcarg{\vertname}$ associated to $\vertname$ in the functional model $\fcm_{\graphname}$.
    \item[\textbf{Step 5:}] We define a teleportation causal model induced by $\cm_{\graphname}$. Firstly, let us consider the graph $\graphtele$ which is constructed using~\cref{def: graph_family_v3} and has split edges
    \begin{equation}
        \splitedges{\graphtele}=\bigcup_{\vertname\in\splitvert{\graphname\sn}} \outedges{\vertname},
    \end{equation}
    i.e., the split edges of $\graphtele$ correspond to the outgoing edges of split vertices in $\graphname\sn$ (see \textbf{Step 2}). As given by the definition, this amounts to adding pre and post-selection vertices for each split edge. For $\vertname\in\splitvert{\graphname\sn}$ and $\edgename\in\outedges{\vertname}$ let us denote the pre and post-selection vertices respectively as $\prevertname_\vertname^\edgename$ and $\postvertname_\vertname^\edgename$.
    Thus, the graph $\graphtele$ is obtained from $\graphname$ through performing the following substitution for all $\vertname\in\splitvert{\graphname\sn}$
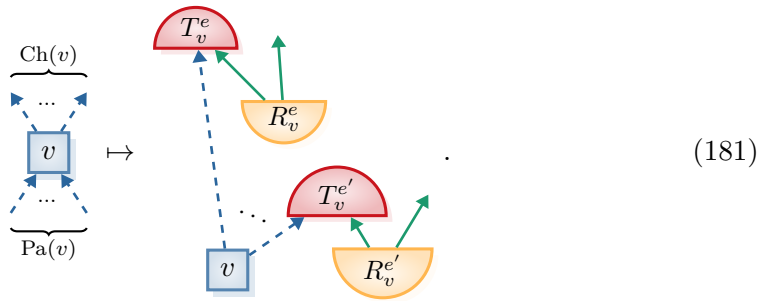
    \begin{equation}
    \label{eq:sub qtele}
\centertikz{
        \node[onode] (v) {$\vertname$};
        \draw[cleg] ($(v)-(0.5,0.8)$) -- (v.240);
        \draw[cleg] ($(v)-(-0.5,0.8)$) -- (v.300);
        \draw[cleg] (v.60) -- ($(v)+(0.5,0.8)$);
        \draw[cleg] (v.120) -- ($(v)+(-0.5,0.8)$);
        \node [above = 0.4*\chanvspace of v] {\indexstyle{\dots}};
        \node [below = 0.4*\chanvspace of v] {\indexstyle{\dots}};
        \draw [thick, decoration={brace},decorate]($(v)-(-0.5,1)$) -- ($(v)-(0.5,1)$) node [pos=0.5,anchor=north] {\indexstyle{\parnodes{\vertname}}};
        \draw [thick, decoration={brace},decorate]($(v)+(-0.5,1)$) -- ($(v)+(0.5,1)$) node [pos=0.5,anchor=south] {\indexstyle{\childnodes{\vertname}}}; 
    } \mapsto\centertikz{
    \begin{scope}[xscale=1.2]

        \node[rotate=-20] at (0.55, 1.2) {\small$\dots$};
       
    \end{scope}
    \node[onode] (vi) at (0.25*1.2,0.5) {$\vertname$};
    \begin{scope}[xscale=1.2,shift={(0,2)}, rotate=30]
        \node[prenode] (pre) at (1,0) {\small$\prevertname_\vertname^\edgename$};
        \node[psnode] (post) at (0.75,1.5) {\small$\postvertname_\vertname^\edgename$};
        \draw[qleg] (pre) -- (post);
        \draw[qleg] (pre) -- (1.5,1);
    \end{scope}
    \begin{scope}[xscale=1.2]
         \draw[cleg] (vi) -- (post);
    \end{scope}
    \begin{scope}[xscale=1.2,shift={(0.95,0.5)}, rotate=0]
        \node[prenode] (pre) at (1,0) {\small$\prevertname_\vertname^{\edgename'}$};
        \node[psnode] (post) at (0.5,1) {\small$\postvertname_\vertname^{\edgename'}$};
        \draw[qleg] (pre) -- (post);
        \draw[qleg] (pre) -- (1.5,1);
    \end{scope}
    \begin{scope}[xscale=1.2]
         \draw[cleg] (vi) -- (post);
    \end{scope}
    }.
\end{equation}
Let us denote with $\tilde{\vertset}_{\textup{post}}$ the set of post-selection vertices of $\graphtele$.

Then, we define a teleportation causal model on $\graphtele$ through~\cref{def:causal model of graphfamily_v3}. In particular, this construction prescribes that the same maps and Hilbert spaces are associated to vertices and edges that are preserved from $\graphname$ to $\graphtele$. To pre- and post-selection vertices we associate (without loss of generality, \cref{corollary:probs indep of tele implementation v3}) Bell teleportation protocols (\cref{def:bell tele}), i.e., for all $\vertname\in\splitvert{\graphname\sn}$ and $\edgename=(\vertname,\vertname')\in\outedges{\vertname}$ we associate to the corresponding pre-selection vertex $\prevertname_\vertname^\edgename$ the Bell state
\begin{equation}
    \tilde{\sigma}^{\prevertname_\vertname^\edgename} =\ketbra{\bellstate}_{\outedges{\prevertname_\vertname^\edgename}}=\frac{1}{|\outcomemaparg{\vertname}|} \sum_{x,y\in\outcomemaparg{\vertname}} \ketbraa{x}{y}_{(\prevertname_\vertname^\edgename,\postvertname_\vertname^\edgename)}\otimes \ketbraa{x}{y}_{(\prevertname_\vertname^\edgename,\vertname')}
\end{equation}
and to the corresponding post-selection vertex $\postvertname_\vertname^\edgename$ and outcome $\postoutcome=\checkmark$ the CP map acting on all $\rho\in\linops\left(\hilmaparg{\inedges{\postvertname_\vertname^\edgename}}\right)\cong\linops\left(\hilmaparg{(\vertname,\postvertname_\vertname^\edgename)}\otimes\hilmaparg{(\prevertname_\vertname^\edgename,\postvertname_\vertname^\edgename)}\right)$ as
\begin{equation}
\begin{split}
    \tilde{\mathcal{M}}^{\postvertname_\vertname^\edgename}_{\checkmark}(\rho)&=\Tr\left(\ketbra{\bellstate}_{\inedges{\postvertname_\vertname^{\edgename}}}\rho\right) \\&=\frac{1}{|\outcomemaparg{\vertname}|} \Tr\left[\sum_{\outcome\in\outcomemaparg{\vertname}}\left(\ketbraa{\outcome}{y}_{\hilmaparg{(\vertname,\postvertname^\edgename_\vertname)}}\otimes\ketbraa{\outcome}{y}_{\hilmaparg{(\prevertname^\edgename_\vertname,\postvertname^\edgename_\vertname)}}\right)\rho\right].
\end{split}  
\end{equation}

By~\cref{def: probability distribution v3}, given an observed event $\outcome:=\{\outcome_{\vertname}\in\outcomemaparg{\vertname}\}_{\vertname\in\vertset}$, the probability of $\cm_{\graphname}$ (see \textbf{Step 4}) is defined as 
\begin{equation}
    \prob(\outcome)_{\graphname}=\probacyc\left(\outcome|\{\postoutcome^\edgename_\vertname=\checkmark\}_{\postvertname^\edgename_\vertname\in\tilde{\vertset}_{\textup{post}}}\right)_{\graphtele},
\end{equation}
where $\tilde{\vertset}_{\textup{post}}$ denotes the set of post-selection vertices of $\graphtele$.
\item[\textbf{Step 6:}] Finally, we show that the causal models $\cm_{\graphname\sn'}$ (introduced in~\textbf{Step 3}) and $\cm_{\graphtele}$ (introduced in~\textbf{Step 5}) are equivalent, in the sense that for all observed events $\outcome:=\{\outcome_{\vertname}\in\outcomemaparg{\vertname}\}_{\vertname\in\vertset}$ it holds
\begin{equation}
        \probacyc\left(\outcome|\{\postoutcome^\edgename_\vertname=\checkmark\}_{\postvertname^\edgename_\vertname\in\tilde{\vertset}_{\textup{post}}}\right)_{\graphtele}=\probacyc(\outcome|\{\postoutcome^\edgename_\vertname=1\}_{\postvertname^\edgename_\vertname\in\psvertset'})_{\graphname\sn'}.
    \end{equation}

Firstly, we notice that the construction of $\graphname\sn'$ and $\graphtele$ give the same causal graphs, indeed, the subsequent substitutions in \cref{eq: sub classical tele} and \cref{eq: q-c_tele} are equivalent to \cref{eq:sub qtele}. Let us show that the maps associated to each vertex are equivalent in the two causal models, $\cm_{\graphname\sn'}$ and $\cm_{\graphtele}$. We have:
\begin{enumerate}
    \item For all vertices which are preserved from $\graphname$ to $\graphname\sn'$ (and $\graphtele$), $\vertname\in\vertset$ the maps associated in $\cm_{\graphtele}$ and $\cm_{\graphname\sn'}$ are equal (see~\cref{eq: map assiciated to preserved nodes,eq:tilde pers nodes map}), i.e., for all $\rho\in\linops(\hilmaparg{\inedges{\vertname}})$ it holds
    \begin{equation}
        \tilde{\mathcal{M}}_{\outcome}^{\vertname}(\rho) = \measmaparg{\outcome}{\vertname}(\rho)
    \end{equation}
    where we recall that $\tilde{\mathcal{M}}_{\outcome}^{\vertname}$ is the map associated to $\vertname$ in $\cm_{\graphtele}$ and $\measmaparg{\outcome}{\vertname}$ is the map associated to $\vertname$ in $\cm_{\graphname\sn'}$. 
    \item For all post-selection vertices, i.e., for all $\vertname\in\splitvert{\graphname\sn}$ and $\edgename\in\outedges{\vertname}$, we have (by the decoherence condition for observed vertices in~\cref{def:causal model})
    \begin{equation}
        \label{eq:composition and decoherence}\tilde{\mathcal{M}}^{\postvertname_\vertname^\edgename}_{\checkmark} \circ \tilde{\mathcal{M}}_{\outcome}^{\vertname} = \tilde{\mathcal{M}}^{\postvertname_\vertname^\edgename}_{\checkmark} \circ \mathcal{D}_{(\vertname,\postvertname_\vertname^\edgename)}\circ\tilde{\mathcal{M}}_{\outcome}^{\vertname}
    \end{equation}
    where all maps $\tilde{\mathcal{M}}$ are in $\cm_{\graphtele}$ and $\mathcal{D}_{(\vertname,\postvertname_\vertname^\edgename)}$ is a decohering channel acting as $\mathcal{D}_\edgename(\rho) = \sum_{\outcome \in \outcomemaparg\edgename} \ketbra{\outcome} \rho \ketbra{\outcome}$ for any $\rho\in\linops(\hilmaparg{\edgename})$.
    In order to evaluate probabilities, the maps of $\cm_{\graphtele}$ are composed as described in~\cref{def: acyclic probability}. In particular, the map $\tilde{\mathcal{M}}^{\postvertname_\vertname^\edgename}_{\checkmark}$ is always composed with $\tilde{\mathcal{M}}_{\outcome}^{\vertname}$ as in~\cref{eq:composition and decoherence}, hence, it can be replaced with $ \tilde{\mathcal{M}}^{\postvertname_\vertname^\edgename}_{\checkmark} \circ \mathcal{D}_{(\vertname,\postvertname_\vertname^\edgename)}$ without affecting probabilities. Thus, we have for all $\rho\in\linops(\hilmaparg{\inedges{\postvertname_\vertname^\edgename}})$ 
    \begingroup
      \crefname{lemma}{lem.}{Lem.}
        \begin{equation}
        \begin{split}
            \tilde{\mathcal{M}}^{\postvertname_\vertname^\edgename}_{\checkmark} \circ \mathcal{D}_{(\vertname,\postvertname_\vertname^\edgename)}(\rho)= &\Tr\left(\ketbra{\bellstate}_{\inedges{\postvertname_\vertname^{\edgename}}}\mathcal{D}_{(\vertname,\postvertname_\vertname^\edgename)}(\rho)\right)\\\stackrel{\textup{\cref{lem:trace and decoherence}}}{=}&\Tr\left[\mathcal{D}_{(\vertname,\postvertname_\vertname^\edgename)}\left(\ketbra{\bellstate}_{\inedges{\postvertname_\vertname^{\edgename}}}\right)\rho\right]
        \end{split}
        \end{equation}
        and 
        \begin{equation}
            \begin{split}
            \label{eq:decohered_bellstate_new}
                \mathcal{D}_{(\vertname,\postvertname_\vertname^\edgename)}\left(\ketbra{\bellstate}_{\inedges{\postvertname^{\edgename}_\vertname}}\right) &= \frac{1}{|\outcomemaparg{\vertname}|}\sum_{\outcome \in \outcomemaparg\vertname} \ketbra{x}_{{(\vertname,\postvertname_\vertname^\edgename)}}\otimes \ketbra{x}_{(\prevertname_\vertname^\edgename,\postvertname_\vertname^\edgename)}.
            \end{split}
        \end{equation}
    \endgroup
    Thus we can compare this map, associated to the post-selection vertex $\postvertname^\edgename_\vertname$ in $\cm_{\graphtele}$, to the map $\measmaparg{1}{\postvertname_\vertname^\edgename}$ associated to the same post-selection vertex in $\cm_{\graphname\sn'}$~(see \cref{eq:postsel_delta_new}). It holds
    \begin{equation}
        \label{eq:equiv maps}\tilde{\mathcal{M}}^{\postvertname_\vertname^\edgename}_{\checkmark} \circ \mathcal{D}_{(\vertname,\postvertname_\vertname^\edgename)} = \frac{1}{|\outcomemaparg{\vertname}|} \measmaparg{1}{\postvertname_\vertname^\edgename}.
    \end{equation}
    
    \item For all pre-selection vertices, i.e., for all $\vertname\in\splitvert{\graphname\sn}$ and $\edgename = (\vertname,\vertname')\in\outedges{\vertname}$ we have
    \begingroup
      \crefname{lemma}{lem.}{Lem.}
      \crefname{equation}{eq.}{Eq.}
    \begin{equation}
    \label{eq:composition and decoherence2}
        \tilde{\mathcal{M}}_{\outcome'}^{\vertname'}\circ \tilde{\sigma}^{\prevertname_\vertname^\edgename} \stackrel{\textup{\cref{lem: decoherence in fcm}}}{=} \tilde{\mathcal{M}}_{\outcome'}^{\vertname'}\circ\mathcal{D}_{(\prevertname_\vertname^\edgename,\vertname')}\circ \tilde{\sigma}^{\prevertname_\vertname^\edgename}  
    \end{equation}
    where $\tilde{}$ denotes associations in $\cm_{\graphtele}$ and $\mathcal{D}_{(\prevertname_\vertname^\edgename,\vertname')}$ is a decohering channel. In order to evaluate probabilities, the maps of $\cm_{\graphtele}$ are composed as described in~\cref{def: acyclic probability}. In particular, the state $\tilde{\sigma}^{\prevertname_\vertname^\edgename}$ is always composed with the map $\tilde{\mathcal{M}}_{\outcome'}^{\vertname'}$ as in~\cref{eq:composition and decoherence2}, hence, the state $\tilde{\sigma}^{\prevertname_\vertname^\edgename}$ can be replaced with $ \mathcal{D}_{(\prevertname_\vertname^\edgename,\vertname')}\circ \tilde{\sigma}^{\prevertname_\vertname^\edgename}$ without affecting probabilities. Thus, we can compare the decohered state $ \mathcal{D}_{(\prevertname_\vertname^\edgename,\vertname')}\circ \tilde{\sigma}^{\prevertname_\vertname^\edgename}$, associated to the pre-selection vertex $\prevertname_\vertname^\edgename$ in $\cm_{\graphtele}$, to the state $\sigma^{\prevertname_\vertname^\edgename}$ associated to the same pre-selection vertex in $\cm_{\graphname\sn'}$ (see \cref{eq:postsel_uniform_new}). It holds:
    \begin{equation}
        \mathcal{D}_{(\prevertname_\vertname^\edgename,\vertname')}\circ \tilde{\sigma}^{\prevertname_\vertname^\edgename} =  \mathcal{D}_{(\prevertname_\vertname^\edgename,\vertname')}\left(\ketbra{\bellstate}_{\outedges{\prevertname_\vertname^\edgename}}\right) \stackrel{\textup{\cref{eq:decohered_bellstate_new}}}{=} \sigma^{\prevertname_\vertname^\edgename}.
    \end{equation}
    \endgroup
    Using the probability rule in~\cref{def: acyclic probability}, we can easily see that the conditional probabilities in $\cm_{\graphtele}$ and $\cm_{\graphname\sn'}$ satisfy
    \begin{equation}
        \probacyc\left(\outcome|\{\postoutcome^\edgename_\vertname=\checkmark\}_{\postvertname^\edgename_\vertname\in\tilde{\vertset}_{\text{post}}}\right)_{\graphtele}=\probacyc(\outcome|\{\postoutcome^\edgename_\vertname=1\}_{\postvertname^\edgename_\vertname\in\psvertset'})_{\graphname\sn'},
    \end{equation}
    notice that the scaling factor between the two maps in~\cref{eq:equiv maps} drops out once we condition on successful post-selections.
\end{enumerate}
\end{itemize}
Finally, we can combine the results of \textbf{Steps 1-6} to complete the proof. Indeed, we have
\begin{equation}
    \begin{split}
        \probf(\outcome)_{\graphname}\stackrel{\textbf{S1}}{=}&\probfacyc\left(\outcome|\{\postoutcome_\vertname=1\}_{\postvertname_\vertname\in\psvertset}\right)_{\graphname\sn}\\
        \stackrel{\textbf{S2}}{=}&\probfacyc(\outcome|\{\postoutcome^{\edgename}_\vertname=1\}_{\postvertname^{\edgename}_\vertname\in\psvertset'})_{\graphname\sn'}\\
        \stackrel{\textbf{S3}}{=}&\probacyc(\outcome,\{\postoutcome^\edgename_\vertname=1\}_{\postvertname^\edgename_\vertname\in\psvertset'})_{\graphname\sn'}\\        
        \stackrel{\textbf{S6}}{=}&\probacyc\left(\outcome|\{\postoutcome^\edgename_\vertname=\checkmark\}_{\postvertname^\edgename_\vertname\in\tilde{\vertset}_{\text{post}}}\right)_{\graphtele}\\
        \stackrel{\textbf{S5}}{=}&\prob\left(\outcome\right)_{\graphname}\\
    \end{split}
\end{equation}
where $\probf_{\graphname}$ is the probability in $\fcm_{\graphname}$ (\textbf{Step 0}), $\probfacyc_{\graphname\sn}$ in $\fcm_{\graphname\sn}$ (\textbf{Step 1}), $\probfacyc_{\graphname\sn'}$ in $\fcm_{\graphname\sn'}$ (\textbf{Step 2}), $\probacyc_{\graphname\sn'}$ in $\cm_{\graphname\sn'}$ (\textbf{Step 3}), $\prob_{\graphname}$ in $\cm_{\graphname}$ (\textbf{Step 4}) and $\probacyc_{\graphtele}$ in $\cm_{\graphtele}$ (\textbf{Step 5}).

\end{proof}
\section{Proofs of~\cref{sec: cyclic_to_acyc_v3}}
\label{app:proofs_map}
\acyclicitytele*
\begin{proof}
    In the first step of \cref{def: graph_family_v3}, the graph $\graphname'$ is acyclic.
    The second step introduces pre- and post-selection vertices, $\prevertname_i$ and $\postvertname_i$, and associated edges according to \cref{eq: acyc_graph_split_edge_v3}. 
    To show that this step preserves acyclicity, notice that a directed graph is acyclic if and only if it can be drawn on a page with all directed edges oriented from bottom to top of the page\footnote{This representation is equivalent to a Hasse diagram for partially ordered sets, recognising the fact that every directed acyclic graph induces a partial order and vice versa.}. 
    Since $\graphname'$ is acyclic, we can represent it on a page through such a diagram. 
    We now add the pre- and post-selection vertices to this diagram by drawing all the $\prevertname_i$ at the bottom of the page, below all other vertices, and all the $\postvertname_i$ at the top of the page, above all the other vertices. 
    Since all pre-selection vertices $\prevertname_i$ only have outgoing edges and all post-selection vertices $\postvertname_i$ only have incoming arrows, both of which will be oriented from bottom to top in our diagram, it follows that $\graphtele$ is acyclic.
\end{proof}

\equivprobabilitiestele*
\begin{proof}
    \underline{Proof in a special case.}
    Consider the special case where $\graphtele_{,1}$ and $\graphtele_{,2}$ were constructed by choosing
    \begin{align}
        \splitedges{\graphtele_{,2}} = \splitedges{\graphtele_{,1}} \cup \{\edgename_0\}, 
    \end{align}
    i.e., $\graphtele_{,2}$ splits the same edges as $\graphtele_{,1}$ plus the edge $\edgename_0$. Let $\psvertset^1$ denote the set of post-selection vertices of $\graphtele_{,1}$ and $\postoutcome:=\{\postoutcome_{i}\in \{\ok,\fail\}\}_{\postvertname_i\in\psvertset^1}$. Then, taking $\postoutcome_0$ to be the outcome of the post-selection vertex associated to the edge $\edgename_0$, we have $\postoutcome \cup \postoutcome_0:=\{\postoutcome_{i}\in \{\ok,\fail\}\}_{\postvertname_i\in\psvertset^2}$.
    
    Consider the probability $\prob(\outcome,\postoutcome)_{\graphtele_{,1}}$ associated to the causal model $\cm_{\graphtele_{,1}}$ on the causal graph underlying the teleportation graph $\graphtele_{,1}$. As this is an acyclic causal model by construction (c.f., \cref{lemma: teleportation_graph_acyclic}), this probability is immediately given by applying the acyclic probability rule of \cref{def: acyclic probability} to this causal model. Here, $\outcome = \{\outcome_\vertname\in\outcomemaparg\vertname\}_{\vertname\in\overtset}$ denotes an observed event on $\graphtele_{,1}$ (which is also an event on $\graphtele_{,2}$, since these graphs have the same observed vertices). We can write
    \begin{align}
        \label{eq:proof graphfamily 1}
      \prob(\outcome,\postoutcome)_{\graphtele_{,1}} = \Tr[F^{(\outcome,\postoutcome)} \rho^{(\outcome,\postoutcome)}],
    \end{align}
    where we defined the subnormalized density matrix and POVM element $\rho^{(\outcome,\postoutcome)}, F^{(\outcome,\postoutcome)} \in \linops(\hilmaparg{\edgename_0})$ obtained from composing all the channels of the causal model $\cm_{\graphtele_{,1}}$, but leaving the composition along the $\edgename_0$ edge as the last step.\footnote{This decomposition is not unique for teleportation graphs with disconnected components, in which case we can pick any such decomposition.}
    We then consider the post-selected teleportation protocol that the causal model $\cm_{\graphtele_{,2}}$ implements in place of the edge $\edgename_0 =: \edgearg{\vertname_0}{\vertname_0'}$, with associated pre- and post-selection vertices $\prevertname_0$ and $\postvertname_0$.
    The associated teleportation probability is denoted $\teleprob\in(0,1]$.
    For convenience, we label as $A$, $B$ and $C$ the Hilbert spaces involve in this teleportation protocol:
    \begin{align}
        \hilmaparg A = \hilmaparg{\edgearg{\vertname_0}{\postvertname_0}}, \qquad
        \hilmaparg B = \hilmaparg{\edgearg{\prevertname_0}{\postvertname_0}}, \qquad
        \hilmaparg C = \hilmaparg{\edgearg{\prevertname_0}{\vertname_0'}}.
    \end{align}
    This allows us to denote the post-selected teleportation implementation as $(\telepovm_{AB},\telestate_{BC})$.
    Visually, the situation is as follows:
    \begin{align*}
        \graphtele_{,1}\st \centertikz{
            \node (a) at (0,0) {$\rho^{(\outcome,\postoutcome)}$};
            \node (b) at (0,2) {$F^{(\outcome,\postoutcome)}$};
            \draw[qleg] (a) -- node[right] {\small$\edgename_0$} (b);
        } \qquad
        \graphtele_{,2}\st \centertikz{
            \node (in) at (0,0) {$\rho^{(\outcome,\postoutcome)}$};
            \node (out) at (3,2) {$F^{(\outcome,\postoutcome)}$};
            \node[psnode] (ps) at (1,2) {$\telepovm_{AB}$};
            \node[prenode] (pre) at (2,0) {$\telestate_{AB}$};
            \draw[qleg] (in) -- node[right] {\small$A$} (ps);
            \draw[qleg] (pre) -- node[right] {\small$B$} (ps);
            \draw[qleg] (pre) -- node[right] {\small$C$} (out);
        }
    \end{align*}
    Thanks to the Hilbert space identities of \cref{eq:hilbert space identities_v3}, we have $\hilmaparg{\edgename_0} = \hilmaparg A = \hilmaparg C$.
    We can thus rewrite \cref{eq:proof graphfamily 1} as
    \begin{align}
       \prob(\outcome,\postoutcome)_{\graphtele_{,1}} &= \Tr_A[F^{(\outcome,\postoutcome)}_A \rho^{(\outcome,\postoutcome)}_A] \nonumber\\
        &= \frac{1}{\teleprob} \Tr_{AB}\left[(\telepovm_{AB} \otimes F^{(\outcome,\postoutcome)}_C)(\rho^{(\outcome,\postoutcome)}_A \otimes \telestate_{BC})\right] \nonumber\\
        &= \frac{1}{\teleprob} \prob(\outcome,\postoutcome,\postoutcome_0 = \ok)_{\graphtele_{,2}},
    \end{align}
    where $\prob(\outcome,\postoutcome,\postoutcome_0)_{\graphtele_{,2}}$ denotes the probabilities associated to the causal model $\cm_{\graphtele_{,2}}$ on the causal graph underlying the teleportation graph $\graphtele_{,2}$.
    In the last equation, we used that $\cm_{\graphtele_{,1}}$ and $\cm_{\graphtele_{,2}}$ can be chosen to have the same associated $\rho^{(\outcome,\postoutcome)}$ and $F^{(\outcome,\postoutcome)}$ (to see this, one may for simplicity restrict all the post-selected teleportation protocols to be implemented as in \cref{def:bell tele} --- if we establish \cref{lemma: acyclic_prob_same_v3} in this case, the general result follows from \cref{corollary:probs indep of tele implementation v3}).
    We can now relate the success probability $\successprob^{(i)}$ of the causal model $\cm_{\graphtele_{,i}}$ on the teleportation graph $\graphtele_{,i}$ for $i\in\{1,2\}$.
    By \cref{def:success_prob_causal_model}, we have
    \begin{align}
        \successprob^{(1)}
        &= \sum_\outcome \prob(\outcome,\postoutcome=\ok)_{\graphtele_{,1}} \\
        &= \frac{1}{\teleprob} \sum_\outcome \prob(\outcome,\postoutcome=\ok,\postoutcome_0=\ok)_{\graphtele_{,2}} \\
        &= \frac{1}{\teleprob} \successprob^{(2)},
    \end{align}
    where the sum $\sum_\outcome$ runs over all $\outcome = \{\outcome_\vertname\in\outcomemaparg\vertname\}_{\vertname\in\overtset}$, and where we denote with $\postoutcome=\ok$ the event $\{\postoutcome_{i} = \ok\}_{\postvertname_i\in \psvertset^1}$.
    We see that the two success probabilities differ by a multiplicative constant (the success probability $\teleprob$ of the post-selected teleportation protocol defined by $\cm_{\graphtele_{,2}}$ for the split edge $\edgename_0$), and hence, $\successprob^{(1)}$ is zero if and only if $\successprob^{(2)}$ is zero.
    Hence, the probabilities $\prob_{\graphtele_{,1}}$ are defined if and only if the probabilities $\prob_{\graphtele_{,2}}$ are defined.
    Furthermore, if they are defined, we can relate them as follows (again appealing to \cref{def:success_prob_causal_model}):
    \begin{align}
        \prob(\outcome)_{\graphtele_{,1}} 
        &= \frac{\prob(\outcome,\postoutcome=\ok)_{\graphtele_{,1}}}{\successprob^{(1)}}
        = \frac{\teleprob}{\teleprob} \frac{\prob(\outcome,\postoutcome=\ok,\postoutcome_0=\ok)_{\graphtele_{,2}}}{\successprob^{(2)}} 
        = \prob(\outcome)_{\graphtele_{,2}},
    \end{align}
    which completes the proof in this case.

    \noindent\underline{Proof in the general case.}
    Consider two general elements of $\graphfamily\graphname$, $\graphtele_{,1}, \graphtele_{,2}\in\graphfamily\graphname$.
    We can prove the general statement by first noticing that the repeated application of the above argument can be used to prove that the probabilities of both $\graphtele_{,1}$ and $\graphtele_{,2}$ are equivalent to those of $\graphtele_{,0}$, where the latter corresponds to the graph where all the edges have been split (i.e., $\splitedges{\graphtele_{,0}} = \edgeset$, where $\graphname = (\vertset,\edgeset)$).
    The result follows by transitivity.
\end{proof}

\selfcycleprob*
\begin{proof}
    We start by considering a teleportation graph $\graphtele\in \graphfamily\graphname$ associated with our given graph $\graphname$, and  the causal model $\cmtele$ on the teleportation graph $\graphtele$, associated with a set $\psvertset:=\{\postvertname_i\}_{i=1}^k$ of post-selection vertices.
    
    The associated probability distribution $\prob(\outcome,\{\postoutcome_{i}=\ok\}_{\postvertname_i\in\psvertset})_{\graphtele}$ is then given by \cref{def: acyclic probability} since $\cmtele$ is an acyclic causal model by construction (as $\graphtele$ is a directed acyclic graph, \cref{lemma: teleportation_graph_acyclic}). 
    \begin{align}
    \label{eq:proof cycle probs 1}
       \prob(\outcome,\{\postoutcome_{i}=\ok\}_{\postvertname_i\in\psvertset})_{\graphtele}
        = \Tr\left[ \bigotimes_{i=1}^k \telepovm^{(i)} \, \etot_\outcome\!\left(\bigotimes_{i=1}^k \telestate^{(i)}\right)\right],
    \end{align}
    where the state $\bigotimes_i \telestate^{(i)}$ acts on $\bigotimes_i \hilmaparg{\edgearg{\prevertname_i}{\postvertname_i}} \otimes \hilmaparg{\edgename_i'}$, and consists of the tensor product of the pre-selection states of each post-selected teleportation protocol.
    The POVM element $\bigotimes \telepovm^{(i)}$ acts on $\bigotimes_i \hilmaparg{\edgearg{\prevertname_i}{\postvertname_i}} \otimes \hilmaparg{\edgename_i}$ and consists of the tensor product of the post-selection POVM elements of each post-selected teleportation protocol.
    Recall that the map $\etot_\outcome$ takes inputs in $\linops(\bigotimes_i \hilmaparg{\edgename_i'})$ and outputs in $\linops(\bigotimes_i \hilmaparg{\edgename_i})$.
    Let us define the following Hilbert spaces:
    \begin{align}
        \hilmaparg A = \bigotimes_{i=1}^k \hilmaparg{\edgename_i}, \qquad
        \hilmaparg B = \bigotimes_{i=1}^k \hilmaparg{\edgearg{\prevertname_i}{\postvertname_i}}, \qquad
        \hilmaparg C = \bigotimes_{i=1}^k \hilmaparg{\edgename_i'}.
    \end{align}
    By construction (\cref{def:causal model of graphfamily_v3}), we have that $\hilmaparg{\edgename_i} = \hilmaparg{\edgename_i'}$, so that we also have $\hilmaparg A = \hilmaparg C$.
    We denote $\telepovm_{AB} = \bigotimes_{i=1}^k \telepovm^{(i)}$ and $\telestate_{BC} = \bigotimes_{i=1}^k \telestate^{(i)}$, and we denote the map $\etot_\outcome$ as $(\etot_\outcome)_{A|C}$.
    The probabilities of \cref{eq:proof cycle probs 1} can be rewritten as
    \begin{align}
    \label{eq:proof cycle probs 2}
     \prob(\outcome,\{\postoutcome_{i}=\ok\}_{\postvertname_i\in\psvertset})_{\graphtele}
        =
        \Tr_{AB}[\telepovm_{AB} (\etot_\outcome)_{A|C}(\telestate_{BC})].
    \end{align}
    It is easy to prove that $(\telepovm_{AB},\telestate_{BC})$ implements a post-selected teleportation protocol.
    Indeed, this follows from the linearity of the post-selected teleportation condition \cref{eq:ps teleportation condition}, which allows to successfully teleport part of a multipartite state (see also \cref{lemma:teleprob indep input}), together with the fact that $(\telepovm_{AB},\telestate_{BC})$ corresponds to $k$ parallel independent post-selected teleportation protocols.
    Furthermore, it is easy to see that the teleportation success probability of $(\telepovm_{AB},\telestate_{BC})$ is the product of the individual teleportation probabilities.
    Using \cref{lemma:cyclic indep of tele implementation}, we see that the probabilities of \cref{eq:proof cycle probs 2} can be rewritten as
    \begin{align}
    \label{eq: self_cycle_proof}
      \prob(\outcome,\{\postoutcome_{i}=\ok\}_{\postvertname_i\in\psvertset})_{\graphtele}
        =
        \left(\prod_{i=1}^k\teleprob^{(i)}\right) \selfcycle(\etot_\outcome).
    \end{align}
 Finally, by recalling \cref{def:success_prob_causal_model} of the post-selection success probability $\successprob$ and \cref{def: probability distribution v3} of the probability $\prob(\outcome)_{\graphname}$ of the observed outcomes of the original causal model $\cm_{\graphname}$ in terms of the probability $ \prob(\outcome,\{\postoutcome_{i}=\ok\}_{\postvertname_i\in\psvertset})_{\graphtele}$ of the teleportation causal model $\cmtele$ (on the graph $\graphtele$), we see that \cref{eq: self_cycle_proof} immediately implies the required \cref{eq:cyclic probs_v2},
\end{proof}
\section{Proofs of~\cref{sec: introducing pseparation}}
\label{app:pseparation}
\pseptheorem*
\begin{proof}
\noindent\textbf{(Soundness)}
From \cref{def: p-separation}, we have $(V_1\perp^p V_2|V_3)_{\graphname} \Leftrightarrow \exists  \graphtele\in \graphfamily\graphname$, such that $(V_1\perp^d V_2|V_3\cup\psvertset)_{\graphtele}$, where $\psvertset$ is the set of all post-selection vertices in the chosen teleportation graph $\graphtele$. From \cref{theorem: dsep theorem}, it then follows that the conditional independence $(X_1\indep X_2|X_3\cup\psvertset)_{\probacyc_{\graphtele}}$ holds since $\graphtele$ is a directed acyclic graph (\cref{lemma: teleportation_graph_acyclic}). Finally, by \cref{def: probability distribution v3}, the probability distribution associated with a causal model on the possibly cyclic causal graph $\graphname$ is given by the corresponding probability computed in any representative teleportation graph $\graphtele\in \graphfamily\graphname$ with an additional conditioning on the associated post-selection vertices $\psvertset$. Recall from \cref{lemma: acyclic_prob_same_v3} that this distribution is the same, independently of which representative $\graphtele\in \graphfamily\graphname$ is chosen.  Therefore, using \cref{def: probability distribution v3}, we know that $(X_1\indep X_2|X_3\cup\psvertset)_{\probacyc_{\graphtele}}$ is equivalent to $(X_1\indep X_2|X_3)_{\prob_{\graphname}}$ by \cref{def: p-separation}, which completes the proof.\\

\noindent\textbf{(Completeness)}
Suppose we have a $p$-connection $ (V_1\not\perp^p V_2|V_3)_{\graphname}$ in the given directed graph $\graphname$ according to \cref{def: p-separation}. Then by \cref{lemma: pseparation between families} this implies $p$-connection $ (V_1\not\perp^{p,c} V_2|V_3)_{\graphname}$ according to the alternative definition, \cref{def: cl_p-separation}, of this notion introduced in \cite{Sister_paper}. By the completeness result for $\perp^{p,c}$ proven in \cite{Sister_paper} (Theorem 20)\footnote{We note that the notion of \cref{def: cl_p-separation} is referred to as $\perp^p$ in \cite{Sister_paper}, but we use $\perp^{p,c}$ here to distinguish it from the  \cref{def: p-separation} introduced here.}, it follows that there exists a functional causal model (\cref{def:functional_CM}) $\fcm_\graphname$ on $\graphname$ where the conditional dependence $(X_1\not\indep X_2|X_3)_{\probf_{\graphname}}$ holds. By \cref{lem: equivalence acyclic probabilities fcm and cm(fcm)}, $\fcm_\graphname$ induces a causal model $\cm(\fcm)_{\graphname}$ on $\graphname$ with equivalent probabilities, and hence the same conditional dependence holds in this causal model of $\graphname$, $(X_1\not\indep X_2|X_3)_{\prob_{\graphname}}$ as required. This completes the proof. 

\end{proof}
\end{document}